\documentclass[final,3p,times]{elsarticle}

\usepackage[]{todo}

\usepackage{latexsym}
\usepackage{amsthm}
\usepackage[cmex10]{amsmath}
\usepackage{graphicx}
\usepackage{float}
\usepackage{subfig}
\usepackage[]{placeins}
\usepackage[]{algorithm}
\usepackage{algpseudocode}
\usepackage{multirow}
\usepackage[subnum]{cases}
\usepackage{tabularx}

\newcommand\noi{\noindent}
\newcommand{\E}{\mathbb{E}}

\renewcommand{\algorithmicforall}{\textbf{for each}}

\makeatletter
\renewcommand\appendix{\par
  \setcounter{section}{0}%
  \setcounter{subsection}{0}%
  \setcounter{equation}{0}%
  \setcounter{table}{0}%------------ << add
  \setcounter{figure}{0}%----------- << add
  \gdef\theequation{\@Alph\c@section.\arabic{equation}}%
  \gdef\thefigure{\@Alph\c@section.\arabic{figure}}%
  \gdef\thetable{\@Alph\c@section.\arabic{table}}%
  \gdef\thesection{\Alph{section}}%
  \@addtoreset{equation}{section}%
  \@addtoreset{table}{section}%----- << add
  \@addtoreset{figure}{section}%---- << add
}
\makeatother

\begin{document}
\begin{frontmatter}
\title{Distributed mining of time--faded heavy hitters}
\cortext[cor1]{Corresponding author.}
\author [] {Marco~Pulimeno}
\ead{marco.pulimeno@unisalento.it}
\author [] {Italo Epicoco}
\ead{italo.epicoco@unisalento.it}
\author [] {Massimo~Cafaro\corref{cor1}}
\ead{massimo.cafaro@unisalento.it}
\address {Dept. of Engineering for Innovation, University of Salento, Via per Monteroni, 73100 Lecce, Italy}

\begin{abstract} We present \textsc{P2PTFHH} (Peer--to--Peer Time--Faded Heavy Hitters) which, to the best of our knowledge, is the first distributed algorithm for mining time--faded heavy hitters on unstructured P2P networks. \textsc{P2PTFHH} is based on the  \textsc{FDCMSS} (Forward Decay Count--Min Space-Saving) sequential algorithm, and efficiently exploits an averaging gossip protocol, by merging in each interaction the involved peers' underlying data structures. We formally prove the convergence and correctness properties of our distributed algorithm and show that it is fast and simple to implement. Extensive experimental results confirm that \textsc{P2PTFHH} retains the extreme accuracy and error bound provided by \textsc{FDCMSS} whilst showing excellent scalability. Our contributions are three-fold: (i) we prove that the averaging gossip protocol can be used jointly with our augmented sketch data structure for mining time--faded heavy hitters; (ii) we prove the error bounds on frequency estimation; (iii) we experimentally prove that \textsc{P2PTFHH} is extremely accurate and fast, allowing near real time processing of large datasets.
\end{abstract}

\begin{keyword}
% keywords here, in the form: keyword \sep keyword
peer--to--peer \sep averaging protocol \sep heavy hitters \sep time--fading model \sep sketches.
% PACS codes here, in the form: \PACS code \sep code
%\PACS
\end{keyword}

\theoremstyle{plain}% default
\newtheorem{thm}{Theorem}
\newtheorem{lem}[thm]{Lemma}
\newtheorem{prob}{Problem}
\newdefinition{rmk}{Remark}
\newproof{pf}{Proof}
 \newproof{pot}{Proof of Theorem \ref{thm2}}
\newtheorem{prop}[thm]{Proposition}
\newtheorem*{cor}{Corollary}
\newtheorem{corollary}[thm]{Corollary}
\newtheorem{observation}[thm]{Observation}
\newdefinition{defn}{Definition}
\newtheorem{conj}{Conjecture}
\newtheorem{exmp}{Example}
\theoremstyle{remark}
\newtheorem*{rem}{Remark}
\newtheorem*{note}{Note}
\newtheorem{case}{Case}
\newtheorem{claim}[thm]{Claim}
\newtheorem{fact}[thm]{Fact}
\newtheorem{assumption}[thm]{Assumption}

%\tableofcontents
\end{frontmatter}

%\tableofcontents

\section{Introduction}
\label{intro}

Distributed algorithms supporting large scale decentralized systems are being studied and developed in order to cope with complex infrastructures (peer-to-peer and sensor networks, IoT devices) and to face the fast paced growing of both the volume and dimension of data, coupled with the increasing information flow deriving by novel applications in disparate fields such as data mining, machine learning, artificial intelligence, optimization and control, and social networking.

Truly distributed solutions (i.e., we exclude here the class of \textit{centralized algorithms} in which an agent acts as a \textit{coordinator}) can be shown to provide the highly desirable properties of \textit{resilience} and \textit{scalability}, by empowering each individual agent (also known as peer) with information sensing and decision making. However, exploiting the communication network underlying the agents to design a protocol in which the peers interact and collaborate towards a common goal can be quite challenging. 

One way to cope with this issue is to impose a well-known structured topology (e.g., a spanning-tree), whose properties can then be leveraged by the distributed protocol. A different approach, which shall be used in this paper, is the use of a gossip--based protocol \cite{Demers:1987}. In particular, we are interested in \textit{distributed averaging}, which is a consensus algorithm serving as a foundational tool for the information dissemination of our distributed algorithm. 

The task of distributed averaging is to drive the states of the nodes of a network towards an agreed common value with regard to the initial values held by the agents; this value is the average of the agents' initial values. In such protocol, the agents share their status information only with a few other connected agents that are called \textit{neighbors}.

A synchronous gossip--based protocol consists of a sequence of rounds in which each peer randomly selects one or more peers, exchanges its local state information with the selected peers and updates its local state by using the received information. Therefore, gossip can be thought as a simple way of self-organizing the peer interactions in consensus seeking.

In this paper, we present our \textsc{P2PTFHH} (Peer--to--Peer Time--Faded Heavy Hitters) distributed algorithm. Building on distributed averaging, we design a gossip--based version of our sequential \textsc{FDCMSS} (Forward Decay Count-Min Space-Saving) algorithm for distributed mining of time--faded heavy hitters on unstructured P2P networks \cite{Cafaro-Pulimeno-Epicoco-Aloisio}. To the best of our knowledge, this is the first distributed protocol designed specifically for this task. The problem of mining heavy hitters (also known as frequent items) is considered a fundamental data mining task and has been extensively studied. Among the many possible applications, we recall here its wide applicability in the context of sensor data mining, for business decision support, analysis of web query logs, network measurement, monitoring and traffic analysis. In particular, we are interested to the problem of distributed mining of time--faded heavy hitters, a model that puts heavier weights on recent batches of streaming  data than older batches \cite{exp-decay}.

Our contributions are three-fold: (i) we formally prove that the averaging gossip protocol can be used jointly with our augmented sketch data structure for mining time--faded heavy hitters; (ii) we prove the error bounds on frequency estimation; (iii) we experimentally prove that \textsc{P2PTFHH} is extremely accurate and fast, allowing near real time processing of large datasets.

This paper is organized as follows. We recall in Section \ref{related-work} relevant related work and in Section \ref{definitions} preliminary definitions and concepts that shall be used in the rest of the manuscript. We present in Section \ref{alg} our \textsc{P2PTFHH} (Peer--to--Peer Time--Faded Heavy Hitters) algorithm and formally prove in Section \ref{correctness} its correctness. Next, we provide extensive experimental results in Section \ref{results}. Finally, we draw our conclusions in Section \ref{conclusions}.

\section{Related Work}
\label{related-work}

Misra and Gries \cite{Misra82} designed in 1982 the first algorithm for mining frequent items. However, their solution was not proposed as a streaming algorithm, even though it works in a streaming context. The \emph{Lossy Counting} and \emph{Sticky Sampling} algorithms by Manku et al. \cite{Manku02approximatefrequency}, published in 2002, were considered the first streaming algorithms until 2003, the year in which the Misra and Gries algorithm was independently rediscovered and improved, with regard to its running time, by Demaine et al. (the so-called \textit{Frequent} algorithm) \cite{DemaineLM02} and Karp et al. \cite{Karp}. The key idea to decrease the per item processing time was the use of a better data structure, based on a hash table. A few years later, Metwally et al. presented \emph{Space-Saving}, a novel algorithm providing a significant improvement with regard to the output's accuracy, whilst retaining the same running time and space required to process the input stream. 

All of the previous algorithms for detecting heavy hitters are commonly known as \emph{counter--based}, since they use a set of counters to keep track of the frequent items. Other algorithms are instead known as \emph{sketch--based}, owing to their use of a sketch data structure to monitor the input stream. A sketch is usually a bi-dimensional array data structure containing a counter in each cell. Hash functions map the stream items to corresponding cells in the sketch. In this class of algorithms, we recall here \emph{CountSketch} by Charikar et al. \cite{Charikar}, \emph{Group Test} \cite{Cormode-grouptest} and \emph{Count-Min} \cite{Cormode05} by Cormode and Muthukrishnan and \emph{hCount} \cite{Jin03} by Jin et al.

Parallel algorithms designed for fast mining of heavy hitters span both the message--passing and shared--memory paradigms. Cafaro et al. \cite{Cafaro-Pulimeno-Tempesta} \cite{cafaro-tempesta} \cite{Cafaro-Pulimeno} designed parallel versions of the Frequent and Space-Saving algorithms for message--passing architectures. Regarding shared-memory, we recall here parallel versions of Lossy Counting and Frequent by Zhang et al. \cite{Zhang2012} \cite{Zhang2013}, parallel versions of Space-Saving by Dat et al. \cite{Das2009}, Roy et al. \cite{Roy2012}, and Cafaro et al \cite{CPE:CPE4160}. Parallel algorithms for fundamental frequency--based aggregates, including heavy hitters were designed by Tangwongsan et al. \cite{Tangwongsan2014}. Among the accelerator based algorithms, exploiting GPU (Graphics Processing Unit) and the Intel Phi, we find Govindaraju et al. \cite{Govindaraju2005}, Erra and Frola \cite{Erra2012} and Cafaro et al. \cite{HPCS2017} \cite{CPE:CPE4160}. 

Mining Correlated Heavy Hitters (CHH) has been recently proposed by Lahiri et al. \cite{Lahiri2016}; common applications requiring accurate mining of CHHs are related to network monitoring and management, and to anomaly and intrusion detection. For instance, taking into account the stream of pairs (source address, destination address) consisting of IP packets traversing a router, the data mining task is the identification of the nodes accounting for the majority of the traffic passing through that router (these are the frequent items over a single dimension); however, given a frequent source, it is desirable and important as well to simultaneously discover the identity of the destinations receiving the majority of connections by the same source. After mining the first dimension to detect the most important sources, the second dimension is mined to detect the frequent destinations in the context of each identified source, i.e., the stream's correlated heavy hitters. Epicoco et al. \cite{Epicoco:2018:FAM:3182040.3182103} proposed a fast and more accurate algorithm for mining CHHs; a parallel, message-passing based version of this algorithm appears in Pulimeno et al. \cite{Pulimeno:WPDM2018}.

The above algorithms, when processing an input stream, do not discount the effect of old data. In practice, all of the items are given equal weight. However, this is not appropriate for those applications based on the underlying assumption that recent data is more useful and valuable than older, stale data. One way to handle such as a situation is using the so-called \emph{sliding window} model \cite{Datar} \cite{TCS-002}. The key idea is the use of a temporal window to capture fresh, recent items. this window periodically slides forward, allowing detection of only those frequent items falling in the window.

The time--fading model \cite{exp-decay} puts heavier weights on recent batches of streaming data than older batches. This is achieved by \emph{fading} the frequency count of older items: each item in the stream has an associated timestamp that shall be used to determine its weight, by means of a decaying factor $0 < \lambda < 1$. Therefore, instead of an item's frequency count, algorithms compute its \textit{decayed count} (also known as \textit{decayed frequency}) through a decay function assigning greater weight to more recent items. The older an item, the lower its decayed count is: in the case of exponential decay, the weight of an item occurred $n$ time units in the past, is $e^{-\lambda n}$, which is an exponentially decreasing quantity.

Noteworthy algorithms for mining time--faded heavy hitters include $\lambda$-HCount \cite{Chen-Mei} by Chen and Mei, FSSQ (Filtered Space-Saving with Quasi--heap) \cite{FSSQ} by Wu et al. and the \textsc{FDCMSS} algorithm \cite{Cafaro-Pulimeno-Epicoco-Aloisio} \cite{FDCMSSvsFSSQ} by Cafaro et al. These algorithms are sketch--based variants of Count-Min. Whilst $\lambda$-HCount and FSSQ rely on a backward decay exponential function, \textsc{FDCMSS} provides support for both backward and forward decay functions \cite{forward-decay}. For instance, by using a polynomial forward function, \textsc{FDCMSS} provides more flexibility with regard to time fading, since the time fades more slowly than by using an exponential function. Moreover, the use of forward decay allows \textsc{FDCMSS} dealing with out of order arrival of stream items easily. On the contrary, neither $\lambda$-HCount nor FSSQ can handle out of order items. Both $\lambda$-HCount and FSSQ require a dedicated data structure to keep track of frequent item candidates: $\lambda$-HCount uses a doubly linked list indexed by an hash function, whilst FSSQ uses a so--called Quasi-Heap, an heap in which heapify operations are delayed until necessary. Our \textsc{FDCMSS} algorithm does not require additional space beyond its data structure, a sketch in which each cell hosts a Space-Saving stream summary holding exactly two counters. We have shown that, fixing the same amount of space for all of the algorithms, \textsc{FDCMSS} outperforms $\lambda$-HCount and FSSQ with regard to the error committed. Moreover, \textsc{FDCMSS} provides the best error bound. Regarding the \textit{recall} property (i.e., the ability of outputting all of the true frequent items), both $\lambda$-HCount and FSSQ guarantee 100\% recall, whilst \textsc{FDCMSS} provides a probabilistic recall guarantee. Finally, with regard to the running time, \textsc{FDCMSS} is the fastest owing to its data structures and update strategy. A parallel message-passing based algorithm has been recently proposed in \cite{CAFARO2018115}.

Many algorithms for distributed mining of heavy hitters have been designed, e.g., \cite{Cao:2004}, \cite{Zhao:2006}, \cite{Keralapura:2006}, \cite{recent-freq-items}, \cite{Venkataraman}. Among the gossip--based algorithms for mining heavy hitters on unstructured P2P networks, it is worth recalling here \cite{Sacha}, \cite{CEM20131544}, and \cite{LAHIRI20101241}. 

The algorithms by Sacha and Montresor \cite{Sacha} and by Aem and Azkasap \cite{CEM20131544} are based on the same ideas. Indeed, the dataset to be mined is partitioned among the peers, which periodically exchange their local state consisting of the current subset of the whole dataset. The key assumption is that a peer is allowed to store the whole dataset (even though this requires obviously a huge amount of space), which is indeed obtained by means of the distributed averaging gossip interactions. In order to estimate $N$, the number of peers, the distributed averaging approach is used again as follows: one of the peers starts with a value equal to one and all of the others with a value equal to zero. It can be shown that the protocol converges to $1/N$, so that $N$ can be easily estimated. In particular, the convergence of the averaging gossip protocol has been investigated in \cite{Jelasity2005}; here, the authors proved that the variance around the mean value being computed is reduced in each round by a specific convergence factor.

Exchanging the full local state, besides requiring a huge amount of space, also increases to communication complexity. Therefore, in \cite{Sacha} the authors suggest to exchange only the top-$k$ most frequent items where $k$ is a user's defined parameter. The protocol's termination condition requires for each peer that the subset consisting of the top-$k$ items does not change for a specified number of consecutive rounds. In \cite{CEM20131544}, to cope with the communication complexity, the authors uses an additional data structure, which is an hash table storing all of the items (these items are never deleted, at the cost of increased space used) from which the algorithm randomly selects a specified number of items corresponding to a predefined message size. The termination condition requires two user's defined parameters: $\epsilon$ and $convLimit$. The algorithm terminates when for each item, the absolute difference between the estimated frequency in the current and previous rounds is less than or equal to $\epsilon$ for at least $convLimit$ consecutive rounds.

Both \cite{Sacha} and \cite{CEM20131544} require space complexity linear in the length $n$ of the dataset; this allows solving the \textit{exact} problem rather than the approximate problem, but at the expense of space.
 
In Lahiri and Tirthapura \cite{LAHIRI20101241}, the authors design an algorithm based on random sampling of the items and the averaging gossip protocol. Each item is assigned a random weight in the interval (0, 1) and a peer maintains and exchanges in each round the $t$ items whose weight is the lowest, with $t = \frac{128}{\psi^2} \ln \frac{3}{\delta}$; $\psi$ and $\delta$ denote respectively a threshold and a probability of failure. The proposed approach can only detect frequent items, since the algorithm does not provide frequency estimation. In particular, the output is a list of items that with high probability (defined by $\delta$) contains the frequent items (with regard to the $\psi$ threshold). Regarding the space used, for each of the $t$ items the algorithm stores a tuple consisting of four fields: the peer identifier, the item index in the peer's local dataset, the item value and its random weight.   

None of the previous gossip-based protocols has been designed to mine time--faded heavy hitters on unstructured P2P networks. To the best of our knowledge, \textsc{P2PTFHH} is the first distributed protocol designed specifically for this task.

\section{Preliminary definitions}
\label{definitions}

In this Section we introduce preliminary definitions and the notation used throughout the paper. We deal with an input data stream $\sigma$ consisting of a sequence of $n$ items drawn from a universe $\mathcal{U}$; without loss of generality, let $m$ be the number of distinct items in $\sigma$ i.e., let $\mathcal{U}=\{1,2,\ldots,m\}$, which we shall also denote by $[m]$. Let $f_i$ be the frequency of the item $i \in \mathcal{U}$ (i.e., its number of occurrences in $\sigma$), and denote the frequency vector by $\textbf{f} = (f_1,\ldots,f_m)$. Moreover, let $0 < \phi < 1$ be a support threshold, $0 < \epsilon < 1$ a tolerance such that $\epsilon < \phi$ and denote the 1-norm of $\textbf{f}$ (which represents the total number of occurrences of all of the stream items) by $||\textbf{f}||_1$. 

In this paper, we are concerned with the problem of detecting heavy hitters in a stream which is distributed among $p$ peers with the additional constraint that recent items are more relevant with regard to older items. In the \emph{time--faded} model, recent items are weighted more than former items while in the \emph{sliding window} model only the items' occurrences which appear in a window sliding over time are considered. Our algorithm, \textsc{P2PTFHH}, works in the former model.

The time--fading model \cite{recent-freq-items} \cite{exp-decay} \cite{Chen-Mei} \cite{FSSQ} uses a \emph{fading} frequency count of older items to represent the freshness of recent items. This is achieved by computing the item's \textit{decayed frequency} through the use of a decay function that assign greater weight to more recent occurrences of an item than to older ones: the older an occurrences is, the lower its decayed weight.

\begin{defn}
\label{decay-function}
Let $w(t_i,t)$ be a decayed function which computes the decayed weight at time $t$ for an occurrence of the item $i$ arrived at time $t_i$. A decayed function must satisfy the following properties: 
\begin{enumerate}
\item $w(t_i,t) = 1$ when $t_i = t$ and $0 \leq w(t_i,t) \leq 1$ for all $t > t_i$;
\item $w$ is a monotone non-increasing function as time $t$ increases, i.e., $t' \geq t \implies w(t_i, t') \leq w(t_i, t)$.
\end{enumerate}
\end{defn}

Related work has mostly exploited \textit{backward decay} functions, in which the weight of an item is a function of its age, $a$, where the age at time $t > t_i$ is simply $a = t-t_i$. In this case, $w(t_i, t)$ is given by $w(t_i, t) = \frac{h(t-t_i)}{h(t-t)}=\frac{h(t-t_i)}{h(0)}$, where $h$ is a positive monotone non-increasing function.

The term backward decay stems from the aim of measuring from the current time back to the item's timestamp. Prior algorithms and applications have been using backward exponential decay functions such as $h(a) = e^{-\lambda a}$, with $\lambda > 0$ being the decaying factor. 

In our algorithm, we use instead a \textit{forward decay} function, defined as follows (see \cite{forward-decay} for a detailed description of the forward decay approach). Under forward decay, the weight of an item is computed on the amount of time between the arrival of an item and a fixed point $L$, called the \textit{landmark} time, which, by convention, is some time earlier than the timestamps of all of the items. The idea is to look forward in time from the landmark to see an item, instead of looking backward from the current time.

\begin{defn}
	Given a positive monotone non-decreasing function $g$, and a landmark time $L$, the forward decayed weight of an item $i$ with arrival time $t_i > L$ measured at time $t \geq t_i$ is given by $w(t_i, t) = \frac{g(t_i-L)}{g(t-L)}$.
\end{defn}

The denominator is used to normalize the decayed weight so that $w(t_i, t)$ is always less than or equal to 1 as requested by Definition~\ref{decay-function}.

\begin{defn}
\label{item-decayed-count}
The \textit{decayed frequency} of an item $v$ in the input stream $\sigma$, computed at time $t$, is given by the sum of the decayed weights of all the occurrences of $v$ in $\sigma$: $f_v(t) = \sum_{v_i = v} w(t_i,t)$.
\end{defn}

\begin{defn}
\label{decayed-count}
The \textit{decayed count} at time $t$, $C(t)$, of a stream $\sigma$ of $n$ items is the sum of the decayed weights of all the items occurring in the stream: $C(t)=\sum_{i=1}^n w(t_i, t)$.
\end{defn}

The Approximate Time--Faded Heavy Hitters (ATFHH) problem is formally stated as follows.

\begin{prob} 
	\label{prob1}
	Approximate Time--Faded Heavy Hitters. Given a stream $\sigma$ of items with an associated timestamp, a threshold $0 < \phi < 1$ and a tolerance $0 < \epsilon < 1$ such that $\epsilon < \phi$, and letting $g$ be a decaying function used to determine the decayed frequencies and $t$ be the query time, return the set of items $H$, so that:

\begin{itemize}
\item $H$ contains all of the items $v$ with decayed frequency at time $t$ $f_v(t) > \phi C(t)$ (decayed frequent items);
 
\item $H$ does not contain any item $v$ such that $f_v(t) \leq (\phi-\epsilon) C(t)$. 
\end{itemize}

\end{prob}

In the following, when clear from the context, the query time shall be considered an implicit parameter, and we shall write $f_v$ and $C$ in place of $f_v(t)$ and $C(t)$. Our algorithm uses an averaging gossip protocol for unstructured P2P networks in which the peers exchange their internal state consisting of a Count--Min sketch data structure augmented by a Space-Saving summary associated to each sketch cell. In the following, we recall the main properties of the Count--Min and Space-Saving algorithms for ordinary, non decaying frequencies and the averaging gossip protocol.

Count--Min is based on a sketch of $d \times w$ cells, where each item is mapped to by means of $d$ different pairwise independent hash functions. Count--Min solves the \textit{frequency estimation} problem for arbitrary items with an estimation error $\epsilon \leq e/w$ and with probability $1 - e^{-d}$. The algorithm may also be extended to solve the \textit{approximate frequent items} problem as well, by using an additional heap data structure which is updated each time a cell is updated. Since in Count--Min the frequencies stored in the cells overestimate the true frequencies, a point query for an arbitrary item simply inspects all of the $d$ cells in which the query item is mapped to by the corresponding hash functions and returns the minimum of those $d$ counters. 

Space-Saving is a counter-based algorithm solving the heavy hitters problem. It is based on a stream summary data structure consisting of a given number of counters $k \ll n$, $n$ being the length of the stream. Each counter monitors an item in the stream and tracks its frequency. A substitution strategy is used when the algorithm processes an item not already monitored and all of the counters are occupied.

Let $\sigma$ be the input stream and denote by $\mathcal{S}$ the summary data structure of $k$ counters used by the Space-Saving algorithm. Moreover, denote by $\left|\mathcal{S}\right|$ the sum of the counters in $\mathcal{S}$, by $f_v$ the exact frequency of an item $v$ and by $\hat{f}_v$ its estimated frequency, let $\hat{f}^{min}$ be the minimum frequency in $\mathcal{S}$. If there exist at least one counter not monitoring any item, $\hat{f}^{min}$ is zero.

Finally, denote by $\textbf{f} = (f_1,\ldots,f_m)$ the frequency vector. The following relations hold (as shown in \cite{Metwally2006}):

\begin{equation}
\label{ss1}
\left|\mathcal{S}\right| = ||\textbf{f}||_1,
\end{equation}

\begin{equation}
\label{ss5}
\hat{f}_v - f_v \leq \hat{f}^{min} \leq \left\lfloor\frac{||\textbf{f}||_1}{k}\right\rfloor, \hspace{3mm} v \in \mathcal{U}.
\end{equation}

The gossip--based protocol \cite{Demers:1987} is a synchronous distributed algorithm consisting of periodic rounds. In each of the rounds, a peer (or agent) randomly selects one or more of its neighbors, exchanges its local state with them and finally updates its local state. The information is disseminated through the network by using one of the following possible communication methods: (i) \textit{push}, (ii) \textit{pull} or (iii) \textit{push--pull}. The main difference between push and pull is that in the former a peer randomly selects the peers to whom it wants to send its local state, whilst in the latter it randomly selects the peers from whom to receive the local state. Finally, in the hybrid push--pull communication style, a peer randomly selects the peers to send to and from whom to receive the local state. In this synchronous distributed model it is assumed that updating the local state of a peer is done in constant time, i.e., with $O(1)$ worst-case time complexity; moreover, the duration of a round is such that each peer can complete a push--pull communication within the round. 

We are interested in a specific gossip--based protocol, which is called \textit{distributed averaging}, and can be considered as a consensus protocol. For the purpose of our theoretical analysis, we assume that peers and communication links do not fail, and that neither new peers can join the network nor existing peers can leave it (the so-called \textit{churning} phenomenon). Therefore, the graph $G = (V, E)$ in which $V$ is the set of peers and $E$ models the connectivity, fully describing the underlying network topology, is not time-varying. However, it is worth noting here that our algorithm also works in time-varying graphs in which the network can change owing to failures or churning and we shall show an experimental evidence of that in Section~\ref{effect-of-churn}, in which we discuss the effect of churn. 

In \textit{uniform gossiping}, a peer $i$ can communicate with a randomly selected peer $j$. Instead, in our scenario the communication among the peers is restricted to neighbor peers i.e., two peers $i$ and $j$ are allowed to communicate if and only if a one-hop communication link exists between them; we assume that communication links are bidirectional. Initially, each peer $i$ is provided with or computes a real number $v_i$; the distributed averaging problem requires designing a distributed algorithm allowing each peer computing the average $v_{{\rm avg}} = \frac{1}{p}\sum_{i=1}^p v_i$ by exchanging information only with its neighbors. Letting $v_i(r)$ be the peer $i$ estimated value of $v_{{\rm avg}}$ at round $r$, a gossip interaction between peers $i$ and $j$ updates both peers' variables so that at round $r+1$ it holds that $v_{i}(r+1)=v_j(r+1) = \frac{1}{2}(v_i(r)+v_j(r))$. Of course, for a peer $i$ which is not gossiping at round $r$ it holds that $v_i(r+1) = v_i(r)$. It can be shown that distributed averaging converges exponentially fast to the target value $v_{{\rm avg}}$. In general, a peer is allowed to gossip with at most one peer at a time. In our algorithm, we allow each peer the possibility of gossiping with a predefined number of neighbors. We call \textit{fan-out}, $fo$, of the peer $i$ the number of its neighbors with which it communicates in each round; therefore, $1 \leq fo \leq \left\vert \{j: (i, j) \in E\} \right\vert $. Therefore, we explicitly allow two or more pairs of peers gossiping at the same time, with the constraint that the pairs have no peer in common. We formalize this notion in the following definition.

\begin{defn}
	Two gossip pairs of peers $(i,j)$ and $(x,y)$
	are \textit{noninteracting} if neither $i$ nor $j$ equals either $x$ or $y$.
\end{defn}

In our algorithm multiple noninteracting  pairs of allowable gossips may occur simultaneously. Non-interactivity is required in order to preserve and guarantee correctness of the results; in the literature non-interactivity is also called \textit{atomic} push--pull communication: given two peers $i$ and $j$, if peer $i$ sends a push message to $j$, then peer $i$ can not receive in the same round any intervening push message from any other peer $k$ before receiving the pull message from $j$ corresponding to its initial push message. 

It is worth noting here that our algorithm does not require explicitly assigning identifiers to the peers, and we do so only for convenience, in order to simplify the analysis; however, we do assume that each peer can distinguish its neighbors.

\section{The algorithm}
\label{alg}

In this section, we start by recalling our sequential algorithm \textsc{FDCMSS} \cite{Cafaro-Pulimeno-Epicoco-Aloisio}. The key data structure is an augmented Count--Min sketch $\mathcal{D}$, whose dimensions $d$ (rows) and $w$ (columns) are derived by input parameters $\epsilon$, the error tolerance, and $\delta$, the probability of failure. Whilst every cell in an ordinary Count--Min sketch contains a counter used for frequency estimation, in our case a cell holds a Space-Saving summary with exactly two counters. The idea behind the augmented sketch is to monitor the time--faded items that the sketch hash functions map to the corresponding cells by an instance of Space-Saving with two counters, so that for a given cell we are able to determine a \textit{majority item candidate} with regard to the sub-stream of items falling in that cell. In fact, we proved that, with high probability, if an item is frequent, then it appears as a majority item candidate in at least one of the $d$ sketch cells in which it falls.

Indeed, by using a Space--Saving summary with two counters in each cell, and letting $C_{i,j}$ denote the total decayed count of the items falling in the cell $\mathcal{D}[i][j]$, the majority item is, if it exists, the item whose decayed frequency is greater than $\frac{C_{i,j}}{2}$. The corresponding majority item candidate in the cell is the item monitored by the Space-Saving counter whose estimated decayed frequency is maximum. 

The main idea of the distributed version of the \textsc{FDCMSS} algorithm is to let each peer process its local data stream with the \textsc{FDCMSS} algorithm. Then, the peers engage in a gossip--based distributed averaging protocol, exchanging their local state which consists of the augmented sketch data structure obtained after processing the input stream and an estimate of the number of peers in the network.

The algorithm's initialization (pseudo-code provided as Algorithm \ref{initialize}) requires as input parameters $d$ and $w$, the dimensions of the local sketch; $\mathcal{N}_l$, the local dataset to be processed, being $l$ the peer's identifier; $\epsilon$, the error tolerance; $\delta_g$, a probability of failure of the gossip protocol (i.e., related to the event that the gossip protocol fails to provide the theoretical properties that it should guarantee with high probability); $\phi$, the frequent support threshold; $p^*$, an estimate of the number $p$ of peers in the network (we only require $p^* \geq p$); $R$, the number of rounds to be performed by the distributed algorithm; and $\gamma$, the gossip protocol convergence factor, whose meaning shall be explained in the theoretical analysis section. 

The initialization procedure starts calling the \textsc{CreateSketch} function which allocates and returns an augmented sketch: for each one of the $d \times w$ cells of the sketch, \textsc{CreateSketch} allocates a Space-Saving summary with two counters $c_1$ and $c_2$, both maintaining an item and a frequency value (given a counter $c_j, j=1,2$, we denote by $c_j.i$ and $c_j.f$ respectively the counter's item and its estimated decayed frequency). Moreover, \textsc{CreateSketch} selects $d$ pairwise independent hash functions $h_1,\ldots,h_d:[m] \rightarrow [w]$, mapping $m$ distinct items into $w$ cells. It is worth noting here that each peer is required to select the same $d$ hash functions in order to guarantee the correctness of the results. 

After the main data structures are created, each peer stores a local copy of the remaining input parameters. Then, the peer whose identifier is $l=1$ sets $\tilde{q}_{0,l}$ to 1 and all of the other peers sets this value to zero. The variable $\tilde{q}_{0,l}$ represents the estimate of peer $l$ at round $0$ of the number $p$ of peers (to be obtained by using the distributed averaging protocol: indeed, upon convergence this value approaches $1/p$ with high probability). 

\begin{algorithm}
	\begin{algorithmic}[1]
		\Require $d, w$, number of rows and columns of the sketch; $\mathcal{N}_l$, local dataset to be processed; $R$, total number of rounds to be executed; $\epsilon$, parameter controlling the false positives rate; $\delta_g$, the probability of gossip failure; $\phi$, the support threshold; $p^*$, an estimate of the number of peers; $\gamma$, the convergence factor; $L$, the landmark time.  
		\Ensure Initialization of the internal peer's state.
		\caption{\textsc{P2PTFHH}: Initialization}	
		\label{initialize}
		\Procedure{Initialize}{$d$, $w$, $\mathcal{N}_l$, $R, \epsilon, \delta_g, \phi, p^*, \gamma, L$}
		\Comment{initialization of node $l$}
		\State Let $l$ be the peer's identifier
		\State $\mathcal{D}_{0,l} \leftarrow$ \Call {CreateSketch} {$w$,$d$}
		\If{$l == 1$}
			\State $\tilde{q}_{0,l} \leftarrow 1$
		\Else
			\State $\tilde{q}_{0,l} \leftarrow 0$
		\EndIf
		\State  \Call{FDCMSS-UPDATE}{ $\mathcal{D}_{0,l}$, $\mathcal{N}_l$}
		\State $state_{0,l} \leftarrow (\mathcal{D}_{0,l}, %\tilde{n}_{r,l},%
																				 \tilde{q}_{0,l})$
		\EndProcedure
	\end{algorithmic}
\end{algorithm}

Each peer processes its local data stream following the update procedure of the \textsc{FDCMSS} algorithm, shown in pseudo-code as Algorithm \ref{FDCMSS-UPDATE}. For each item $i$ with timestamp $t_i$, the non normalized forward decayed weight $x$ of the item is computed, then the $d$ cells in which the item is mapped to by the corresponding hash functions $h_j(x), j=1,\dots,d$ are updated by using the Space-Saving item update procedure.  

We shortly describe here the update procedure of the Space-Saving algorithm, \textsc{SpaceSavingUpdate} in Algorithm \ref{FDCMSS-UPDATE}: let $\mathcal{S}$ denote the Space-Saving summary with two counters corresponding to the cell to be updated. When processing an item which is already monitored by a counter, its estimated frequency is incremented by the non normalized weight $x$. When processing an item which is not already monitored by one of the available counters, there are two possibilities. If a counter is available, it will be in charge of monitoring the item, and its estimated frequency is set to the non normalized weight $x$.  Otherwise, if all of the counters are already occupied (their frequencies are different from zero), the counter storing the item with minimum frequency is incremented by the non normalized weight $x$. Then, the monitored item is evicted from the counter and replaced by the new item. This happens since an item which is not monitored can not have a frequency greater than the minimal frequency. 

\begin{algorithm}
\begin{algorithmic}[1]
\Require $\mathcal{D}_{0,l}$, sketch data structure of peer $l$; $\mathcal{N}_l$, local dataset to be processed.
\Ensure update of sketch related to items in $\mathcal{N}_l$.
\Procedure {FDCMSS-UPDATE}{$\mathcal{D}_{0,l}$, $\mathcal{N}_l$}
\Comment{let $i$ be an item, and $t_i$ its timestamp}
\ForAll{$(i, t_i) \in \mathcal{N}_l$}
	\State $x \leftarrow g(t_i-L)$
	\Comment{compute the non normalized decayed weight of item $i$}
	\For{$j=1$ to $d$}
		\State $\mathcal{S} \leftarrow \mathcal{D}_{0,l}[j][h_j(i)]$
		\State \Call{SpaceSavingUpdate}{$\mathcal{S}, i, x$}
		\Comment {update the sketch}
	\EndFor
\EndFor
\EndProcedure
\caption{FDCMSS-UPDATE}
\label{FDCMSS-UPDATE}
\end{algorithmic}
\end{algorithm}

We represent the local state of a peer $l$ at round $r$ as a tuple $state_{r,l}$ consisting of the peer's local sketch $\mathcal{D}_{r,l}$, and the estimate $\tilde{q}_{r,l}$.

\begin{algorithm}
\begin{algorithmic}[1]
\caption{\textsc{P2PTFHH: Gossip protocol}}	
\label{p2ptfhhalg}

\Procedure{GOSSIP}{}
\For{$r = 0$ to $R$}
	\State $neighbours \leftarrow$ select $fo$ random neighbours
	\ForAll {$i \in neighbours$}
		\State \Call{SEND}{$push$, $i$, $state_{r,l}$}
	\EndFor
\EndFor
\EndProcedure

\Procedure{ON\_RECEIVE}{$msg$}
\If{$msg.type == push$}
	\State $state_{r+1,l} \leftarrow$ \Call{UPDATE}{$msg.state$, $state_{r,l}$}
	\State \Call{SEND}{$pull$, $msg.sender$, $state_{r+1,l}$}
\EndIf
\If{$msg.type == pull$}
	\State $state_{r+1,l} \leftarrow msg.state$
\EndIf
\EndProcedure

\Procedure{UPDATE}{$state_i$, $state_j$}
\State $(\mathcal{D}_i, \tilde{q}_i) \leftarrow state_i$
\State $(\mathcal{D}_j, \tilde{q}_j) \leftarrow state_j$
\State $\mathcal{D} \leftarrow$ \Call{MERGE}{$\mathcal{D}_i$, $\mathcal{D}_j$}
\ForAll{summary $\mathcal{S}$ in a cell of $\mathcal{D}$}
	\ForAll{counter $c \in \mathcal{S}$}
		\State $c.f \leftarrow \frac{c.f}{2}$
	\EndFor	
\EndFor		
%\State $\tilde{n} \leftarrow \frac{\tilde{n}_i + \tilde{n}_j}{2}$  
\State $\tilde{q} \leftarrow \frac{\tilde{q}_i + \tilde{q}_j}{2}$  
\State $state \leftarrow (\mathcal{D}, \tilde{q})$
\State \Return $state$
\EndProcedure

\end{algorithmic}
\end{algorithm}

\begin{algorithm}
\begin{algorithmic}[1]
\Require $\mathcal{D}_1, \mathcal{D}_2$, sketch data structures to be merged.
\Ensure $\mathcal{G}$, merged sketch.
\caption{\textsc{Merge}}	
\label{p2ptfhhalg_merge}

\Procedure {MERGE}{$\mathcal{D}_1, \mathcal{D}_2$}
\ForAll{$\mathcal{S}^1_{ij} \in \mathcal{D}_1, \mathcal{S}^2_{ij} \in \mathcal{D}_2$}
		\State $m_1 \leftarrow $ \Call {min}{$c^1_1.f, c^1_2.f$}
		\Comment{$m_1$, the minimum of counters' frequency in $\mathcal{S}^1_{ij}$}
		\State $m_2 \leftarrow $ \Call {min}{$c^2_1.f, c^2_2.f$}
		\Comment{$m_2$, the minimum of counters' frequency in $\mathcal{S}^2_{ij}$}
\ForAll{$c_{s_1} \in \mathcal{S}_1$}
	\State $c_{s_2} \leftarrow \Call {Find}{\mathcal{S}_2, c_{s_1}.i}$
	\If{$c_{s_2}$}
		\State $c_{s_c}.f \leftarrow c_{s_1}.f + c_{s_2}.f$
		\State \Call {Delete}{$\mathcal{S}_2, c_{s_2}$}
	\Else
		\State $c_{s_c}.f \leftarrow c_{s_1}.f + m_2$
	\EndIf
	\State $c_{s_c}.i \leftarrow c_{s_1}.i$
	\State \Call {Insert}{$\mathcal{S}_C, c_{s_c}$}
\EndFor
\ForAll{$c_{s_2} \in \mathcal{S}_2$}
	\State $c_{s_c}.i \leftarrow c_{s_2}.i$
	\State $c_{s_c}.f \leftarrow c_{s_2}.f + m_1$
	\State \Call {Insert}{$\mathcal{S}_C, c_{s_c}$}
\EndFor
		\State \Call {Purge}{$\mathcal{S}_C$}
		\Comment {$\mathcal{S}_C$ now contains $2$ counters with the greatest frequencies}
		\State $\mathcal{G}[i][j] \leftarrow \mathcal{S}_C$
\EndFor
\State \Return $\mathcal{G}$
\EndProcedure

\end{algorithmic}
\end{algorithm}

\begin{algorithm}
\begin{algorithmic}[1]
\Require $t$, query time.
\Ensure $H$, the set of frequent items.
\Procedure{QUERY}{$t$}
\State $(\mathcal{\tilde{D}}_{r,l}, \tilde{q}_{r,l}) \leftarrow state_{r,l}$
\State $\epsilon^* \leftarrow p^* \times \sqrt{\frac{\gamma^{r}}{\delta_g}}$
%\State $t \leftarrow \phi \tilde{n}_{r,l} \frac{1-\epsilon^*}{1+\epsilon^*}$
\State $\tilde{p}_{r,l} \leftarrow 1/\tilde{q}_{r,l}$
\State $\tilde{C}_{r,l} \leftarrow 0$
\ForAll{summary $\mathcal{S} \in \mathcal{\tilde{D}}_{r,l}$}, 
	\ForAll{counter $c \in \mathcal{S}$}
		\State $\tilde{C}_{r,l} \leftarrow \tilde{C}_{r,l} + c.f$
	\EndFor	
\EndFor
\State $\tilde{C}_{r,l} \leftarrow \frac{\tilde{C}_{r,l}}{g(t-L)} $
\State $H \leftarrow \emptyset$
\State $\tau \leftarrow \phi \tilde{C}_{r,l} \frac{1-\epsilon^*}{1+\epsilon^*}$
\ForAll{summary $\mathcal{S} \in \mathcal{\tilde{D}}_{r,l}$}
		\State let $c_1$ and $c_2$ be the counters in $\mathcal{S}$ 
		\State $c_m \leftarrow \Call{argmax}{c_1, c_2}$
		\Comment {$c_m$ the counter with maximum decayed count}
		\If{$\frac{c_m.f}{g(t-L)}  > \tau$}	
			\State $p \leftarrow \Call{PointEstimate}{c_m.i, t}$
			\If{$p > \tau$}
				\State $H \leftarrow H \cup \{(c_m.i, p)\}$
			\EndIf
		\EndIf
\EndFor
\State \Return $H$
\EndProcedure

\Procedure {PointEstimate}{$i, t$}
\State $answer \leftarrow \infty$
\For{$k=1$ to $d$}
	\State $\mathcal{S} \leftarrow \mathcal{\tilde{D}}_{r,l}[k][h_k(i)]$
	\Comment{let $c_1$ and $c_2$ be the counters in $\mathcal{S}$}
	\If{$i == c_1.i$}
		\State $answer \leftarrow \Call{min}{answer, c_1.f}$
	\ElsIf{$i == c_2.i$}
			\State $answer \leftarrow \Call{min}{answer, c_2.f}$
		\Else
			\State $m \leftarrow \Call{min}{c_1.f, c_2.f}$
			\State $answer \leftarrow \Call{min}{answer, m}$
		\EndIf
\EndFor
\State \Return $\frac{answer}{g(t-L)}$
\EndProcedure

\caption{QUERY}
\label{query}
\end{algorithmic}
\end{algorithm}

In the \textsc{P2PTFHH} algorithm, all of the peers participate to the distributed gossip protocol for $R$ rounds (\textsc{Gossip} procedure in Algorithm \ref{p2ptfhhalg}). During each round, a peer increments $r$, the  round counter, selects $fo$ (the fan-out) neighbours uniformly at random and sends to each of them its local state in a message of type \textit{push}. Upon receiving a message, each peer executes the \textsc{on\_receive} procedure, whose pseudo-code is shown as Algorithm \ref{p2ptfhhalg}. From the message, the peer extracts the message's type, sender and state sent. A message is processed accordingly to its type as follows. A push message is handled in two steps. In the first one, the peer updates its local state by using the state received; this is done by invoking the \textsc{update} procedure that we shall describe later. In the second one, the peer sends back to the sender, in a message of type \textit{pull}, its updated local state. A pull message is handled by a peer setting its local state equal to the state received.

The \textsc{update} procedure, which is the last procedure described by Algorithm \ref{p2ptfhhalg}, works as follows: the two local sketches $\mathcal{D}_i$ and $\mathcal{D}_j$ belonging respectively to the peers $i$ and $j$ are merged by invoking the \textsc{merge} procedure, producing the new sketch $\mathcal{D}$; since we want to implement a distributed averaging protocol, we scan the counters of each stream summary $\mathcal{S}$ in the cells of $\mathcal{D}$, and for each counter $c$ we update its frequency $c.f$ dividing it by 2; finally, we compute as required by the averaging protocol the estimate $\tilde{q}$ and return the updated state just computed.

Two sketches are merged by the \textsc{merge} procedure (pseudo-code provided as Algorithm \ref{p2ptfhhalg_merge}) as follows: for every corresponding cell in the two sketches to be merged, the hosted Space--Saving summaries are merged following the steps described in \cite{Cafaro-Pulimeno-Tempesta}, i.e., building a temporary summary $\mathcal{S}_C$ consisting of all of the items monitored by both $\mathcal{S}_1$ and $\mathcal{S}_2$. To each item in $\mathcal{S}_C$ is assigned a decayed frequency computed as follows: if an item is present in both $\mathcal{S}_1$ and $\mathcal{S}_2$, its frequency is the sum of its corresponding frequencies in each summary; if the item is present only in one of either $\mathcal{S}_1$ or $\mathcal{S}_2$, its frequency is incremented by the minimum frequency of the other summary. At last, in order to derive the merged summary, we take only the two items in $\mathcal{S}_C$ with the greatest frequencies and discard the others. 

Once $R$ rounds of gossip have been executed, the user can invoke the \textsc{query} procedure (pseudo-code provided as Algorithm \ref{query}) on an arbitrary peer to retrieve the time--faded heavy hitters. This is done by computing $\tau$, a threshold that determines whether an item is a candidate heavy hitter or not, $\tilde{p}_{r,l}$, the estimate of $p$ and $\tilde{C}_{r,l}$, the global decayed count. Note that $\tau$ is defined in terms of $\epsilon^*$, whose meaning shall be explained in Section \ref{gba}. The total decayed count can be computed by summing up the decayed frequencies stored by the Space--Saving counters inside each single cell in a row of the sketch data structure $\mathcal{D}$. It is worth noting here that the decayed frequencies are still non normalized; the normalization occurs dividing by $g(t-L)$, where $t$ is the query time and $L$ denotes a landmark time in the past. The query procedure initializes $H$, an empty set, and then it inspects each of the $d \times w$ cells in the sketch $\mathcal{D}$. For a given cell, we determine $c_m$, the counter in the data structure $\mathcal{S}$ with maximum decayed count. We compare the normalized decayed count stored in $c_m$ with the threshold $\tau =  \phi \tilde{C}_{r,l} \frac{1-\epsilon^*}{1+\epsilon^*}$. If the normalized decayed frequency is greater, we pose a point query for the item $c_m.i$, shown in pseudo-code as the \textsc{PointEstimate} procedure. If the returned value is greater than the threshold $\tau$, then we insert in $H$ the pair $(c_m.i,p)$.

The point query for an item $j$ returns its estimated decayed frequency. After initializing the \textit{answer} variable to infinity, we inspect each of the $d$ cells in which the item is mapped to by the corresponding hash functions, to determine the minimum decayed frequency of the item. In each cell, if the item is stored by one of the Space-Saving counters, we set \textit{answer} to the minimum between \textit{answer} and the corresponding counter's decayed frequency. Otherwise (none of the two counters monitors the item $j$), we set \textit{answer} to the minimum between \textit{answer} and the minimum decayed frequency stored in the counters. Since the frequencies stored in all of the counters of the sketch are not normalized, we return the normalized frequency \textit{answer} dividing by $g(t-L)$. At the end of the query procedure the set $H$ is returned.

\section{Correctness}
\label{correctness}

In this Section, we prove that in our algorithm every peer correctly converges to a suitable sketch data structure that can be queried to solve the Approximate Time--Faded Heavy Hitters problem. We also prove a bound on the frequency estimation error committed by our algorithm. 

Before proceeding with our analysis, we need to recall the results by Jelasity et al. in \cite{Jelasity2005}, related to the convergence factor of the uniform gossiping protocol, and the main properties exhibited by the algorithm \textsc{FDCMSS} and by the procedure for merging \textsc{FDCMSS} sketches \cite{CAFARO2018115}, that we use in Algorithm \ref{p2ptfhhalg_merge}.

\subsection{Jelasity's averaging algorithm}
\label{jelasity}

In the work by Jelasity et al. \cite{Jelasity2005}, the proposed averaging protocol is interpreted as a distributed variance reduction algorithm, fully equivalent to a centralized algorithm which computes the average value among a set of elements. Consider a variance measure $\sigma_r^2$ defined as:

\begin{equation}
\label{eq_sigma2}
\sigma_r^2 = \frac{1}{p-1} \sum_{l=1}^{p}\left(w_{r,l} - \bar{w}\right)^2,
\end{equation}

where $w_{r,l}$ is the value held by the peer $l$ after $r$ rounds of the gossip algorithm and $\bar{w} = \frac{1}{p}\sum_{l=1}^p w_{0,l}$ is the mean of the initial values held by the peers. The authors in \cite{Jelasity2005} state that if $\psi_l$ is a random variable denoting the number of times a node $l$ is chosen as a member of the pair of nodes exchanging their states during a round of the protocol, and each pair of nodes is uncorrelated, then the following theorem holds.

\begin{thm}{\rm \cite{Jelasity2005}}
	\label{Jelasity}
	 If:
	\begin{enumerate}
		\item the random variables $\psi_1, \ldots , \psi_p$ are identically distributed (let $\psi$ denotes a random variable with this common distribution),
		\item after a pair of nodes $(i, j)$ is selected, the number of times the nodes $i$ and $j$ shall be selected again have identical distributions,
	\end{enumerate}
	then:
	\begin{equation}
	\label{eq_conv_factor}
	\mathbb{E}[\sigma_{r+1}^2] \approx \mathbb{E}[2^{-\psi}] \mathbb{E}[\sigma_{r}^2].
	\end{equation}
\end{thm}
\vspace{2mm}

The random variable $\psi$ only depends on how the pair of nodes are selected. From equation \eqref{eq_conv_factor}, the convergence factor is defined as:

\begin{equation}
\frac{\mathbb{E}[\sigma_{r+1}^2]}{ \mathbb{E}[\sigma_{r}^2]} = \mathbb{E}[2^{-\psi}];
\end{equation}

Therefore, the convergence factor depends on $\psi$ and, as a consequence, on the pair selection method. Jelasity et al. compute the convergence factor for different methods of pair selection; we are only interested in the one which allows simulating the distributed gossip--based averaging protocol. This method consists in drawing a random permutation of the nodes and then, for each node in that permutation, choosing another random node in order to form a pair. For this selection method, the convergence factor is $\gamma = \mathbb{E}[2^{-\psi}] = 1/(2\sqrt{e})$.

We now derive from theorem~\ref{Jelasity} the following proposition.

\begin{prop}
	\label{prop1}
	Let $\delta_g$ be a user-defined probability, $w_{r,l}$ the value held by peer $l$ after $r$ rounds of the averaging protocol, $p$ the number of peers participating in the protocol, $\gamma = \mathbb{E}[2^{-\psi}] = 1/(2\sqrt{e})$ the convergence factor, $\bar{w}$ the mean of the initial vector of values $\boldsymbol{w}_0$, i.e. $\bar{w} = 1/p \sum_{l=1}^{p} w_{0,l}$ and $\bar{\epsilon}=\sqrt{\frac{\gamma^r }{\delta_g}}$ the error factor. Then, with probability greater than or equal to $1 - \delta_g$ it holds that, for any peer $l$:
	
	\begin{equation}
	\label{eq_error_due_to_gossip}
	\left|w_{r,l} - \bar{w}\right| < \sqrt{(p-1) \sigma_0^2} \bar{\epsilon} < p\bar{w}\bar{\epsilon}
	\end{equation}
\end{prop}

\begin{proof}
	
	From equation \eqref{eq_conv_factor} it follows that:
	
	\begin{equation}
	\label{eq_expected_sigma2}
	\mathbb{E}[\sigma_{r}^2] = \mathbb{E}[2^{-\psi}]^r \sigma_0^2;
	\end{equation}
	
	where $\sigma_0^2$ depends on the distribution of the initial numbers among the peers. By the Markov inequality, it holds that:
	
	\begin{equation}
	\mathbb{P}[\sigma_{r}^2 \ge \frac{\mathbb{E}[\sigma_{r}^2]}{\delta_g}] \le  \delta_g;
	\end{equation}
	
	\noindent or
	
	\begin{equation}
	\mathbb{P}[\sigma_{r}^2 < \frac{\mathbb{E}[\sigma_{r}^2]}{\delta_g}] \ge  1 - \delta_g.
	\end{equation}
	
	Considering eqs. \eqref{eq_sigma2} and \eqref{eq_expected_sigma2}, it follows that:
	
	\begin{equation}
	\mathbb{P}[\sum_{l=1}^{p}\left(w_{r,l} - \bar{w}\right)^2 < (p-1) \frac{\gamma^r \sigma_0^2}{\delta_g}] \ge  1 - \delta_g.
	\end{equation}
	
	As a consequence, with probability at least $1-\delta_g$:
	
	\begin{equation}
	max_{l \in [p]} \left(w_{r,l} - \bar{w}\right)^2 \le \sum_{l=1}^{p}\left(w_{r,l} - \bar{w}\right)^2 < (p-1) \frac{\gamma^r \sigma_0^2}{\delta_g},
	\end{equation}
	
	\noindent which implies:
	
	\begin{equation}
	max_{l \in [p]} \left|w_{r,l} - \bar{w}\right| < \sqrt{(p-1)\sigma_0^2} \sqrt{\frac{\gamma^r }{\delta_g}}.
	\end{equation}
	
	The initial distribution $\sigma_0^2$ of the elements over the peers typically is not known in advance, but the worst case happens when only one peer has the whole quantity $p\bar{w}$ and the other $p-1$ peers hold the value $0$. In this case, it follows that  $\sigma_0^2 \leq p\bar{w}^2$, and hence, with probability $1-\delta_g$:
	
	\begin{equation}
	max_{l \in [p]} \left|w_{r,l} - \bar{w}\right| < \sqrt{(p-1)\sigma_0^2} \sqrt{\frac{\gamma^r }{\delta_g}} < p\bar{w} \bar{\epsilon}.
	\end{equation}
	
	This proves the proposition.	
\end{proof}

Equation \eqref{eq_error_due_to_gossip} gives an upper bound on the error made by any peer estimating the value $\bar{w}$ after $r$ rounds of the Jelasity's averaging algorithm. This bound is probabilistic and valid with probability greater than or equal to $1-\delta_g$. %Since the bound is valid for all of the peers, from now on we refer to the value of each peer by omitting the $l$ subscript.

\subsection{\textsc{FDCMSS} algorithm and merging of augmented sketches}
\label{merge_section}
We recall here the main results related to the theoretical properties of the \textsc{FDCMSS} algorithm. 
Cafaro et al. in \cite{Cafaro-Pulimeno-Epicoco-Aloisio} proved that \textsc{FDCMSS} solves the  Approximate Time--Faded Heavy Hitters problem and give a bound of its frequency estimation error.
In particular, the authors state that, with high probability, if a time--faded item is frequent, then, in at least one of the \textsc{FDCMSS} sketch cells where it is mapped to by the hash functions, it is a majority item with regard to the sub-stream of items falling in the same cell. Therefore, \textsc{FDCMSS} will detect it. 

The following theorem holds, being $d \times w$ the dimensions of the sketch and $\phi$ the frequent threshold.

\begin{thm}
	\label{fdcmss-correct}
	If an item $i$ is frequent, then it appears as a majority item candidate in at least one of the $d$ cells in which it falls, with probability greater than or equal to $1 - (\frac{1}{2 \phi w})^d$. 
\end{thm}

Moreover, regarding the error bound on frequency estimation of the \textsc{FDCMSS} algorithm, if $f_i$ is the exact decayed frequency of item $i$ in the stream $\sigma$ and $\hat{f}_i$ is the estimated decayed frequency of item $i$ returned by \textsc{FDCMSS} and $C$ is the total decayed count of all of the items in the input stream, then, the authors proved that the following error bound holds:

\begin{thm}
	\label{error-bound}
	Denoting by $C$ the total decayed count, and by $d$ and $w$ the number of rows and columns of the augmented sketch data structure, then $\forall i \in [m]$, $\hat{f}_i$ estimates the exact decayed count $f_i$ of item $i$ at query time  with error less than $eC/2w$ and probability greater than $1-e^{-d}$, i.e.:
	
	\begin{equation}
	f_i \le \hat{f}_i \le f_i + \frac{eC}{2w};
	\end{equation}
	
\end{thm}

\textsc{P2PTFHH} follows the same structure of the gossip--based averaging protocol by Jelasity et al., where, instead of single values, two nodes communicating provide two sketches to be merged and \emph{averaged}. The merge operation refers to the procedure introduced by Cafaro et al. in \cite{CAFARO2018115}, where the authors proved that the merge procedure of two augmented sketch data structures preserves all of the properties of the sketch, including the fact that an \textsc{FDCMSS} sketch is 1-norm equivalent to a classical Count-Min sketch for the same input stream, i.e. the value of a cell in the Count-Min sketch is equal to the sum of the Space-Saving counters in the same cell of the \textsc{FDCMSS} sketch.

Therefore, if $\mathcal{D}_1$ and $\mathcal{D}_2$ are \textsc{FDCMSS} sketches which respectively refer to the streams $\mathcal{N}_1$ and $\mathcal{N}_2$, then the sketch  $\mathcal{D} = \mathcal{D}_1 \oplus \mathcal{D}_2$ is an \textsc{FDCMSS} sketch for the stream $\mathcal{N} = \mathcal{N}_1 \uplus \mathcal{N}_2$, where we have used the symbol $\oplus$ to indicate the merge operation.

\subsection{Convergence of \textsc{P2PTFHH}}
\label{convergence}

In order to prove the convergence of \textsc{P2PTFHH}, we shall use multisets to represent the input streams and introduce some operations on multisets and sketches in order to replicate the average pattern of Jelasty et al. algorithm. 

\begin{defn}
	\label{multiset}
	A multiset $\mathcal{N}=(N, f_{\mathcal{N}})$ is a pair where $N$ is some set, called the underlying set of elements, and $f_{\mathcal{N}}: N \rightarrow \mathbb{N}$ is a function.
	The generalized indicator function of $\mathcal{N}$ is 
	\begin{equation}
	I_\mathcal{N} (x) := \left\{ {\begin{array}{*{20}c}
		{f_{\mathcal{N}}(x)} & {x \in N} , \\
		0 & {x \notin N},  \\
		\end{array} }\right.
	\end{equation}
	
	\noindent where the integer--valued function $f_{\mathcal{N}}$, for each $x \in N$, provides its  \textit{frequency} (or multiplicity), i.e., the number of occurrences of $x$ in $\mathcal{N}$. 
	\noi The cardinality of $\mathcal{N}$ is expressed by

	\begin{equation}
	\left\vert{\mathcal{N}}\right\vert := Card(\mathcal{N}) = \sum\limits_{x \in N} {I_\mathcal{N} (x)},
	\end{equation}
	
	\noi whilst the cardinality of the underlying set $N$ is
	
	\begin{equation}
	\left\vert{N}\right\vert := Card(N) = \sum\limits_{x \in N} {1}.
	\end{equation}
\end{defn}

A multiset, or bag,  is defined by a proper set (the support set) and a multiplicity function: it is a set where elements can be repeated, i.e., an element in a multiset can have multiplicity greater than one. 

For our specific application, we relax the above definition of multiset, allowing the multiplicity function $f_N$ to assume real values.
Let $\mathcal{M}$ be the class of all of the multisets with support set included in $ \mathcal{U}$. 
We introduce the operation $\oslash_v: \mathcal{M} \rightarrow \mathcal{M}$, so that $\oslash_v(\mathcal{N}) = (N, f_N/v))$, i.e, the multiset $\oslash_v(\mathcal{N})$ has the same support set of $\mathcal{N}$, but each  element has a fraction $1/v$ of the multiplicity it has in $\mathcal{N}$. When $v \in \mathbb{N}$, we have that $\biguplus_{i=1}^v\oslash_v(\mathcal{N}) =  \oslash_{\frac{1}{v}}(\oslash_v(\mathcal{N})) = \mathcal{N}$.

We can define a similar operation for an augmented \textsc{FDCMSS} sketch $\mathcal{D}$ and, by abusing the same symbol, say that $\oslash_v(\mathcal{D}) = \mathcal{D'}$, where $\mathcal{D'}$ is obtained from $\mathcal{D}$ by dividing all of the frequencies' values in $\mathcal{D}$ by $v$.

It is immediate to see that if $\mathcal{D}$ is an augmented sketch for $\mathcal{N}$, then $\oslash_v(\mathcal{D})$ is a augmented sketch for $\oslash_v(\mathcal{N})$ and theorems \ref{fdcmss-correct} and \ref{error-bound} continue to hold. 

Algorithm~\ref{algo1}, called \textsc{AVG-Merge}, is a simplified global version of the distributed \textsc{P2PTFHH} that simulates the distributed averaging protocol. \textsc{AVG-Merge} operates on the global state of the network by simulating the distributed \textsc{P2PTFHH} protocol and allowing us to simplify its theoretical analysis. \textsc{GetPair}, not reported here, is a procedure which selects a pair of nodes and meets the properties requested in theorem \ref{Jelasity}. 

\begin{algorithm}
	\begin{algorithmic}
		\caption{\textsc{AVG-Merge}: global \textsc{FDCMSS} sketches average merging}
		\label{algo1}
		\Require{$\boldsymbol{\mathcal{D}}_r = (\mathcal{D}_{r,1}, \mathcal{D}_{r,2}, \ldots, \mathcal{D}_{r,p})$, the vector of \textsc{FDCMSS} sketches at round $r$; $p$, the number of peers.}
		\Ensure{$\boldsymbol{\mathcal{D}}_{r+1}$, the vector of \textsc{FDCMSS} sketches updated for round $r+1$.}
		\Procedure{AVG-Merge}{$\boldsymbol{\mathcal{D}}_r, p$}
		\State $l \leftarrow 0$
		\While{$l < p$}
		\State $(i, j) \leftarrow $ \Call{GetPair}{\ } 	
		\State $\mathcal{D}_{r,i} \leftarrow \mathcal{D}_{r,j} \leftarrow  \oslash_2(\mathcal{D}_{r,i} \oplus \mathcal{D}_{r,j})$	
		\State $l \leftarrow l + 1$
		\EndWhile
		\State $\boldsymbol{\mathcal{D}}_{r+1} \leftarrow \boldsymbol{\mathcal{D}}_r$
		\State \Return $\boldsymbol{\mathcal{D}}_{r+1}$ 
		\EndProcedure
		
	\end{algorithmic}
\end{algorithm}

\vspace{4mm}

Initially, each peer computes a local augmented sketch on its input stream, through the execution of the \textsc{FDCMSS} algorithm, then the distributed protocol starts. The initial distributed state of the system can be represented by the vector of the local augmented sketches ${\boldsymbol{\mathcal{D}}_0 = (\mathcal{D}_{0,1},\mathcal{D}_{0,2},\ldots,\mathcal{D}_{0,p})}$, where $p$ is the number of peers participating in the protocol. Another vector is naturally associated to $\boldsymbol{\mathcal{D}}_0$: the vector of the local input streams $\boldsymbol{\mathcal{N}}_0 = (\mathcal{N}_{0,1}, \mathcal{N}_{0,2}, \ldots, \mathcal{N}_{0,p})$. We have that $\biguplus_{l = 1}^p \mathcal{N}_{0,l} = \mathcal{N}$, where we denote by $\mathcal{N}$ the global input stream.

Each call to \textsc{AVG-Merge} corresponds to a round of \textsc{P2PTFHH}. It modifies $\boldsymbol{\mathcal{D}}_r$, the vector of the local sketches held by the peers at the end of round $r$, producing the vector $\boldsymbol{\mathcal{D}}_{r+1}$. Furthermore,  implicitly also $\boldsymbol{\mathcal{N}}_{r}$, the vector of local input streams to which the sketches refer, changes to $\boldsymbol{\mathcal{N}}_{r+1}$. Indeed, let $\boldsymbol{\mathcal{D}}_r$ and $\boldsymbol{\mathcal{N}}_r$ be the vectors of the sketches owned by the peers and the corresponding partition of the input stream $\mathcal{N}$ after the $r$th round. Then, after each iteration of the main loop of \textsc{AVG-Merge}, letting $(i, j)$ be the pair of communicating peers, i.e. the pair selected by \textsc{GetPair}, the vector of sketches becomes: 

\begin{equation}
\boldsymbol{\mathcal{D}}_{r}' = (\mathcal{D}_{r,1}, \mathcal{D}_{r,2}, \ldots, \oslash_2(\mathcal{D}_{r,i} \oplus \mathcal{D}_{r,j}), \ldots, \oslash_2(\mathcal{D}_{r,i} \oplus \mathcal{D}_{r,j}), \ldots, \mathcal{D}_{r,p}),
\end{equation}

\noindent and the corresponding vector of partitions of the input stream shall change to:

\begin{equation}
\label{elem-it-msets}
\boldsymbol{\mathcal{N}}_{r}' = (\mathcal{N}_{r,1}, \mathcal{N}_{r,2}, \ldots, \oslash_2(\mathcal{N}_{r,i} \uplus \mathcal{N}_{r,j}), \ldots, \oslash_2(\mathcal{N}_{r,i} \uplus \mathcal{N}_{r,j}), \ldots, \mathcal{N}_{r,p}).
\end{equation}

From what we said on the operations $\oplus$ and $\oslash$, after each elementary iteration of \textsc{AVG-Merge}, two invariants hold: 
\begin{enumerate}
	\item each peer $l$ owns a sketch $\mathcal{D}_{r,l}$ which is a correct augmented sketch data structure for the portion of input stream $\mathcal{N}_{r,l}$;
	\item  $\biguplus_{l = 1}^p \mathcal{N}_{r,l} = \mathcal{N}$.
\end{enumerate}

These invariants remain true after each iteration of the main loop of \textsc{AVG-Merge} and, consequently, after each call to \textsc{AVG-Merge}, that is after each round of the \textsc{P2PTFHH} distributed protocol, when we derive from the vectors $\boldsymbol{\mathcal{D}}_r$ and $\boldsymbol{\mathcal{N}}_r$, the new vectors $\boldsymbol{\mathcal{D}}_{r+1}$ and $\boldsymbol{\mathcal{N}}_{r+1}$. 

We can state that, for $r \rightarrow \infty$, the two vectors $\boldsymbol{\mathcal{D}}_r$ and $\boldsymbol{\mathcal{N}}_r$ converge respectively to:

\begin{equation}
\boldsymbol{\mathcal{D}}_{\infty} = \left(\mathcal{D}_{\rm avg}, \mathcal{D}_{\rm avg}, \ldots, \mathcal{D}_{\rm avg}\right)
\end{equation}

\noindent and 

\begin{equation}
\label{conv-msets}
\boldsymbol{\mathcal{N}}_{\infty} = \left(\mathcal{N}_{\rm avg}, \mathcal{N}_{\rm avg}, \ldots, \mathcal{N}_{\rm avg}\right),
\end{equation}

\noindent where $\mathcal{N}_{\rm avg} = \oslash_p(\mathcal{N})$ and $\mathcal{D}_{\rm avg}$ is a correct \textsc{FDCMSS} sketch for $\mathcal{N}_{\rm avg}$. 

This means that all of the peers converge to a sketch of $\oslash_p(\mathcal{N})$, from which, for the properties of the operations $\oplus$ and $\oslash$, a correct sketch for $\mathcal{N}$ can be derived by computing $\oslash_\frac{1}{p}(\mathcal{D}_{\rm avg})$ (we actually need to know the number of peers, which is not always the case, but we shall see in the following how we can estimate $p$), i.e. \textsc{P2PTFHH} converges. 

Thanks to the invariants discussed above, in order to prove the convergence of the local sketches to $\mathcal{D}_{\rm avg}$,  it is enough to verify that the  local input streams implicitly induced by the algorithm converge to $\mathcal{N}_{\rm avg}$.

We can represent each initial local input stream $\mathcal{N}_{0,l}$ for $l = 1,2, \ldots, p$, as the frequencies' vector of the items in that stream, $\boldsymbol{\tilde{f}}_{0,l} = (\tilde{f}_{0,l,1}, \tilde{f}_{0,l,2}, \ldots, \tilde{f}_{0, l,m})$. Each value $\tilde{f}_{0,l,i}$ corresponds to the frequency that item $i$ has in the initial local stream held by peer $l$. In this representation the operator $\oslash_p$ on a multiset translates to a multiplication of the frequencies' vector corresponding to that multiset by the scalar $1/p$. 

Now, the implicit transformation that the local streams of the selected peers, $i$ and $j$, undergo at each elementary iteration of \textsc{AVG-Merge}, i.e., equation \eqref{elem-it-msets}, can be rewritten as:

\begin{equation}
\label{elem-it-fvecs}
\boldsymbol{\tilde{F}}_r' = (\boldsymbol{\tilde{f}}_{r,1}, \boldsymbol{\tilde{f}}_{r,2}, \ldots, \frac{1}{2} (\boldsymbol{\tilde{f}}_{r,i} + \boldsymbol{\tilde{f}}_{r,j}), \ldots, \frac{1}{2} (\boldsymbol{\tilde{f}}_{r,i} + \boldsymbol{\tilde{f}}_{r,j}), \ldots, \boldsymbol{\tilde{f}}_{r,p}),
\end{equation}

\noindent where $\boldsymbol{\tilde{F}}_r$ is a matrix whose columns are the peers' vectors of frequencies after $r$ rounds, i.e. each $\boldsymbol{\tilde{f}}_{r,l}$ is the frequencies' vector corresponding to the multiset  ${\mathcal{N}}_{r,l}$. This matrix corresponds to the vector of multisets $\boldsymbol{\mathcal{N}}_r$ in equation \eqref{elem-it-msets}.

Eventually, it can be recognized in equation \eqref{elem-it-fvecs} the elementary step of the Jelasity's protocol applied componentwise to the frequencies' vectors of peers $i$ and $j$. We already know that the Jelasity's averaging protocol converges to the average of the values initially owned by the peers. Thus, for $r  \rightarrow \infty$,  $\boldsymbol{\tilde{F}}_r$ converges to: 

\begin{equation}
\boldsymbol{\tilde{F}}_\infty = (\boldsymbol{f}_{\rm avg}, \boldsymbol{f}_{\rm avg}, \ldots, \boldsymbol{f}_{\rm avg})
\end{equation}

\noindent where $\boldsymbol{f}_{\rm avg}$ is:

\begin{equation}
\boldsymbol{f}_{\rm avg} = (\bar{f}_1, \bar{f}_2, \ldots, \bar{f}_m),
\end{equation}

\noindent with $\bar{f}_i = \frac{1}{p} \sum_{l=1}^{p} \tilde{f}_{0,l,i}, \text{ for } i = 1,2, \ldots, m$ which is the representation as frequencies' vector of the multiset $\mathcal{N}_{\rm avg}$ in equation \eqref{conv-msets}, proving the convergence.

\subsection{Gossip-based approximation}
\label{gba}
In our \textsc{P2PTFHH} algorithm, all of the peers participate to the averaging protocol by exchanging their internal state, i.e., by  merging the local augmented sketches,  $\mathcal{D}_{r,l}$, and averaging the local values $\tilde{q}_{r,l}$. 

In order to estimate the number of peers, the averaging protocol is performed using the values initially distributed among the peers so that only one peer starts with a value equal to $1$ while the other $p-1$ peers start with values equal to $0$. Denoting with $\tilde{q}_{r,l}$ the value  held by peer $l$ at round $r$, this quantity approaches the average value $1/p$ as the algorithm proceeds. We can bound the error for the estimation of the number of peers. Considering the equation \ref{eq_error_due_to_gossip}, we have that :

\begin{equation}
\label{eq_qestimate}
\left|\tilde{q}_{r,l} - \frac{1}{p}\right| < \bar{\epsilon}.
\end{equation}

In fact, equation \ref{eq_qestimate} can be easily derived from equation \ref{eq_error_due_to_gossip} by substituting $\bar{w}$ with the average value $1/p$. Since $\tilde{q}$ is an estimate of the inverse of the number of peer, we can write equation \ref{eq_qestimate} introducing $\tilde{p} = 1/\tilde{q}$:

\begin{equation}
\frac{1}{p} - \bar{\epsilon} < \frac{1}{\tilde{p}_{r,l}} < \frac{1}{p} + \bar{\epsilon}.
\end{equation}

Since the number of peers is a positive number, we can set the constraint $\bar{\epsilon} < 1/p$; under this constraint all of the members of the previous relation are positive, hence it holds that:

\begin{equation}
\label{error_bound_on_p_1}
\frac{1}{1+p\bar{\epsilon}} < \frac{\tilde{p}_{r,l}}{p} < \frac{1}{1-p\bar{\epsilon}} 
\end{equation}

The problem with equation \eqref{error_bound_on_p_1} is that the estimation error bounds depend on $p$, but we may not know $p$ in advance. To overcome this problem, we introduce the value $p^*~\ge~p$, that is an estimate of the maximum number of peers we expect in the network, and we compute new bounds based on this value. Under the constraint $p^* \geq p$, we can be confident on the new computed bounds, though they may be weaker.

Let us set $\epsilon^* = p^* \bar{\epsilon}$. Given the constraint on $\bar{\epsilon}$, it holds that $0 < \epsilon^* < 1$, and, with probability $1-\delta_g$, for any peer $l = 1,2,\ldots,p$:

\begin{equation}
\label{error_bound_on_p_2}
\frac{p}{1+\epsilon^*} < \tilde{p}_{r,l} < \frac{p}{1-\epsilon^*} 
\end{equation}

As discussed in Section \ref{convergence}, as the the \textsc{P2PTFHH} algorithm proceeds, the augmented Sketch $\mathcal{D}_{r,l}$, held by a peer $l$ after $r$ rounds approaches $\mathcal{D}_{\rm avg}$ which corresponds to the augmented sketch we would obtain had we applied \textsc{FDCMSS} to the stream $\mathcal{N}_{\rm avg}$, which is, in turn, the stream corresponding to the average frequency vector that we would obtain applying the Jelasty averaging protocol component wise to the local peers' frequency vectors. 

Denoting with $\tilde{f}_{r,l,i}$ the decayed frequency of the item $i$ in the stream $\mathcal{N}_{r,l}$, after $r$ rounds, and considering the equation \ref{eq_error_due_to_gossip}, it holds that:

\begin{equation}
\label{eq_festimate}
\left|\tilde{f}_{r,l,i} - \frac{f_{i}}{p}\right| < f_{i}\bar{\epsilon}.
\end{equation}
   
Equation \ref{eq_festimate} can be easily derived from equation \ref{eq_error_due_to_gossip} by substituting $\bar{w}$ with the average value $f_{i}/p$. Considering the expression of $\epsilon^*$, it follows that, with probability $1-\delta_g$, for any peer $l = 1,2,\ldots,p$:

\begin{equation}
\label{freq_bounds}
\frac{f_i}{p} (1 - \epsilon^*) < \tilde{f}_{r,l,i} < \frac{f_i}{p} (1 + \epsilon^*)
\end{equation}

Finally, the total decayed count $C$ shall be estimated too. Let $\tilde{C}_{r,l}$ be the total decayed count estimate, we have that $\tilde{C}_{r,l} = \sum_{i = 1}^{d} \tilde{f}_{r,l,i}$, thus, from the previous equation \ref{freq_bounds}, it follows that, with probability $1-\delta_g$, for any peer $l = 1,2,\ldots,p$:

\begin{equation}
\frac{C}{p} (1 - \epsilon^*) < \tilde{C}_{r,l} < \frac{C}{p} (1 + \epsilon^*)
\end{equation}

Moreover, it is worth noting here that, due to its properties, the augmented sketch data structure being 1-norm equivalent to the classical Count--Min sketch,  
each peer can determine the total decayed count estimate $\tilde{C}_{r,l}$, simply by summing up the values of all of the cells in one row of its sketch, $\mathcal{D}_{r,l}$.
  
\subsection{Correctness and error bounds}
We shall show here that given an augmented sketch data structure $\mathcal{D}_{r,l}$ obtained by the peer $l$ after $r$ rounds of the \textsc{P2PTFHH} algorithm, we can select a set of items and their corresponding estimated frequencies, thus solving the Approximate Time--Faded Heavy Hitters Problem. We shall also determine the error bounds on frequencies' estimation and the relation among the number of the sketch cells to be used by each peer and the number $r$ of rounds to be executed in order to guarantee the false positives' tolerance requested by the user within a given probability of failure.  

\begin{thm}
	\label{thm_falsenegative}
	Given an input stream $\mathcal{N}$, with total decayed count $C$, distributed among $p$ peers, each one making use of an augmented sketch data structure of size $d \times w$, and given a threshold parameter $0<\phi<1$, after $r$ rounds of \textsc{P2PTFHH}, on any peer the \textsc{QUERY} procedure returns a set $H$ of items and their corresponding time--faded frequencies so that, with probability greater than $1-\delta_n$,\\
	
	$H$ includes all of the items in $\mathcal{N}$ with frequency $f > \phi C$; \\
	
	\noindent with $\delta_n = \delta_g + \left(\frac{1}{2 \phi w}\frac{1+\epsilon^*}{1-\epsilon^*}\right)^d(1-\delta_g)$, $\delta_g$ being the probability of failure pertaining to the gossip approximation.\\
\end{thm}

\begin{proof}
We recap the query procedure of the \textsc{P2PTFHH} algorithm: all of the sketch cells are inspected, if the majority item candidate in the current cell has a decayed frequency greater than the selection criteria $\tau = \phi \tilde{C}_{r,l} \frac{1-\epsilon^*}{1+\epsilon^*}$, then a \textsc{PointEstimate} query is issued to compute the estimated frequency $\hat{f}$ of the heavy hitter candidate and finally if  $\hat{f} > \tau$ the item is inserted in the set $H$.

Hence, to be selected a heavy hitter must:
\begin{itemize}
	\item appear as a majority item candidate in at least one of the $d$ cells in which it is mapped to;
	\item have estimated frequency $\hat{f}$ greater than $\tau$.
\end{itemize}

We first prove that if an item is a heavy hitter (i.e. $f>\phi C$), then its estimated frequency is greater than $\tau$ with probability at least $1-\delta_g$. We recap the main relations we proved above, valid with probability at least $1-\delta_g$, for all of the items $i$ and any  peer $l$ after $r$ rounds of the algorithm:

\begin{align}
\label{eq_properties1}
%\begin{split}
\frac{p}{1+\epsilon^*} < & \tilde{p}_{r,l} < \frac{p}{1-\epsilon^*}; \\
\label{eq_properties2}
\frac{f_i}{p} (1 - \epsilon^*) < & \tilde{f}_{r,l,i} < \frac{f_i}{p} (1 + \epsilon^*); \\
\label{eq_properties3}
\frac{C}{p} (1 - \epsilon^*) < & \tilde{C}_{r,l} < \frac{C}{p} (1 + \epsilon^*); \\
%\end{split}
\end{align}

Moreover, owing to the \textsc{FDCMSS} properties it holds that:

\begin{equation}
\label{eq_properties4}
\tilde{f}_{r,l,i} \le \hat{f}_{r,l,i}.
\end{equation}

%\begin{equation}
%	\label{eq_properties4}
%	\tilde{f}_{r,l,i} \le \hat{f}_{r,l,i} \le \tilde{f}_{r,l,i} + \frac{e\tilde{C}_{r,l}}{2w}.
%\end{equation}

From the relations \eqref{eq_properties1}--\eqref{eq_properties4}, we can derive the following:

\begin{equation}
\hat{f}_{r,l,i} \frac{p}{1-\epsilon^*} > \tilde{f}_{r,l,i} \frac{p}{1-\epsilon^*} > f_i > \phi C > \phi \tilde{C}_{r,l} \frac{p}{1+\epsilon^*}
\end{equation}

Hence the following holds with probability at least $(1-\delta_g)$:

\begin{equation}
\hat{f}_{r,l,i} > \phi \tilde{C}_{r,l} \frac{1-\epsilon^*}{1+\epsilon^*}.
\end{equation}

We now prove that a heavy hitter appears in at least one of the $d$ cells in which it is mapped to with probability at least $\left[1-\left(\frac{1}{2w\phi} \frac{1+\epsilon^*}{1-\epsilon^*}\right)^{d}\right](1-\delta_g)$ and that the overall probability of \textsc{P2PTMHH} failing to retrieve all of the frequent items, or alternatively, the probability of \textsc{P2PTFHH} reporting false negatives, is $\delta_n$.

Let us take as a reference the state of a peer $l$ after the $r$th round of the algorithm and let $k$ be a frequent item for the global stream $\mathcal{N}$, i.e. $f_k>\phi C$. Let $\tilde{\mathcal{D}}_{r,l}[i][j]$, for $i \in \{1\ldots d\}$ and $j = h_i(k)$, be one of the $d$ cells in which the item $k$ is mapped to by the corresponding hash functions $h_i: [m] \rightarrow [w]$.  Moreover, let $\tilde{f}^{(i,j)}_{min}$ be the minimum of the two counters available in the Space-Saving summary $\tilde{\mathcal{S}}_{i,j}$, which monitors the items falling in the cell $\tilde{\mathcal{D}}_{r,l}[i][j]$, and let $\hat{f}_{i,k}$ be the estimated decayed count of the item $k$ related to the local stream $\tilde{\mathcal{N}}_{r,l}$ and determined through the Space-Saving summary $\tilde{\mathcal{S}}_{i,j}$. 
The item $k$ shall not appear as a majority item candidate in any of the $d$ cells in which it is mapped to iff the following event happens:

\begin{equation}
\label{failure_eq}
\hat{f}_{i,k} \leq \tilde{f}^{(i,j)}_{min}, \forall i=1,\ldots,d.
\end{equation}

Let $\tilde{S}_{i,j}$ be the total decayed count of the items falling in the cell $\tilde{D}_{r,l}[i][j]$. By equation (\ref{ss5}), for the minimum of the two counters in the cell $\tilde{D}_{r,l}[i][j]$, it holds that: $\tilde{f}^{(i,j)}_{min} \leq \frac{\tilde{S}_{i,j}}{2}$.

We must now determine the probability that the event described in equation (\ref{failure_eq}) occurs for a given row $i$ of the sketch, i.e. $\Pr[\hat{f}_{i,k} \leq \tilde{f}^{(i,j)}_{min}]$, under the condition that $f_k > \phi C$. 

Taking into account that $\tilde{f}^{(i,j)}_{min} \leq \frac{\tilde{S}_{i,j}}{2}$, it holds that:

\begin{equation}
\label{prob}
\Pr[\hat{f}_{i,k} \leq \tilde{f}^{(i,j)}_{min}] < \Pr[\tilde{S}_{i,j} > 2\hat{f}_{i,k}];
\end{equation}

and, as proved in \cite{Cafaro-Pulimeno-Epicoco-Aloisio}, $\E[\tilde{S}_{i,j}] = \frac{\tilde{C}_{r,l}}{w}$. Thus, using the Markov inequality we have: %can bound the probability of failure (i.e., the probability of the item not being reported as a majority item candidate):

\begin{equation}
\Pr[\hat{f}_{i,k} \leq \tilde{f}^{(i,j)}_{min}] < \Pr[\tilde{S}_{i,j} > 2\hat{f}_{i,k}] \leq \frac{\E[\tilde{S}_{i,j}]}{2\hat{f}_{i,k}} = \frac{\tilde{C}_{r,l}}{2 w \hat{f}_{i,k}}.
\end{equation}

By applying equations \eqref{eq_properties1}--\eqref{eq_properties4} we have:

\begin{equation}
\Pr[\hat{f}_{i,k} \leq \tilde{f}^{(i,j)}_{min}] < \frac{\tilde{C}_{r,l}}{2 w \hat{f}_{i,k}} < \frac{\tilde{C}_{r,l}}{2 w \tilde{f}_{r,l,k}} < \frac{\tilde{C}_{r,l}}{2 w \frac{f_{k}}{p}(1-\epsilon^*)} < \frac{\frac{C}{p}(1+\epsilon^*)}{2 w \frac{f_{k}}{p}(1-\epsilon^*)}.
\end{equation}

Finally, taking into account that $k$ is a heavy hitter, i.e. $f_k>\phi C$, it holds that:

\begin{equation}
\label{single-cell-failure}
\Pr[\hat{f}_{i,k} \leq \tilde{f}^{(i,j)}_{min}] < \frac{1+\epsilon^*}{2\phi w(1-\epsilon^*)}.
\end{equation}

We fail to identify a frequent item $k$ when that item is not reported as a majority item candidate in any of the $d$ cells $\tilde{D}_{r,l}[i][j]$ in which it falls. Thus, by equation (\ref{single-cell-failure}), the probability of this event happening is: 

\begin{equation}
\Pr\left[\bigwedge_{i=1}^{d} (\hat{f}_{i,k} \leq \tilde{f}^{(i,j)}_{min})\right]  < \left(\frac{1}{2 \phi w}\frac{1+\epsilon^*}{1-\epsilon^*}\right)^d.
\end{equation}

Considering that the equations \eqref{eq_properties1}--\eqref{eq_properties4} are correct with a given probability $(1-\delta_g)$ (see Section \ref{jelasity}), we can conclude that if an item is a heavy hitter, it is reported with probability greater than or equal to $\left[1 - \left(\frac{1}{2 \phi w}\frac{1+\epsilon^*}{1-\epsilon^*}\right)^d\right]\left(1-\delta_g \right)$, and that, consequently, the probability of failure of \textsc{P2PTFHH} retrieving all of the frequent items is: $$\delta_n = \delta_g + \left(\frac{1}{2 \phi w}\frac{1+\epsilon^*}{1-\epsilon^*}\right)^d(1-\delta_g)$$.
\end{proof}

\begin{thm}
	\label{thm_falsepositive}
	Given an input stream $\mathcal{N}$, with total decayed count $C$, distributed among $p$ peers, each one making use of an augmented sketch data structure of size $d \times w$, and given a threshold parameter $0<\phi<1$, after $r$ rounds of \textsc{P2PTFHH}, on any peer the \textsc{QUERY} procedure returns a set $H$ of items and their corresponding time--faded frequencies so that, with probability greater than or equal $1-\delta_p$, \\
	
	 $H$ does not include any item in $\mathcal{N}$ with frequency $f \le (\phi - \epsilon) C$;\\
	
	\noindent with a false positives tolerance $\epsilon = \frac{4\epsilon^* \phi}{(1+\epsilon^*)^2} + \frac{e}{2w}\frac{1-\epsilon^*}{1+\epsilon^*}$, and a probability of failure $\delta_p = \delta_g + e^{-d}(1-\delta_g)$, $\delta_g$ being the probability of failure pertaining to the gossip approximation.\\
\end{thm}

\begin{proof}
	Owing to the \textsc{FDCMSS} properties, the estimated frequency of an item $i$ is bounded by the following relation, with probability at least $1-e^{-d}$:
	
	\begin{equation}
		\label{eq_properties4b}
		\hat{f}_{r,l,i} \le \tilde{f}_{r,l,i} + \frac{e\tilde{C}_{r,l}}{2w}.
	\end{equation}
	
In order to compute the error, in terms of false positives' tolerance, that we commit with the selection criterion $\hat{f}>\tau$, we can use equations \eqref{eq_properties1}--\eqref{eq_properties3} and equation \eqref{eq_properties4b} and prove that if $\hat{f}_{r,l,i} > \phi \tilde{C}_{r,l} \frac{1-\epsilon^*}{1+\epsilon^*}$, then, with probability at least $(1-\delta_g)(1-e^{-d})$:
	
	\begin{equation}
	f_i > \left\{\phi - \left[\frac{4\epsilon^* \phi}{(1+\epsilon^*)^2} + \frac{e}{2w}\right]\right\} C.
	\end{equation}
	
	\noindent Applying equation \ref{eq_properties4b}, we have that, if an item $i$ is selected to be reported in the set $H$, then, with probability at least $(1-e^{-d})$:
	
	\begin{equation}
	\begin{split}
	& \tilde{f}_{r,l,i} + \frac{e\tilde{C}_{r,l}}{2w} > \hat{f}_{r,l,i} > \phi \tilde{C}_{r,l} \frac{1-\epsilon^*}{1+\epsilon^*} \implies\\
	& \frac{\tilde{f}_{r,l,i}}{\tilde{C}_{r,l}} + \frac{e}{2w} > \phi \frac{1-\epsilon^*}{1+\epsilon^*} .\\
	\end{split}
	\end{equation}
	
	\noindent Applying now equations \ref{eq_properties1}-\ref{eq_properties3}, which are valid with probability at least $1-\delta_g$, we can state that, with probability at least $(1-\delta_g)(1-e^{-d})$, i.e. with a probability of failure $\delta_p = \delta_g + e^{-d}(1-\delta_g)$, if item $i$ is reported in the output set $H$, then:

	\begin{equation}
	\begin{split}
	& \frac{f_i(1+\epsilon^*)}{C(1-\epsilon^*)} + \frac{e}{2w}> \phi \frac{1-\epsilon^*}{1+\epsilon^*} \implies \\
	& f_i > \phi C \left(\frac{1-\epsilon^*}{1+\epsilon^*}\right)^2 - \frac{eC}{2w}\frac{1-\epsilon^*}{1+\epsilon^*}\implies \\
	& f_i > \left\{\phi - \left[1 - \left(\frac{1-\epsilon^*}{1+\epsilon^*}\right)^2\right] \phi + \frac{e}{2w}\frac{1-\epsilon^*}{1+\epsilon^*}\right\}  C \implies \\
	& f_i > \left\{\phi - \left[\frac{4\epsilon^* \phi}{(1+\epsilon^*)^2} + \frac{e}{2w}\frac{1-\epsilon^*}{1+\epsilon^*}\right]\right\} C.
	\end{split}
	\end{equation}
\end{proof}

Theorems \ref{thm_falsenegative} and \ref{thm_falsepositive} show that the set of candidate frequent items produced by \textsc{P2PTFHH} satisfies the requirements of the Approximate Time Faded Heavy Hitters problem (Problem \ref{prob1}). Thus, we can set the parameters of the algorithm, i.e. the dimensions of the sketches and/or the number of rounds to be executed, so that the absence of false negatives and the tolerance on false positives are guaranteed with certain probabilities of failure that in turn depend on the algorithm's parameters and can be kept small. 

Thus, the overall probability of failure of \textsc{P2PTFHH} is driven by $\delta_p$ and we can avoid dealing with the two distinct probabilities of failure, $\delta_n$ and $\delta_p$. 

We know from Problem \ref{prob1} that the desired value for the tolerance $\epsilon$ must be such that $\epsilon < \phi$; therefore, starting from the expression of $\epsilon$ from theorem \ref{thm_falsepositive}, we have that:

\begin{equation}
\begin{split}
& \frac{4\epsilon^* \phi}{(1+\epsilon^*)^2} + \frac{e}{2w}\frac{1-\epsilon^*}{1+\epsilon^*} < \phi \implies \\
&  \left[\frac{4\epsilon^*}{(1+\epsilon^*)^2} - 1\right] \phi + \frac{e}{2w}\frac{1-\epsilon^*}{1+\epsilon^*} < 0 \implies \\
& \frac{e}{2w}\frac{1-\epsilon^*}{1+\epsilon^*} - \left(\frac{1-\epsilon^*}{1+\epsilon^*}\right)^2 \phi < 0 \implies \\
& \frac{e}{2w} -  \phi \frac{1-\epsilon^*}{1+\epsilon^*} < 0 \implies \\
& w > \frac{e}{2\phi} \frac{1+\epsilon^*}{1-\epsilon^*}.
\end{split}
\end{equation}

Under the condition above, we can see that $\delta_n < \delta_p$. In fact:

\begin{equation}
\begin{split}
w > \frac{e}{2\phi} \frac{1+\epsilon^*}{1-\epsilon^*} &\implies \left(\frac{1}{2 \phi w}\frac{1+\epsilon^*}{1-\epsilon^*}\right) < \frac{1}{e}\\
&\implies \delta_g + \left(\frac{1}{2 \phi w}\frac{1+\epsilon^*}{1-\epsilon^*}\right)^d (1 - \delta_g) < \delta_g + \left(\frac{1}{e}\right)^d (1 - \delta_g)\\
&\implies \delta_n < \delta_p
\end{split}
\end{equation}

We are now ready to state the following theorem, which follows straight from previous results and reasoning, and subsumes theorems \ref{thm_falsenegative} and \ref{thm_falsepositive}.

\begin{thm}
	\label{thm_correctness}
	Given an input stream $\mathcal{N}$, with total decayed count $C$, distributed among $p$ peers, each one making use of an augmented sketch data structure of size $d \times w$, and given a threshold parameter $0<\phi<1$, after $r$ rounds of \textsc{P2PTFHH}, on any peer the \textsc{QUERY} procedure returns a set $H$ of items and their corresponding time--faded frequencies so that, with probability greater than $1-\delta$:
	\begin{enumerate}
		\item $H$ includes all of the true time--faded heavy hitters ($f>\phi C$);
		\item $H$ does not include any item in $\mathcal{N}$ with frequency $f \le (\phi - \epsilon) C$;
	\end{enumerate}
	with a false positives tolerance $\epsilon = \frac{4\epsilon^* \phi}{(1+\epsilon^*)^2} + \frac{e}{2w}\frac{1-\epsilon^*}{1+\epsilon^*}$,\\
	and a probability of failure $\delta = \delta_g + e^{-d}(1-\delta_g)$, $\delta_g$ being the probability of failure pertaining to the averaging gossip protocol.\\
\end{thm}

\subsection{Error bounds on the estimated frequency}
The estimated frequencies provided by $\mathcal{D}_{r,l}$ refer to average frequencies. Thus, in order to obtain an estimate of the global frequency $f_i$ of an item $i$, we need to multiply $\hat{f}_{r,l,i}$ by $\tilde{p}_{r,l}$. From eqs. \eqref{eq_properties1}--\eqref{eq_properties4} and \eqref{eq_properties4b} we can compute the corresponding error bounds. The following theorem holds.

\begin{thm}
	Given an input stream $\mathcal{N}$ with total decayed count $C$, distributed among $p$ peers, each one making use of an augmented sketch data structure of size $d \times w$, and given a threshold parameter $0<\phi<1$, after $r$ rounds of \textsc{P2PTFHH}, any peer can report a time--faded estimated frequency $f^s_{r,l,i}$ of an item $i \in [m]$ so that, with probability $1-\delta$:
	
	\begin{equation}
	\label{est_error}
	\frac{1-\epsilon^*}{1+\epsilon^*} f_i < f^s_{r,l,i} < \frac{1+\epsilon^*}{1-\epsilon^*} \left(f_i + \frac{eC}{2w}\right),
	\end{equation}
	
	\noindent with $\delta = \delta_g + e^{-d}(1-\delta_g)$, $\delta_g$ being the probability of failure pertaining to the gossip approximation.
\end{thm}

\begin{proof}
	From equation \eqref{eq_properties1}, with probability greater than $1-\delta_g$:
	
	\begin{equation}
	\label{eq_properties5}
	\frac{1}{1+\epsilon^*} < \frac{\tilde{p}_{r,l}}{p} < \frac{1}{1-\epsilon^*} 
	\end{equation}
	
	\noindent and from equation \eqref{eq_properties2} and equation \eqref{eq_properties3}, with probability greater than $1-\delta_g$:
	
	\begin{equation}
	\label{eq_properties6}
	\begin{split}
	f_i \frac{\tilde{p}_{r,l}}{p} (1 - \epsilon^*) < \tilde{f}_{r,l,i}  \tilde{p}_{r,l} < f_i \frac{\tilde{p}_{r,l}}{p} (1 + \epsilon^*), \\
	C \frac{\tilde{p}_{r,l}}{p} (1 - \epsilon^*) < \tilde{C}_{r,l} \tilde{p}_{r,l} < C \frac{\tilde{p}_{r,l}}{p} (1 + \epsilon^*). \\
	\end{split}
	\end{equation}
	
	\noindent Now, starting from eqs. \eqref{eq_properties4}, \eqref{eq_properties4b}  and taking into account eqs. \eqref{eq_properties5}, \eqref{eq_properties6}, with probability greater than $(1-\delta_g)(1-e^{-d})$, it follows that:
	
	\begin{equation}
	\begin{split}
	\tilde{f}_{r,l,i} \tilde{p}_{r,l} \le &\hat{f}_{r,l,i} \tilde{p}_{r,l} \le \tilde{f}_{r,l,i} \tilde{p}_{r,l} + \frac{e\tilde{C}_{r,l}}{2w} \tilde{p}_{r,l} \implies \\
	f_i \frac{\tilde{p}_{r,l}}{p} (1 - \epsilon^*) < &\hat{f}_{r,l,i} \tilde{p}_{r,l} < \left(f_i + \frac{eC}{2w}\right) \frac{\tilde{p}_{r,l}}{p} (1 + \epsilon^*) \implies \\
	\frac{1-\epsilon^*}{1+\epsilon^*} f_i < &\hat{f}_{r,l,i} \tilde{p}_{r,l}  < \frac{1+\epsilon^*}{1-\epsilon^*} \left(f_i + \frac{eC}{2w}\right).
	\end{split}
	\end{equation}
	
	\noindent Eventually, setting $f^s_{r,l,i} = \hat{f}_{r,l,i} \tilde{p}_{r,l}$, the relation \eqref{est_error} follows.
\end{proof}

\subsection{Practical considerations}
\label{prat-cons}
We conclude this Section discussing how to select proper values for the parameters $d$, $w$ and $R$, which represent respectively the number of rows and columns of the augmented sketch data structure and the minimum number of rounds required to solve the Approximate Time--Faded Heavy Hitters Problem with a  given threshold $\phi$, tolerance $\epsilon$ and probability of failure $\delta$. Theorem~\ref{thm_correctness} proves the correctness of the \textsc{P2PTFHH} algorithm, also providing the relation which ties the value of $\epsilon$ to the parameters of the algorithm. 

Let us consider as fixed the values of $p^*$ and $\delta_g$, then, in order to guarantee that the desired false positives tolerance $\epsilon$ is not exceeded, the user can set the number of rounds $R$ and/or the number of sketch columns $w$.
%Fixing a given tolerance $\epsilon$, the user has one degree of freedom; Figure~\ref{plot_w_r} plots the relationship between the values for $R$ and $w$ which produce a given tolerance $\epsilon$. 
By fixing the value of $\epsilon$, the relation between $w$ and $R$ is given by equation \eqref{eq_w_r}.

\begin{equation}
\label{eq_w_r}
w=\frac{e(1-\epsilon^{*^2})}{2\epsilon \left(1+\epsilon^*\right)^2-8 \phi \epsilon^*}=\frac{e(1-p^{*^2} \frac{\gamma^R}{\delta_g})}{2\epsilon \left(1+p^* \sqrt{\frac{\gamma^R}{\delta_g}}\right)^2-8 \phi p^* \sqrt{\frac{\gamma^R}{\delta_g}}}
\end{equation}

\begin{figure}[]
	\centering
	\caption{Relationship between the number of sketch columns and the number of rounds to guarantee a given level of false positive tolerance $\epsilon$.} 
	\label{plot_w_r}
\end{figure}

Among all of the possible values for $R$ and $w$, the user could follow a strategy oriented to maintain the number of rounds (hence the running time) as fewer as possible and to choose $w$ accordingly or vice-versa to maintain the number of counters (hence the space) as lower as possible and to choose $R$ accordingly. Let us now discuss both strategies.

With the first strategy, which can be called \textit{time-dominant}, the user is interested on choosing $R$ and $w$ which guarantee a given $\epsilon$ such that $R$ is minimum. Equation \eqref{eq_w_r} reveals that $R$ is a monotone decreasing function with $w$, hence the minimum value for $R$ is obtained when $w$ tends to infinity; moreover, since $w>0$ the minimum value for $R$ can be computed by imposing the following constraint:

\begin{equation}
2\epsilon \left(p^* \sqrt{\frac{\gamma^R}{\delta_g}}+1\right)^2-8 \phi p^* \sqrt{\frac{\gamma^R}{\delta_g}} > 0,
\end{equation}

\noindent from which it follows that

\begin{equation}
R> \frac{ \log {\delta_g}+2 \log \left( \frac{2 \phi -\epsilon - 2 \sqrt{\phi^2-\epsilon \phi}}{\epsilon p^*} \right)}{\log \gamma}.
\end{equation}

\noindent Since $R$ is an integer, the minimum value of $R$ is given by:

\begin{equation}
\label{eq_min_r}
R_{min}=\left\lfloor \frac{ \log {\delta_g}+2 \log \left( \frac{2 \phi -\epsilon - 2 \sqrt{\phi^2-\epsilon \phi}}{\epsilon p^*} \right)}{\log \gamma}  \right\rfloor +1.
\end{equation}

\noindent Substituting the value of $R_{min}$ provided by equation \eqref{eq_min_r} into equation \eqref{eq_w_r} for $R$, it is possible to obtain the value for $w$.

With the second strategy, which can be called \textit{space-dominant}, the user wants to keep the memory footprint as low as possible. Equation \eqref{eq_w_r} reveals that $w$ is a monotone decreasing function with $R$, hence the minimum value for $w$ is obtained when $R$ tends to infinity. Evaluating equation \eqref{eq_w_r} for $R \rightarrow \infty$ it holds that the minimum value for $w$ is given by:

\begin{equation}
\label{eq_min_w_exact}
w>\frac{e}{2\epsilon}.
\end{equation}

\noindent Considering that $w$ is an integer value,

\begin{equation}
\label{eq_min_w}
w_{min}= \left \lfloor \frac{e}{2\epsilon} \right \rfloor + 1.
\end{equation}

\noindent solving equation \eqref{eq_w_r} by $\epsilon^*$ and using equation \eqref{eq_min_w} it holds that:

\begin{equation}
\epsilon^* = \frac{2w_{min}(2\phi -\epsilon) - \sqrt{16 \phi w^2_{min} (\phi -\epsilon)+ e^2}}{e + 2\epsilon w_{min}}.
\end{equation}

\noindent Since $\epsilon^* = p^* \sqrt{\frac{\gamma^R}{\delta_g}}$, it holds that:

\begin{equation}
R = \frac{1}{\log \gamma} \left( 2\log \epsilon^* - 2 \log p^* + \log \delta_g \right).
\end{equation}

Since $R$ is an integer, 

\begin{equation}
R = \left \lfloor \frac{1}{\log \gamma} \left( 2\log \epsilon^* - 2 \log p^* + \log \delta_g \right) \right \rfloor + 1.
\end{equation}

Finally, we discuss how the parameters of the algorithm influence the probability of failure $\delta$. \textsc{P2PTFHH} is a probabilistic algorithm in which two sources of randomness occur: the first source of randomness is due to the way the information is spread over the peers at each round of the gossip protocol and is represented by the probability $\delta_g$; the second source of randomness is related to the mapping of an item to a cell of the augmented sketch data structure through the corresponding hash functions. The overall probability of failure of \textsc{P2PTFHH} in solving the Problem \ref{prob1} with the desired guarantees can be controlled and made as small as desired, by acting on the parameter $d$, i.e. the number of rows of sketches, and the value of $\delta_g$, that in turn depends on the value of $R$.

Theorem \ref{thm_correctness} states that the Approximated Time--Faded Heavy Hitters Problem can be solved without false negative items and with a bounded number of false positive items with probability of failure less than $\delta = \delta_g + e^{-d}(1-\delta_g)$. Thus, given a desired value of $\delta$, we have to choose $d$ such that:

\begin{equation}
d \geq \log {\frac{1-\delta_g}{\delta - \delta_g}}.
\end{equation}

%When the number of rounds of the gossip protocol does not satisfy equation \eqref{eq_condition_w}, then the failure probability is driven by $\delta_n$, hence equation \eqref{eq_max_delta} becomes:
%
%\begin{equation}
%\label{eq_rt_2}
%\begin{split}
%\delta_n & \leq \delta \implies \\
%\delta_g + \left(\frac{1}{2 \phi w}\frac{1+\epsilon^*}{1-\epsilon^*}\right)^d(1-\delta_g) & \leq \delta \implies \\
%d & \geq \frac{\log {(1-\delta_g)}-\log{(\delta - \delta_g)}}{\log{2 \phi w}+\log{(1-\epsilon^*)}-\log{(1+\epsilon^*)}}.
%\end{split}	
%\end{equation}
%
%Finally, considering that $w>\frac{e}{2\epsilon}$, it holds that:
%
%\begin{equation}
%d \geq \frac{\log {(1-\delta_g)}-\log{(\delta - \delta_g)}}{1+ \log{\phi} -\log{\epsilon}+\log{(1-\epsilon^*)}-\log{(1+\epsilon^*)}}.
%\end{equation}

\section{Experimental results}
\label{results}
\label{results}
We implemented a simulator in C++ using the igraph library \cite{libigraph} in order to test the performance of our algorithm. The simulator has been compiled using the GNU C++ compiler g++ 4.8.5 on CentOS Linux 7. The tests have been performed on a machine equipped with two hexa-core Intel Xeon-E5 2620 CPUs at 2.0 GHz and 64 GB of main memory.
The source code of the simulator is freely available for inspection and for reproducibility of results contacting the authors by email. 

Objective of the experiments was to provide evidence of \textsc{P2PTFHH}'s behavior in various conditions, through the use of some useful and commonly referred metrics: \textit{Recall}, \textit{Precision}, and \textit{frequency estimation Error}. Recall is defined as the fraction of frequent items retrieved by an algorithm over the total number of frequent items.  Precision is the fraction of frequent items retrieved over the total number of items reported as frequent items candidates. The frequency estimation error that we measure and report is the \textit{Average Relative Error}, ARE, committed in the frequency estimation of items outputted by the algorithm as frequent candidates. Relative Error is defined as usual as $\frac{\left| f^S-f \right|}{f}$, where $f^S$ is the frequency reported for an item and $f$ is its true frequency.

For every experiment, a global input stream of items has been generated following a Zipfian distribution  (items are 32 bits unsigned integers) and each peer has been assigned a distinct part of that global stream, thus simulating the scenario in which each peer processes, independently of the other peers, its own local sub-stream.
The peers collaboratively discover the frequent items of the union of their sub-streams. 
The experiments have been repeated 10 times setting each time a different seed for the pseudo-random number generator used for creating the input data. For each execution, we collected the peers' statistics discussed above. Then, with reference to each peer, we determined the average value of those statistics over the ten executions. At last, we computed the mean and confidence interval of each statistics over all of the peers and plotted these values.

We fixed the size of the global stream at $100$ millions of items, and varied each one of the other parameters: the skew of the Zipfian distribution, $\rho$, the number of peers, $p$, the frequent items threshold, $\phi$, the width of the sketch used by each peer, $sw$, (while the depth is fixed to $4$) , and the fan-out, $fo$, setting non varying parameters to default values. Every experiment has been carried out by generating random P2P network topologies through both the Barabasi-Albert and the Erdos-Renyi random graphs models. Table~\ref{experiments} reports the sets of values (first row) and the default values (second row) used for each parameter. All of the plots representing the evolution of Recall and Precision are on linear-linear scale, whilst the plots that show the behaviour of the relative error are on linear-logarithmic scale.

\begin{table*}
	\renewcommand{\arraystretch}{1.3}
	\caption{Experiment values}
	\label{experiments}
	\centering
	\small
	\begin{tabular}{| c | c | c | c | c | c |}
		\hline
		$\boldsymbol{\rho}$ &  $\boldsymbol{\phi}$ & $\boldsymbol{p\  (\times10^3)}$  & $\boldsymbol{sw\  (\times10^3)}$ & $\boldsymbol{r}$ & $\boldsymbol{fo}$\\
		\hline
		\{0.9, 1.1, 1.3, 1.5\} & \{0.01, 0.02, 0.03, 0.04\}  & \{1, 5, 10, 15\} &  \{1.5, 2.5, 3.5, 4.5\} &  \{21, 24, 27\} & \{1, 2, 3\}\\ 
		\hline
		1,2 & 0.02 & 5 &  2.5 &  24 &  1\\ \hline
	\end{tabular}
\end{table*}

\begin{figure*}[h]
	\centering
	\begin{tabular}{ccc}		
		\subfloat[Recall]{
			\includegraphics[width=0.3\textwidth]{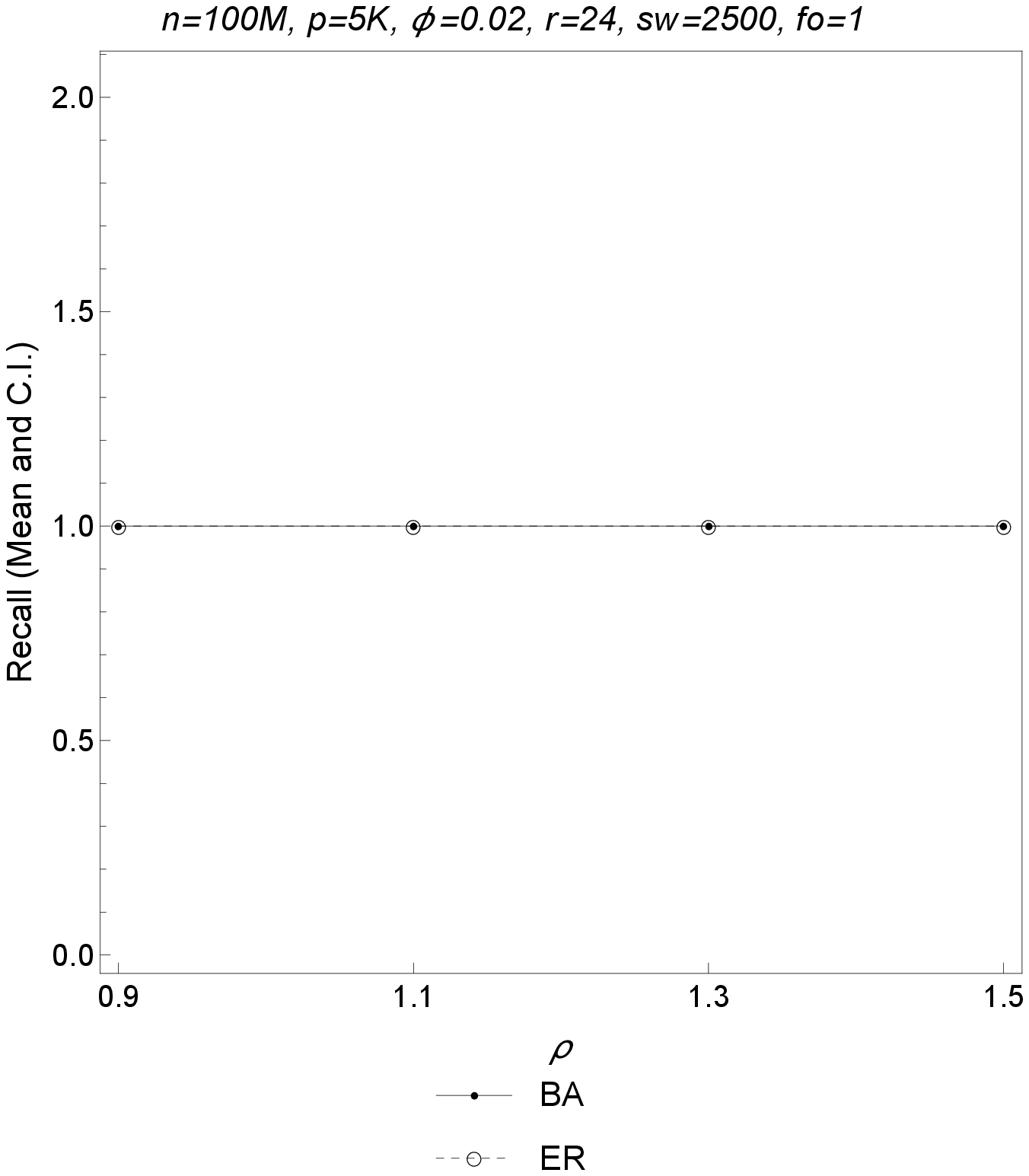}
			\label{sk-rec}
		} &
		
		\subfloat[Precision]{
			\includegraphics[width=0.3\textwidth]{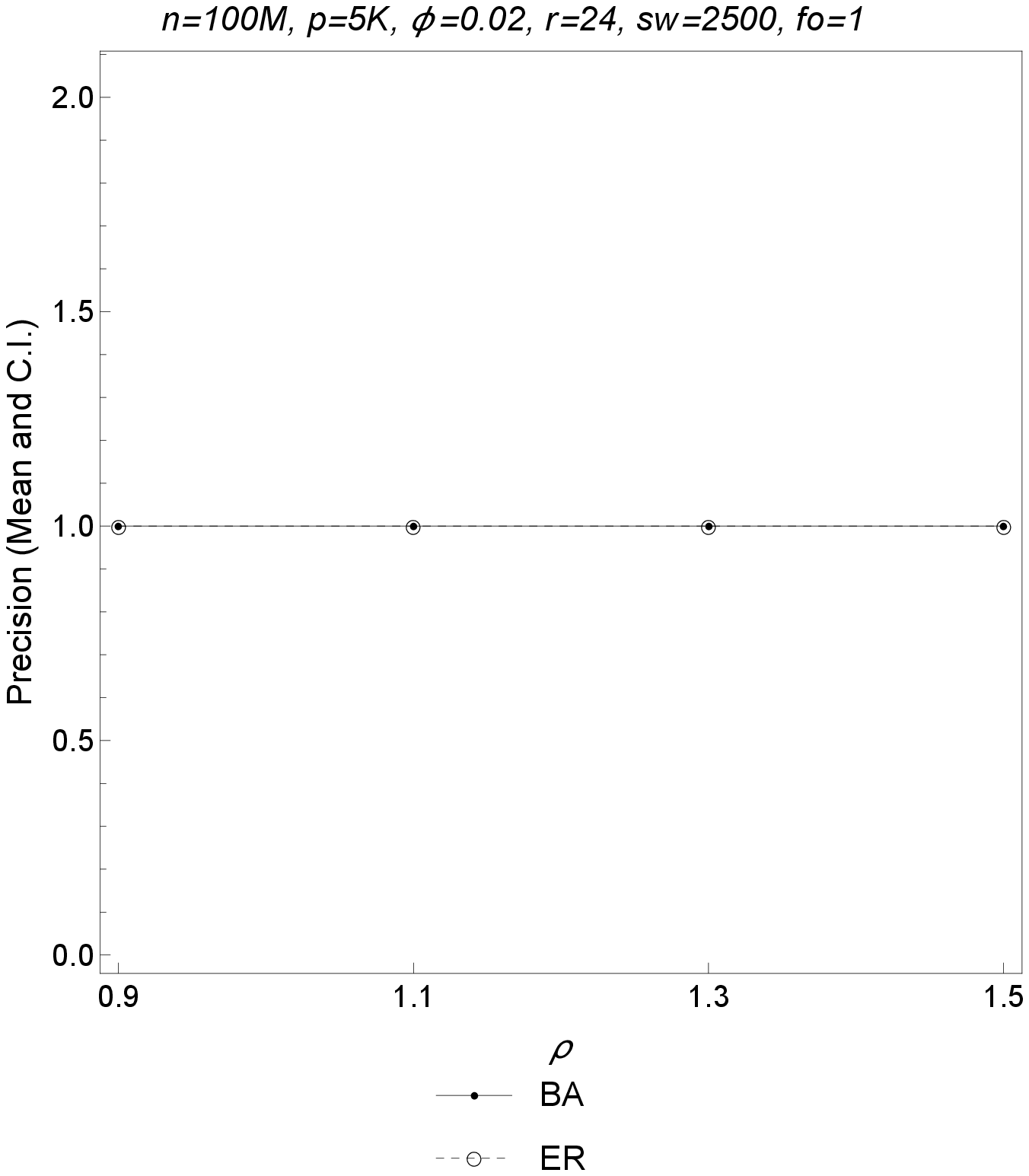}
			\label{sk-prec}
		} &
		
		\subfloat[Average Relative Error ]{
			\includegraphics[width=0.3\textwidth]{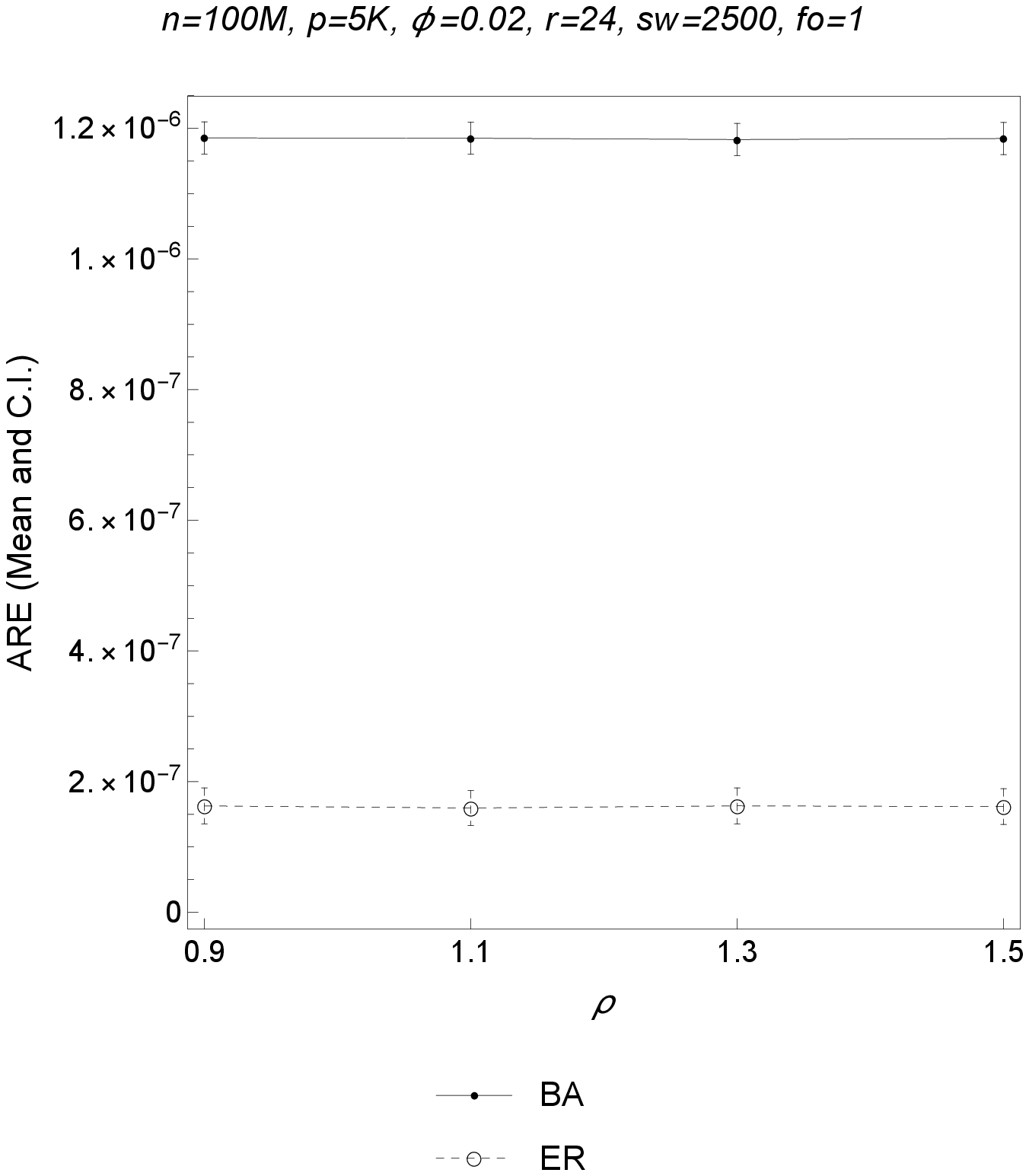}
			\label{sk-are}
		} 
	\end{tabular}
	
	\caption{Recall, Precision and Average Relative Error (mean and confidence interval) varying the skewness of the input distribution, for both a Barabasi-Albert (BA) and an Erdos-Renyi (ER) type of network graph.} 
	\label{skew_plot}
\end{figure*}

Figures from \ref{skew_plot} to \ref{r_plot} report the results of experiments carried out without taking into account peer's churning.  The Recall, Precision and the Average Relative Error measured when varying the skewness of the input disribution are respectively shown by Fig.~\ref{sk-rec}, Fig.~\ref{sk-prec} and Fig.~\ref{sk-are}. Recall and Precision are always $100\%$, i.e. the algorithm is very robust with regard to skewness variations in the input. Moreover, we observe very low Average Relative Errors on frequency estimation, both in Barabasi-Albert graph (the worst) and Erdos-Renyi case (the better).

\begin{figure*}[h]
	\centering
	\begin{tabular}{ccc}		
		\subfloat[Recall]{
			\includegraphics[width=0.3\textwidth]{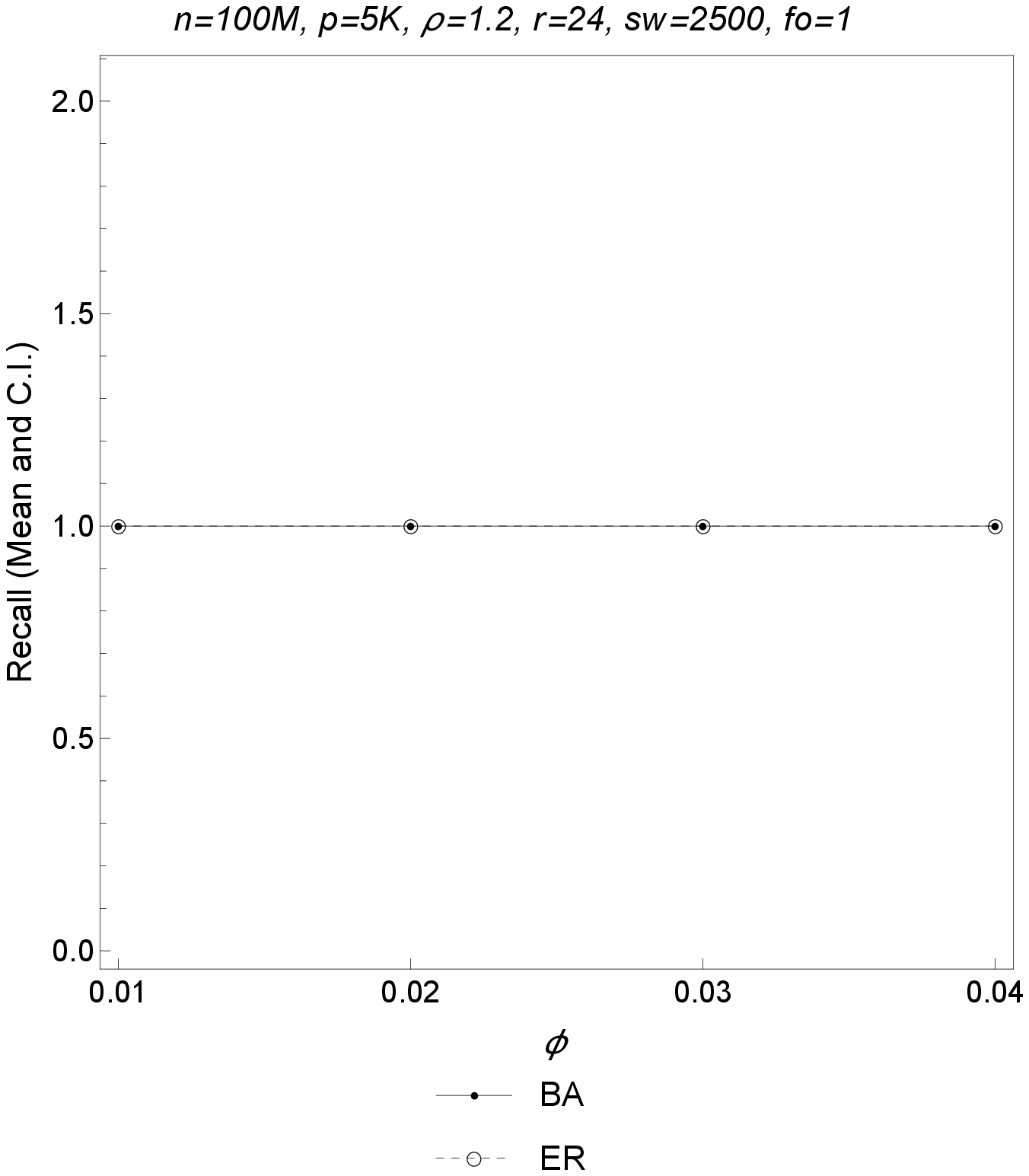}
			\label{phi-rec}
		} &
		
		\subfloat[Precision]{
			\includegraphics[width=0.3\textwidth]{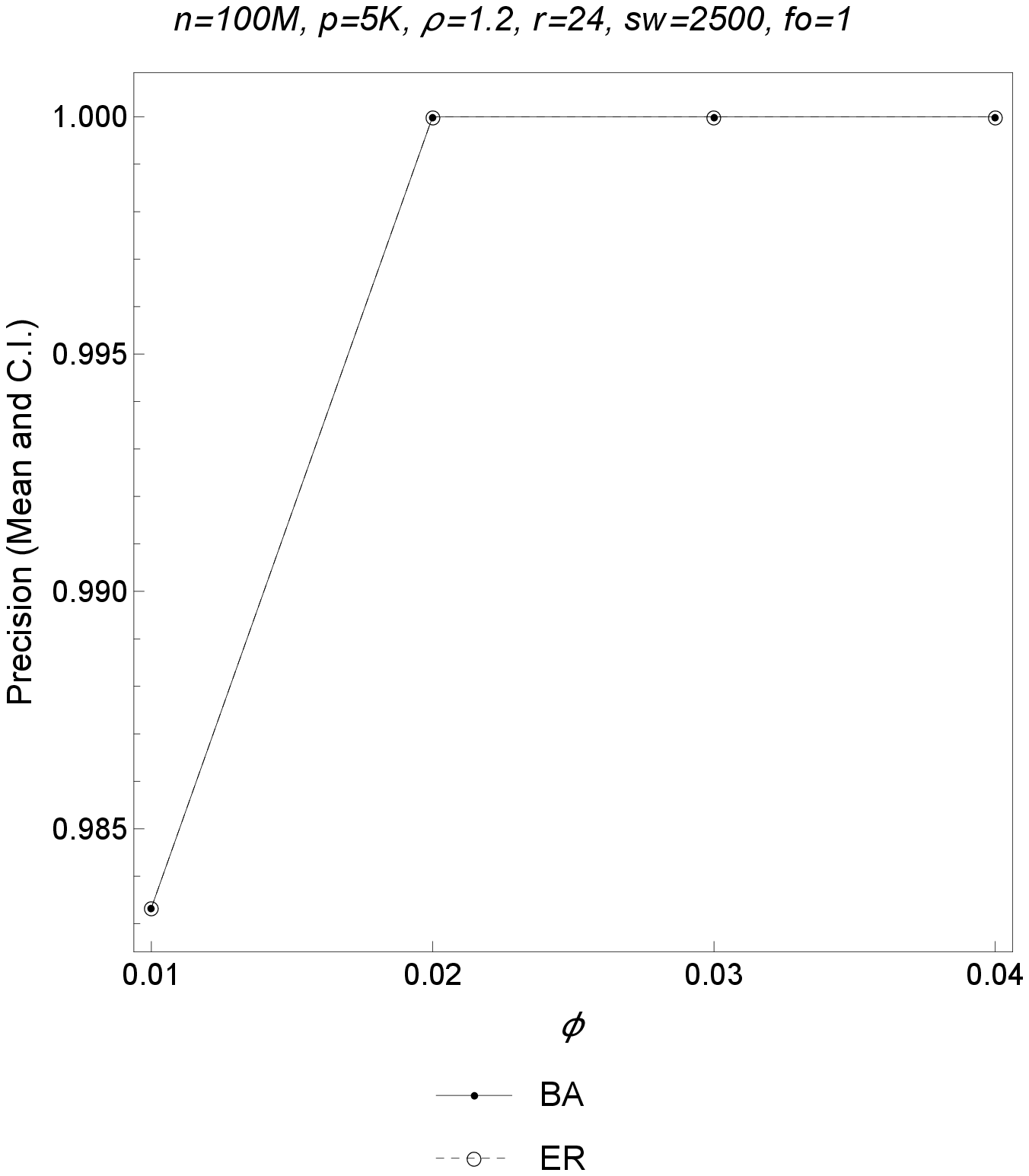}
			\label{phi-prec}
		} &
		
		\subfloat[Average Relative Error ]{
			\includegraphics[width=0.3\textwidth]{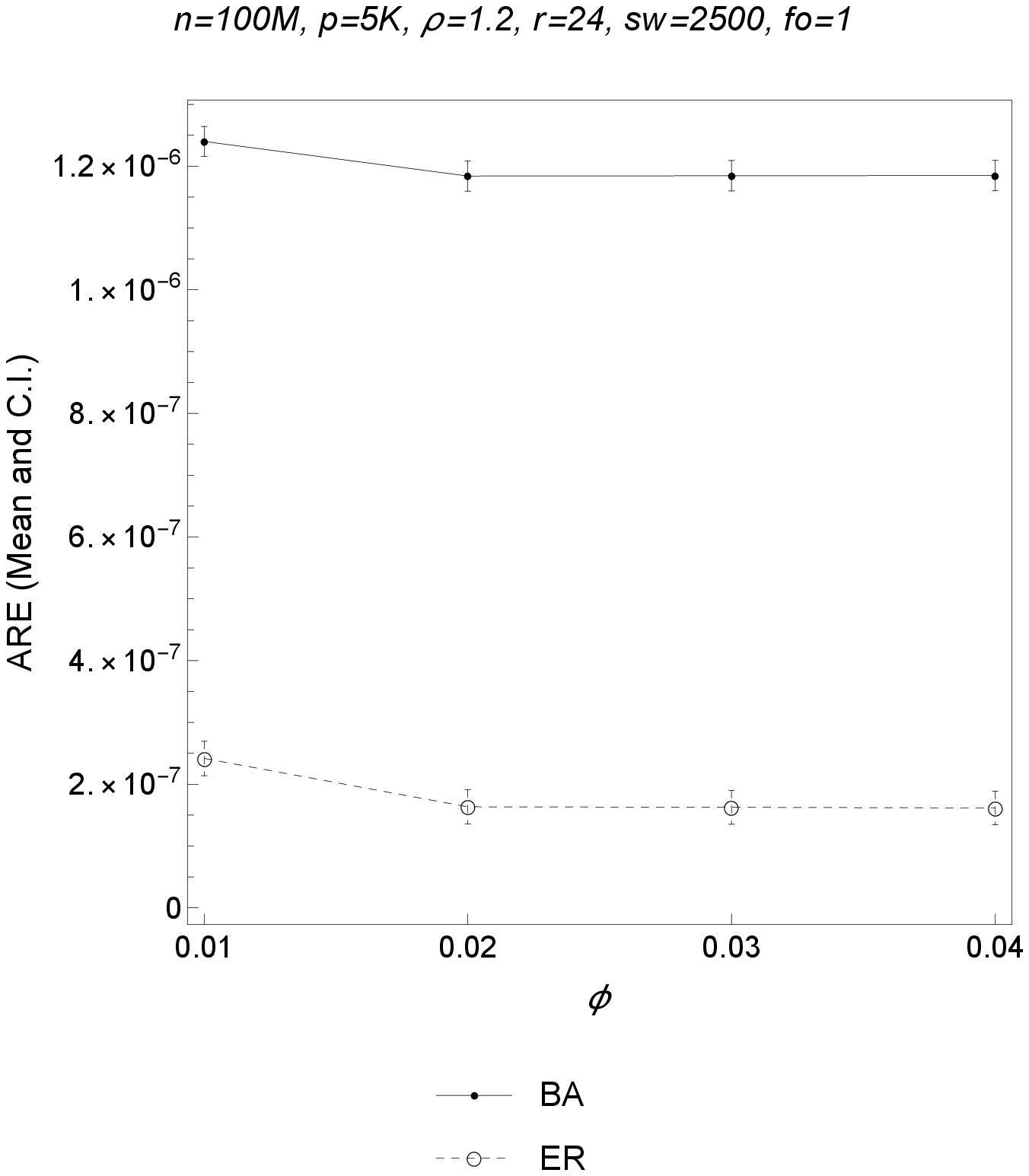}
			\label{phi-are}
		} 
		
	\end{tabular}
	
	\caption{Recall, Precision and Average Relative Error (mean and confidence interval) varying the frequent items threshold $\phi$, for both a Barabasi-Albert (BA) and an Erdos-Renyi (ER) type of network graph.} 
	\label{phi_plot}
\end{figure*}

Figure~\ref{phi_plot} shows how  \textsc{P2PTFHH} behaves with regard  to variations of the threshold $\phi$. Also in this setting, Recall (Fig.~\ref{phi-rec}) is always $100\%$, Precision (Fig.~\ref{phi-prec}), on the other hand, presents a value slightly below $100\%$ only when $\phi = 0.01$. Average Relative Errors (Fig.~\ref{phi-are}) confirm what already told about the skewness plots.

\begin{figure*}[h]
	\centering
	\begin{tabular}{ccc}		
		\subfloat[Recall]{
			\includegraphics[width=0.285\textwidth]{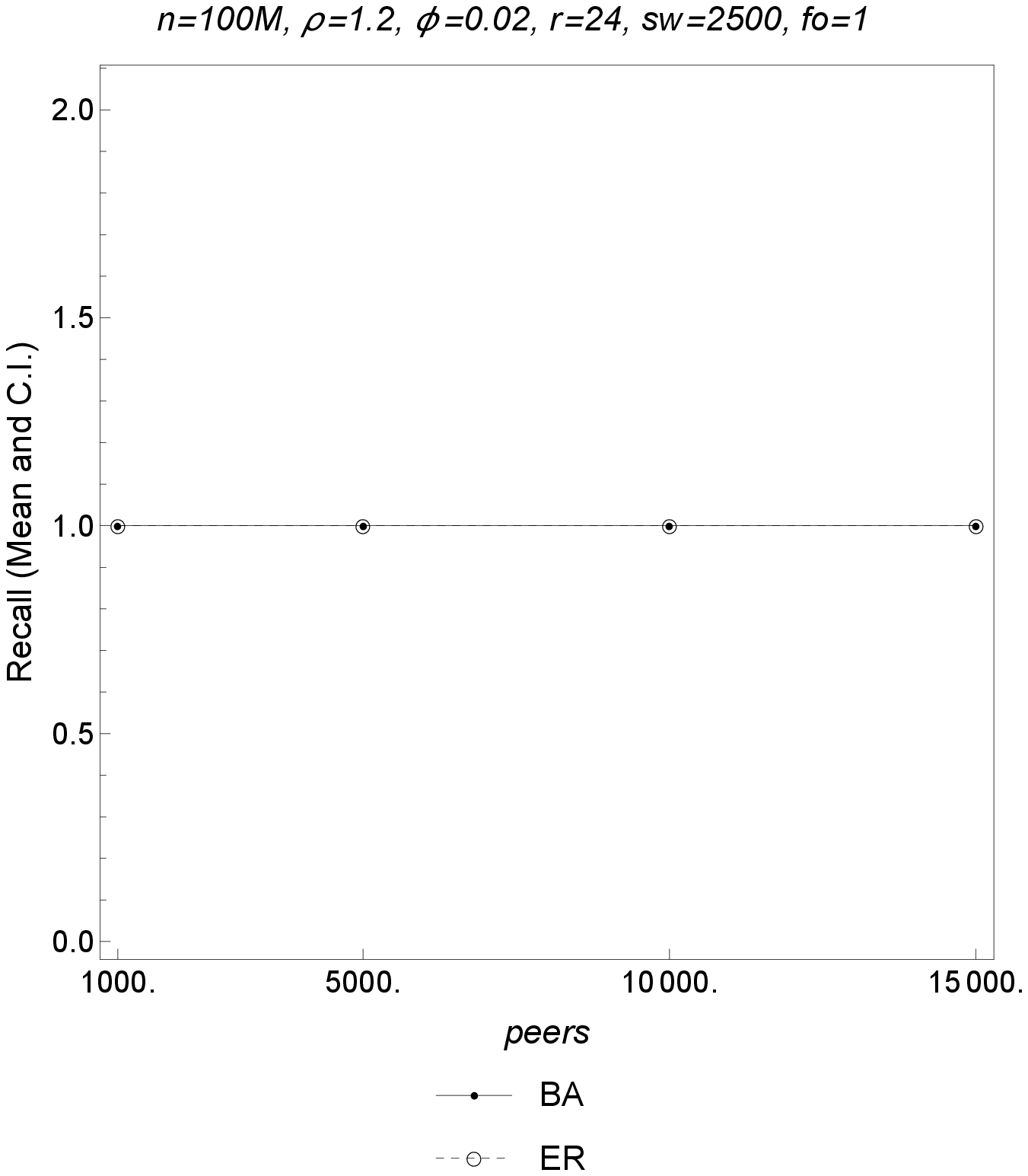}
			\label{p-rec}
		} &
		
		\subfloat[Precision]{
			\includegraphics[width=0.285\textwidth]{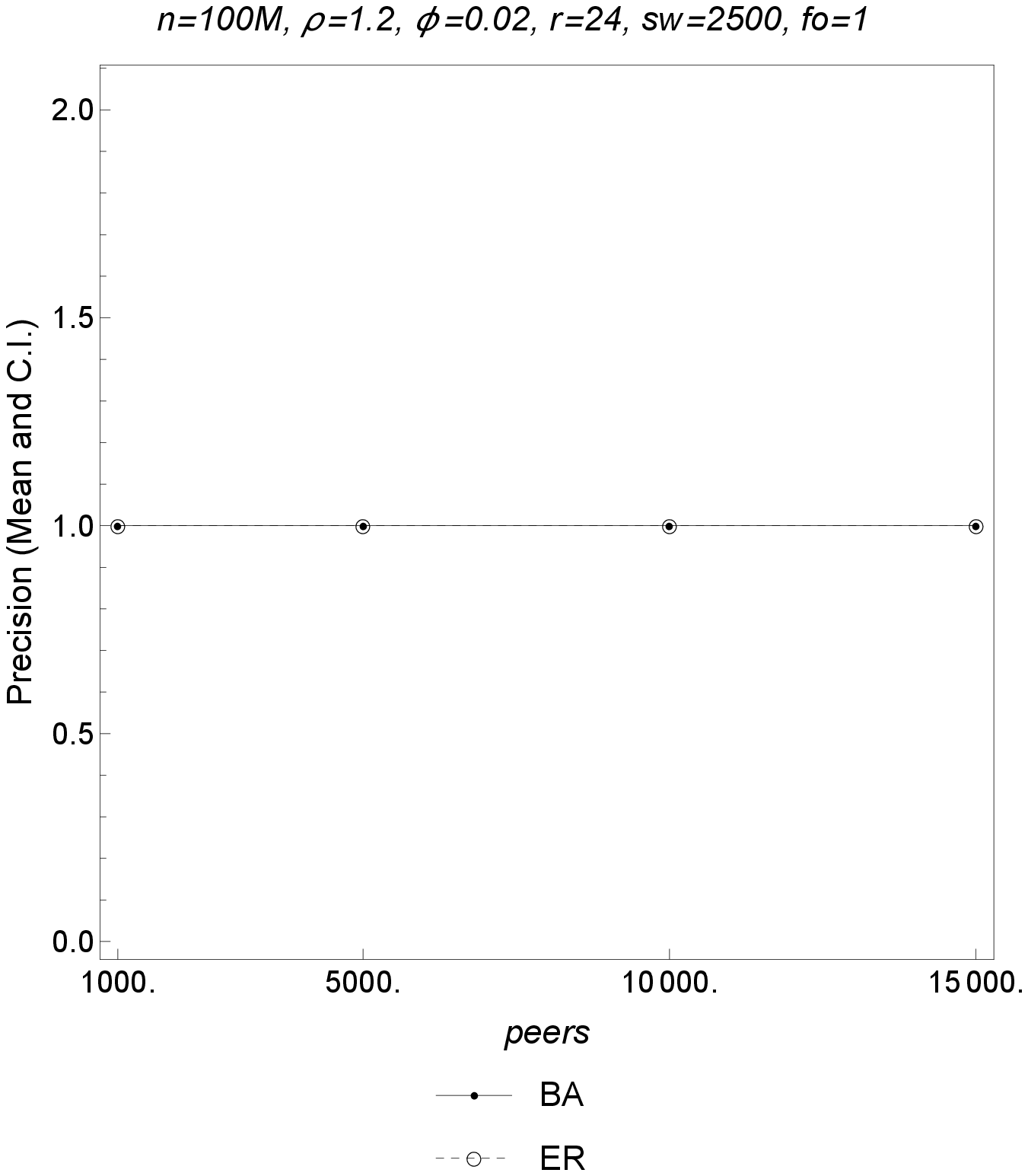}
			\label{p-prec}
		} &
		
		\subfloat[Average Relative Error ]{
			\includegraphics[width=0.33\textwidth]{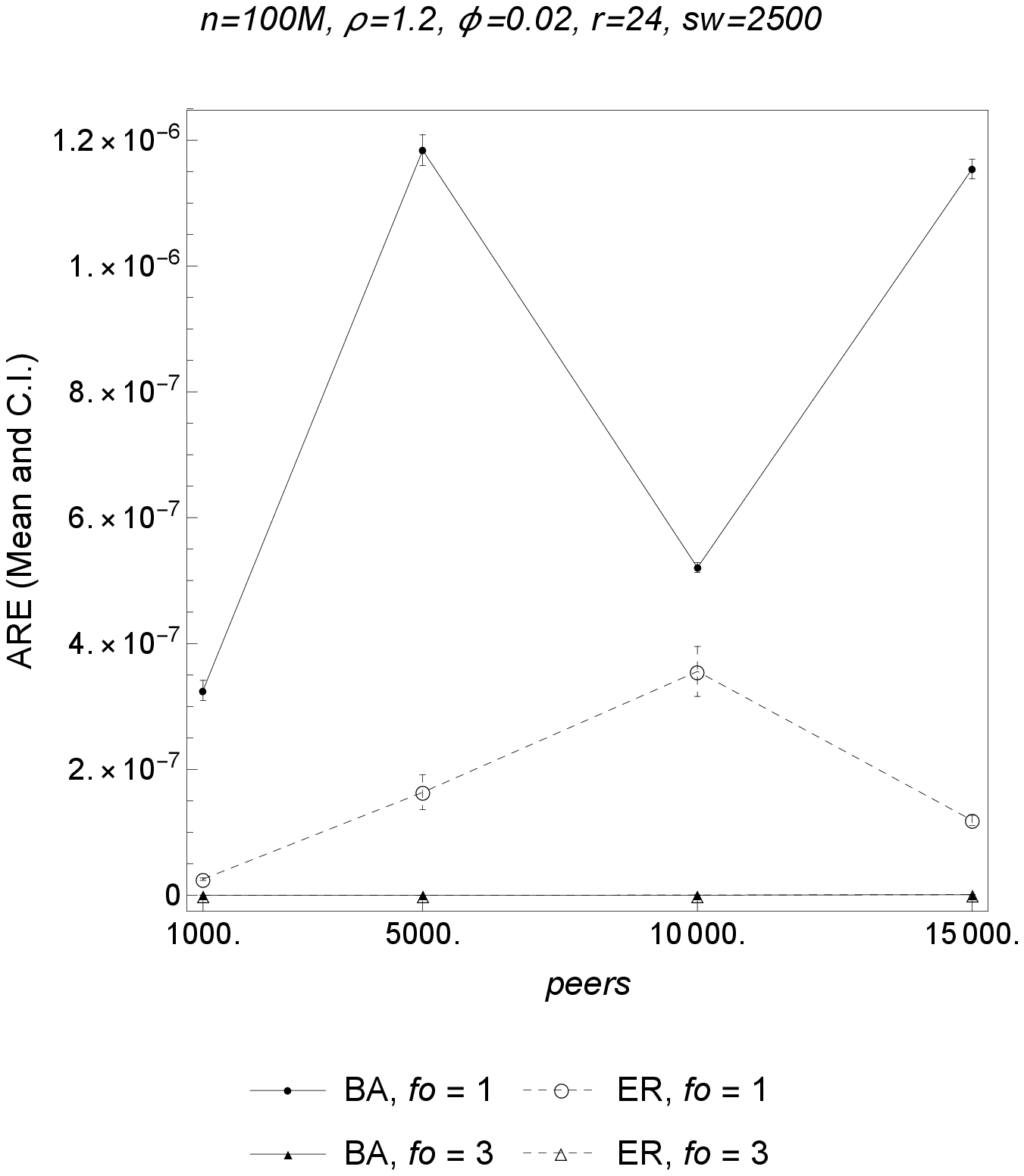}
			\label{p-are}
		}
		
	\end{tabular}
	
	\caption{Recall, Precision and Average Relative Error (mean and confidence interval) varying the number of peers participating in the computation,  for both a Barabasi-Albert (BA) and an Erdos-Renyi (ER) type of network graph, and setting a fan-out $fo$ equal to 1 and 2 in case of the ARE plot.} 
	\label{peers_plot}
\end{figure*}

Figure~\ref{peers_plot} plots the results of experiments where we varied the number of peers. Recall and Precision stay at the max value, whilst the Average Relative Error (Fig.~\ref{p-are}) slightly increases when the number of peers grows. Figure~\ref{p-are} also shows that the error decreases when the fan out increases from $1$ to $3$. This behavior is expected, with a greater fan out the number of messages exchanged in a round increases and information spreads quickly. If the number of rounds executed and the dimension of the sketches used by the peers is not adequate to the number of peers in the network, the metrics tend to get worse.

\begin{figure*}[h]
	\centering
	\begin{tabular}{ccc}		
		\subfloat[Recall]{
			\includegraphics[width=0.3\textwidth]{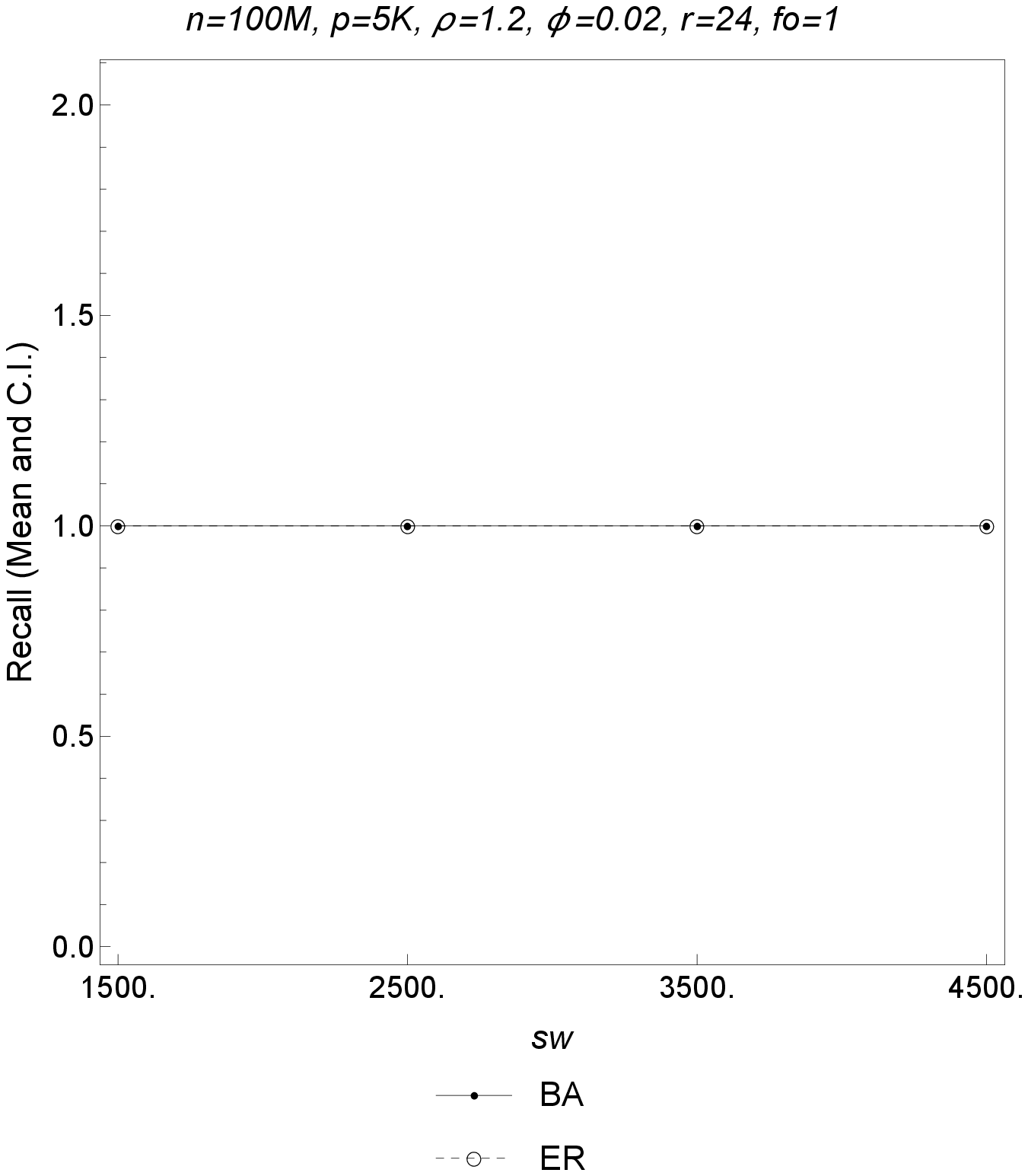}
			\label{sw-rec}
		} &
		
		\subfloat[Precision]{
			\includegraphics[width=0.3\textwidth]{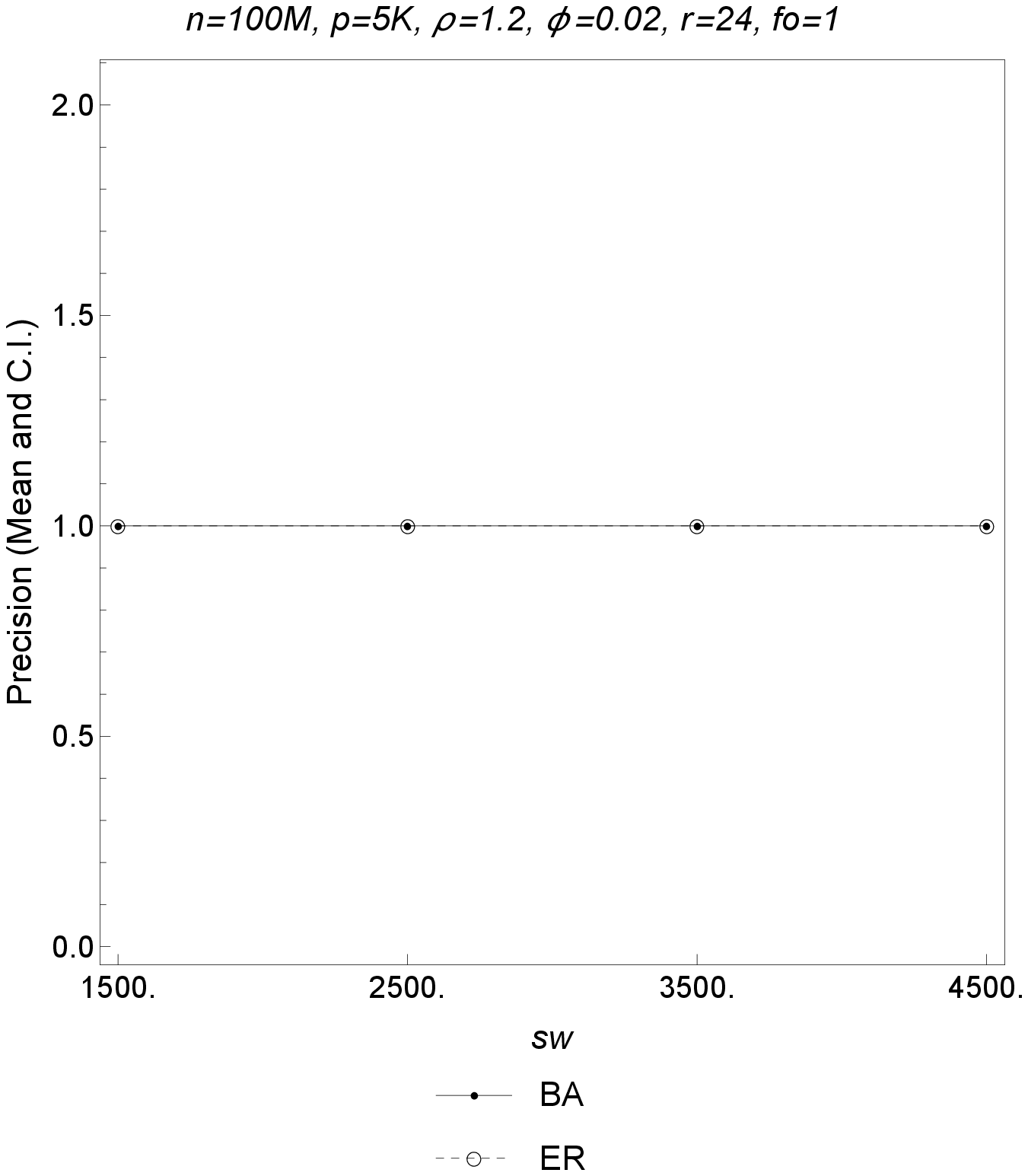}
			\label{sw-prec}
		} &
		
		\subfloat[Average Relative Error ]{
			\includegraphics[width=0.3\textwidth]{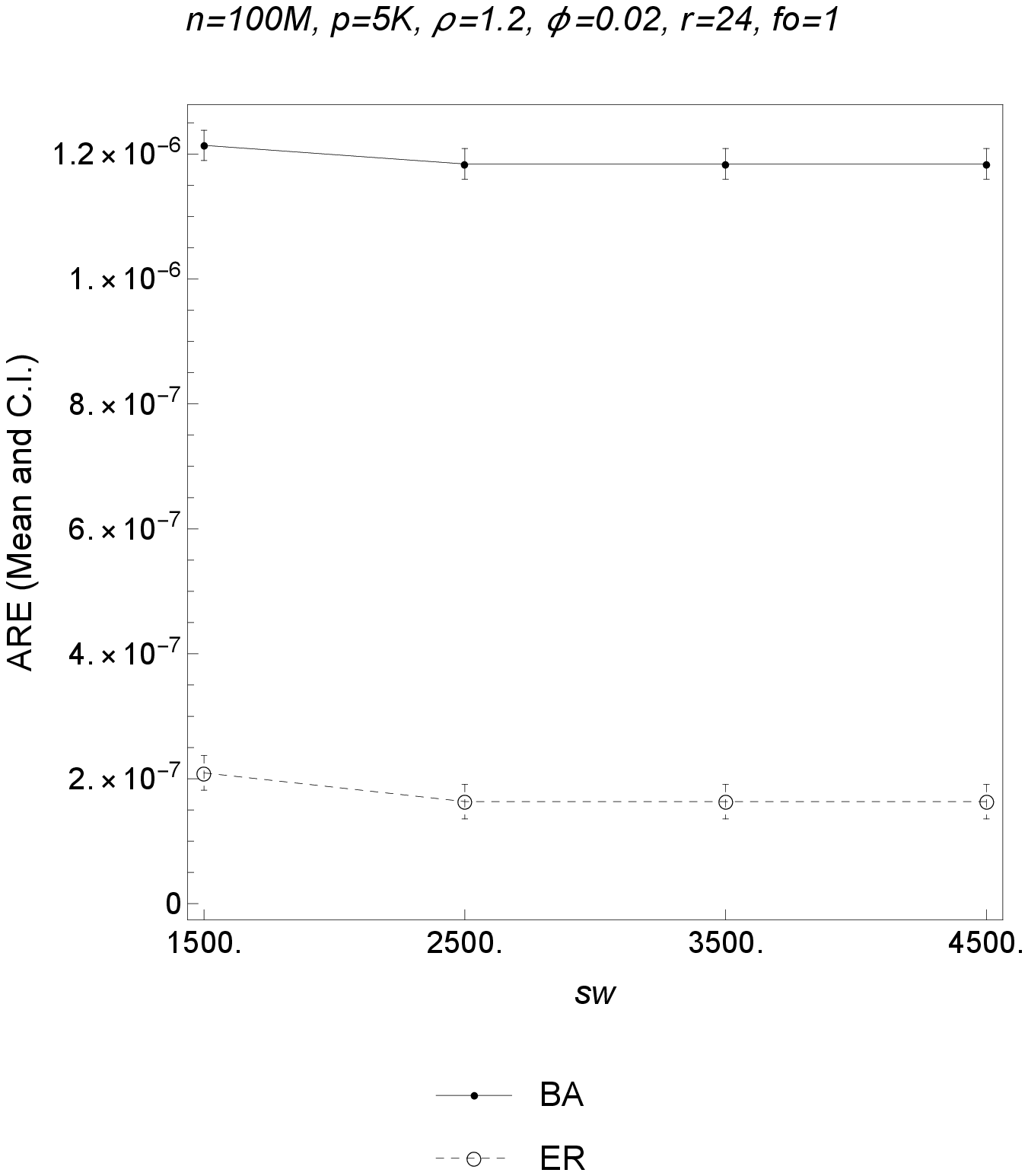}
			\label{sw-are}
		} 
	\end{tabular}
	
	\caption{Recall, Precision and Average Relative Error (mean and confidence interval) varying the number of Space-Saving counters used by each peer,  for both a Barabasi-Albert (BA) and an Erdos-Renyi (ER) type of network graph.} 
	\label{sw_plot}
\end{figure*}

The plots related to the experiments in which we varied the width $sw$ of peers' sketches (Figure~\ref{sw_plot}) do not present a particular behavior in the interval of values tested, showing that in this case the sketch size used was always enough with regard to the number of rounds executed in order to guarantee a good accuracy.

\begin{figure*}[h]
	\centering
	\begin{tabular}{ccc}		
		\subfloat[Recall]{
			\includegraphics[width=0.3\textwidth]{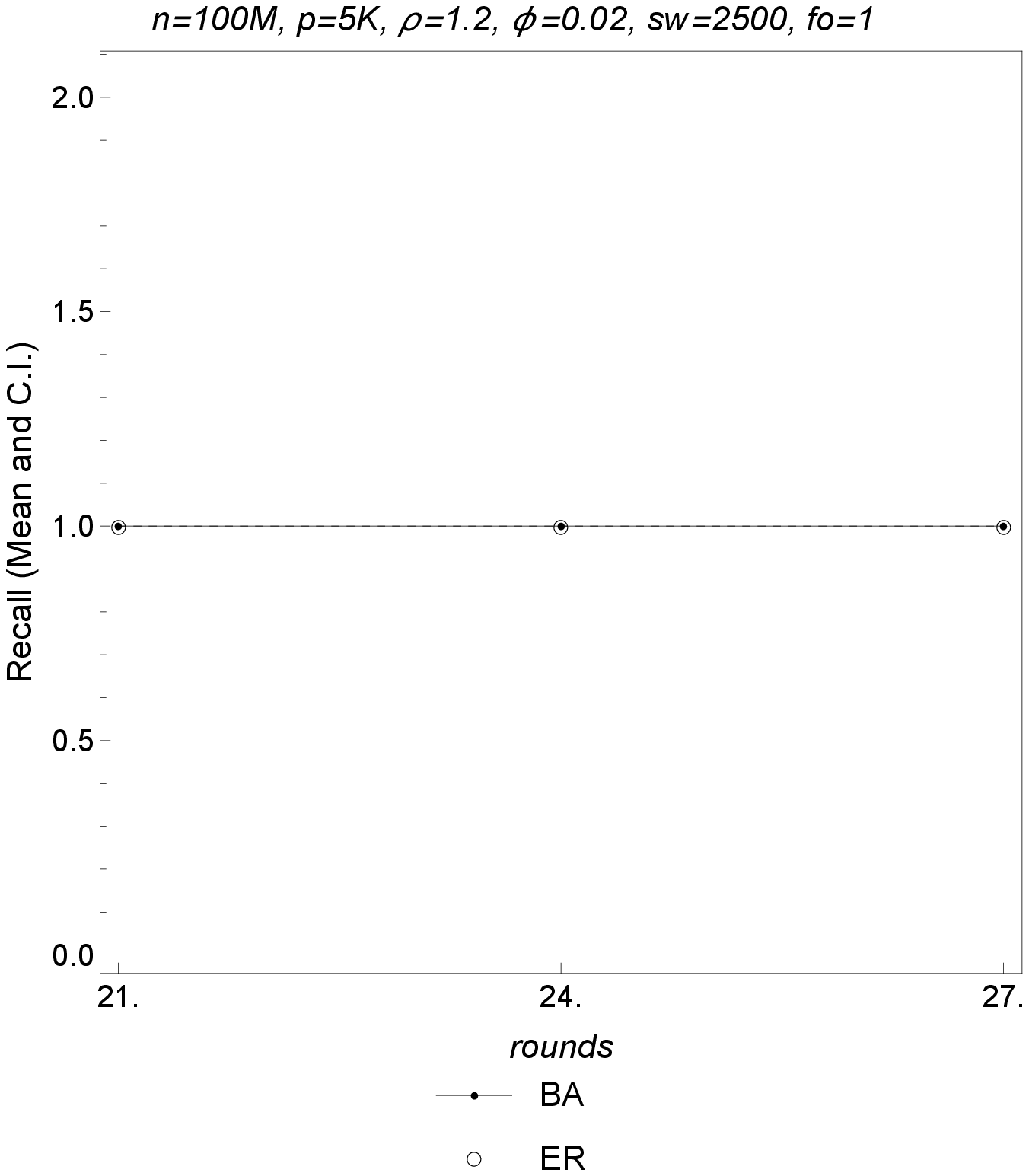}
			\label{r-rec}
		} &
		
		\subfloat[Precision]{
			\includegraphics[width=0.3\textwidth]{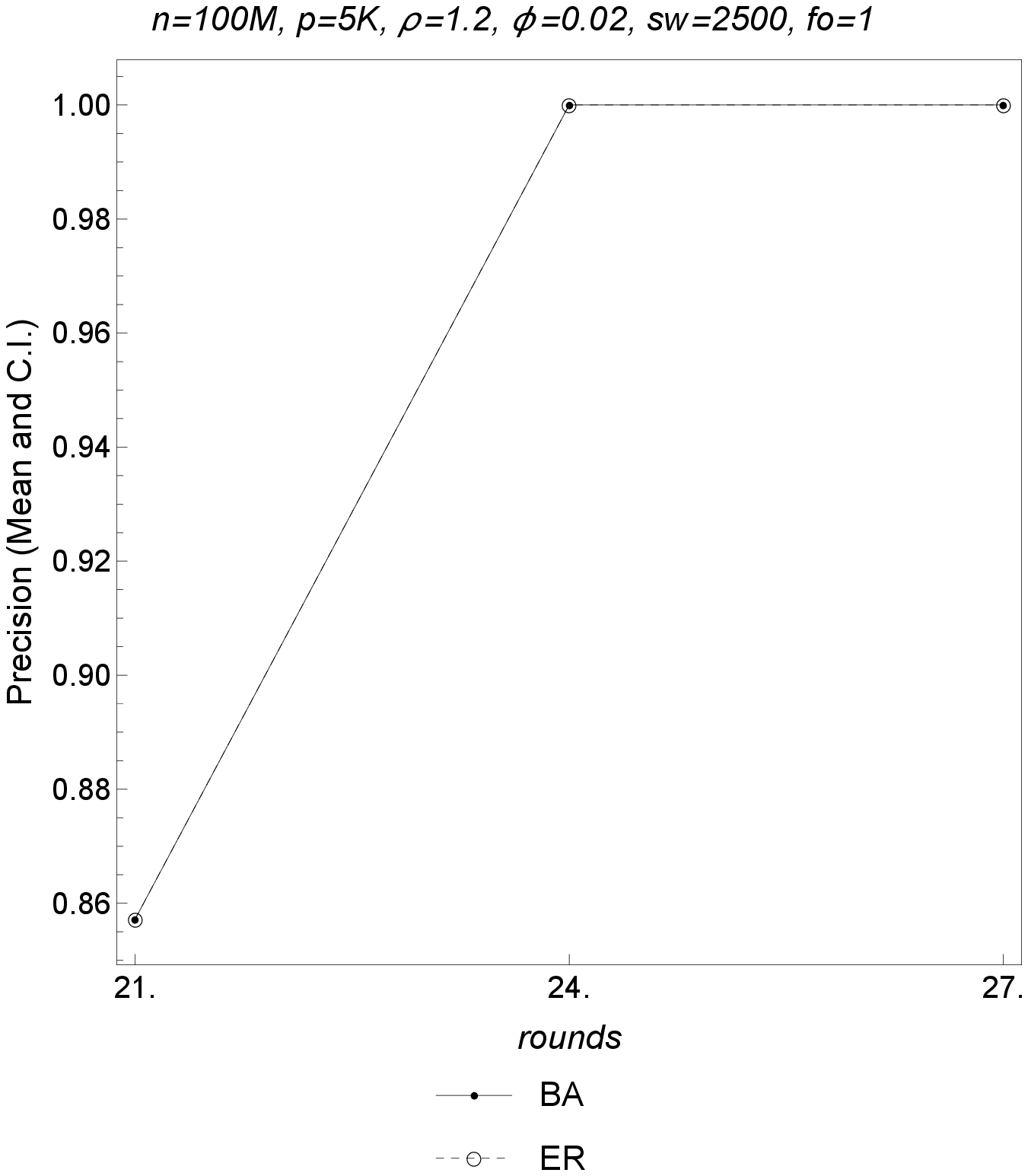}
			\label{r-prec}
		} &
		
		\subfloat[Average Relative Error ]{
			\includegraphics[width=0.3\textwidth]{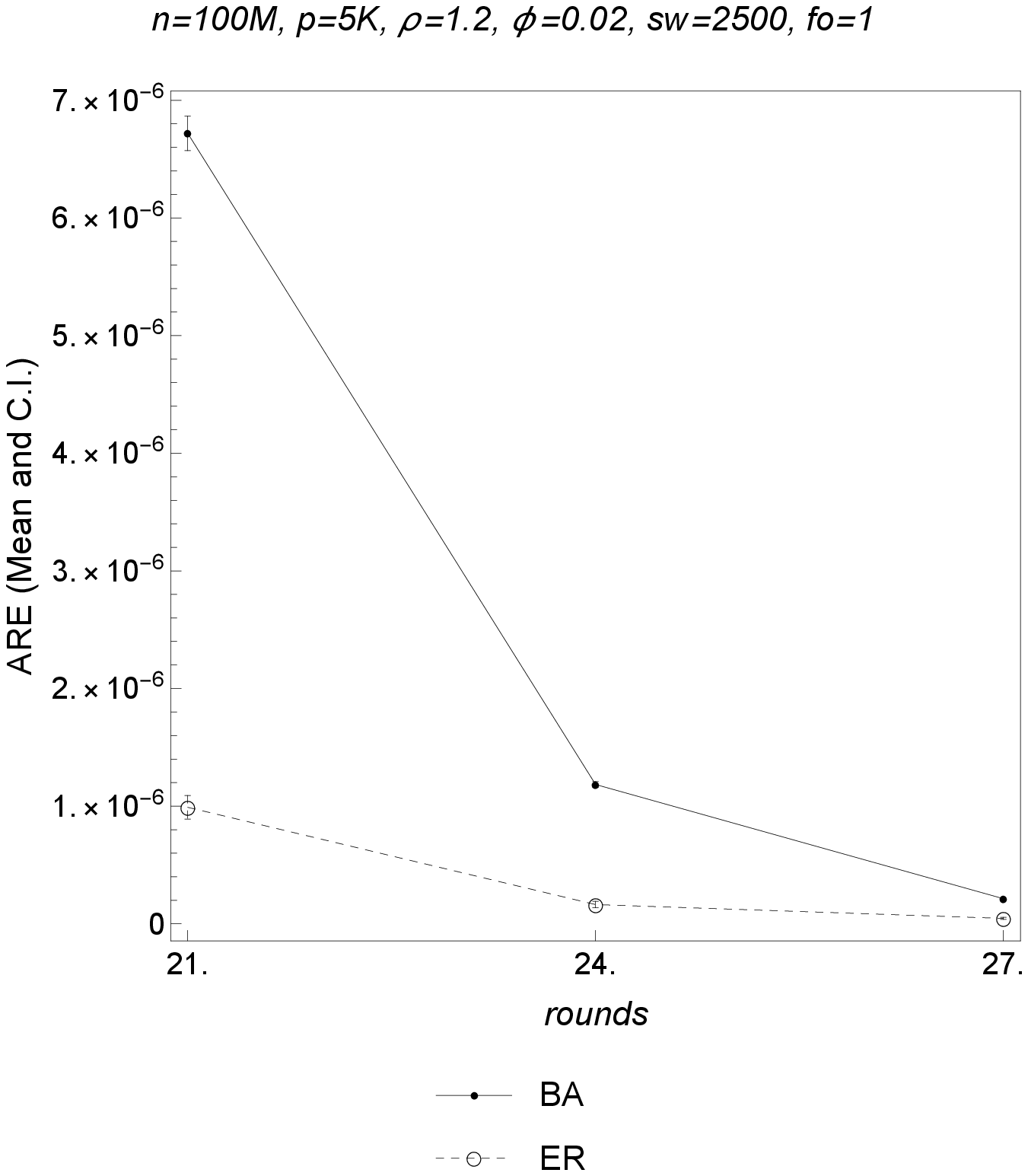}
			\label{r-are}
		} 
		
	\end{tabular}
	
	\caption{Recall, Precision and Average Relative Error (mean and confidence interval) varying the number of rounds executed,  for both a Barabasi-Albert (BA) and an Erdos-Renyi (ER) type of network graph.} 
	\label{r_plot}
\end{figure*}

A major sensitivity is exhibited by the algorithm when the number of rounds executed is varied, Figure~\ref{r_plot}. We note that the Precision grows and the Average Relative Error decreases as the number of rounds increases. This behavior is expected, given the theoretical analysis. 

Overall the experiments show that our algorithm exhibits very good performance in terms of Recall, Precision, and Average Relative Error of the frequency estimation when the guidance of the theoretical analysis is taken into account in determining the size of sketches used by peers and the number of rounds to be executed. Furthermore, the algorithm proves to be very robust to variations in the skewness of the input dataset and the frequent items threshold.

Nonetheless, the results above refer only to the case of a static network. In the following, we show how the algorithm behaves in presence of churning.

\subsection{Effect of the churn}
\label{effect-of-churn}
The performance of our \textsc{P2PTFHH} algorithm has been tested also in more realistic contexts with the adoption of  two churning models, the \textit{fail-stop} model and the  \textit{Yao} model, proposed by Yao et al. \cite{Yao2006}. 

%, the push pull protocol is still considered as atomic, hence we did not modify our algorithm to eventually handle the case when a peer leaves the system after sending the push message.
In the \textit{fail-stop} model, a peer could leave the network with a given failure probability and the failed peers can not join the network anymore. 

In the \textit{Yao} model, peers randomly join and leave the network. For each peer $i$, a random average \textit{lifetime} duration $l_i$ is generated from a Shifted Pareto distribution with $\alpha=3$, $\beta = 1$ and $\mu = 1.01$. At the same way, a random average \textit{offline} duration $d_i$ is generated from a Shifted Pareto distribution with same $\alpha$ and $\mu$ values and $\beta= 2$. We recall here that if $X \sim \text{Pareto(II)}(\mu,\beta,\alpha)$, i.e., $X$ is a random variable with a Pareto Type II distribution (also named Shifted Pareto), than its cumulative distribution function is $F_X(x) = 1 - \left(1 + \frac{x - \mu}{\beta}\right)^{-\alpha}$.

The values $l_i$ and $d_i$  are used to configure, for each peer $i$, two distributions $F_i$ and $G_i$. Distributions $G_i$ are Shifted Pareto distributions with $\beta = 3$ and $\alpha = 2d_i$. While, distributions $F_i$ can be both Pareto distributions with $\beta = 2$, $\alpha = 2l_i$, or exponential distributions with $\lambda = 1/l_i$. Whenever the state of a peer changes, a duration value is drawn from one of the distributions, $F_i$ or $G_i$, based on the type of duration values (lifetime or offtime) which must be generated.  We run our experiments with both the variants Pareto and Exponential lifetimes.

When using the fail-stop model, we run our algorithm with the default parameter values of Table \ref{experiments} and varying the failure probability through the values: $0.0, 0.01, 0.05, 0.1$. 
As shown in Figures~\ref{fp-rec} and \ref{fp-prec}, the recall and precision metrics are not affected by the introduction of peer failures up to a failure probability equal to $0.1$. However, as expected, the average relative error on frequency estimations gets worse going from about $10^{-7}$ in case of no churn to about $10^{-1}$ when the failure probability is $0.1$, see Figure~\ref{fp-are}. 

\begin{figure*}[h]
	\centering
	\begin{tabular}{ccc}		
		\subfloat[Recall]{
			\includegraphics[width=0.3\textwidth]{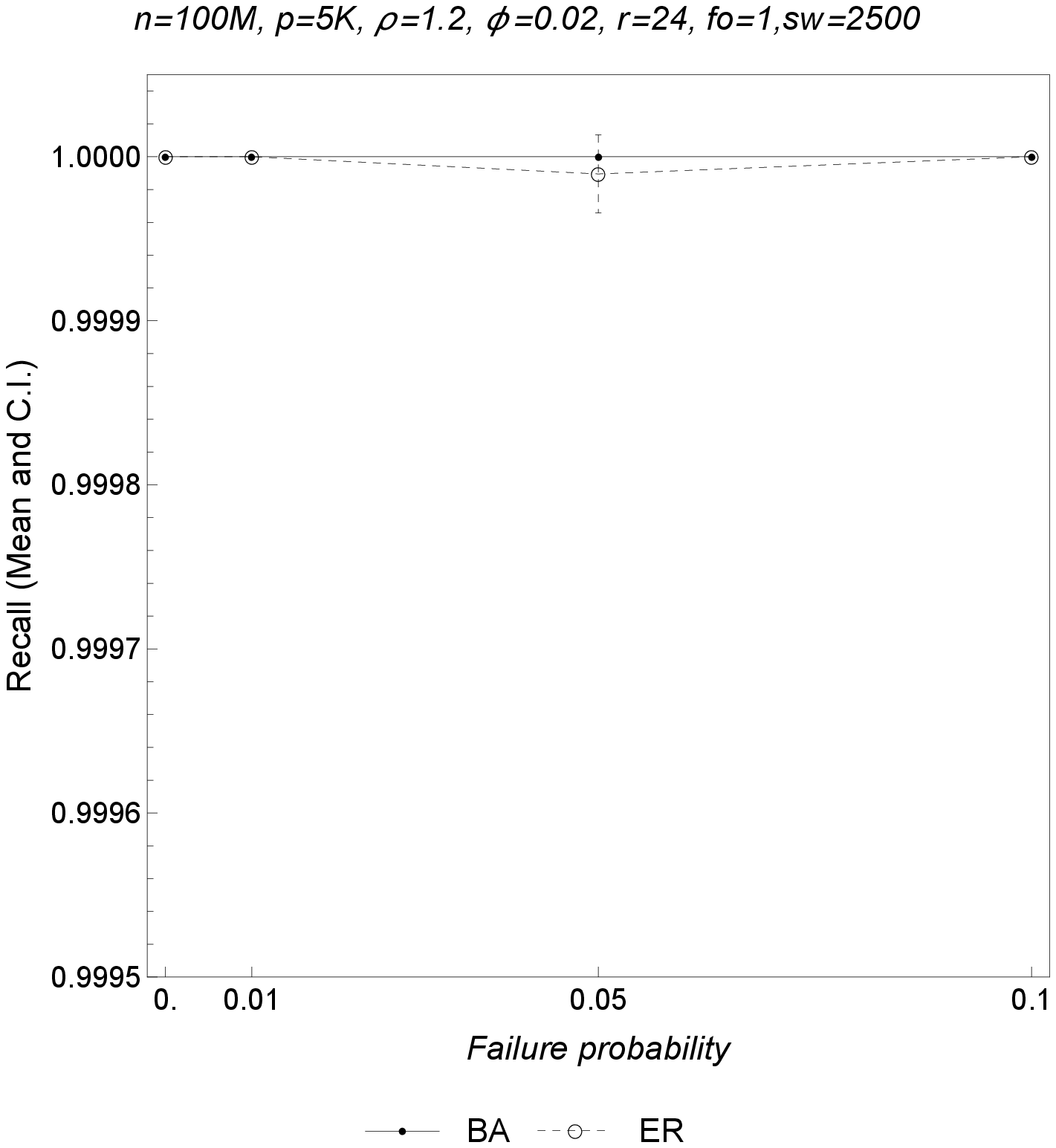}
			\label{fp-rec}
		} &
		
		\subfloat[Precision]{
			\includegraphics[width=0.3\textwidth]{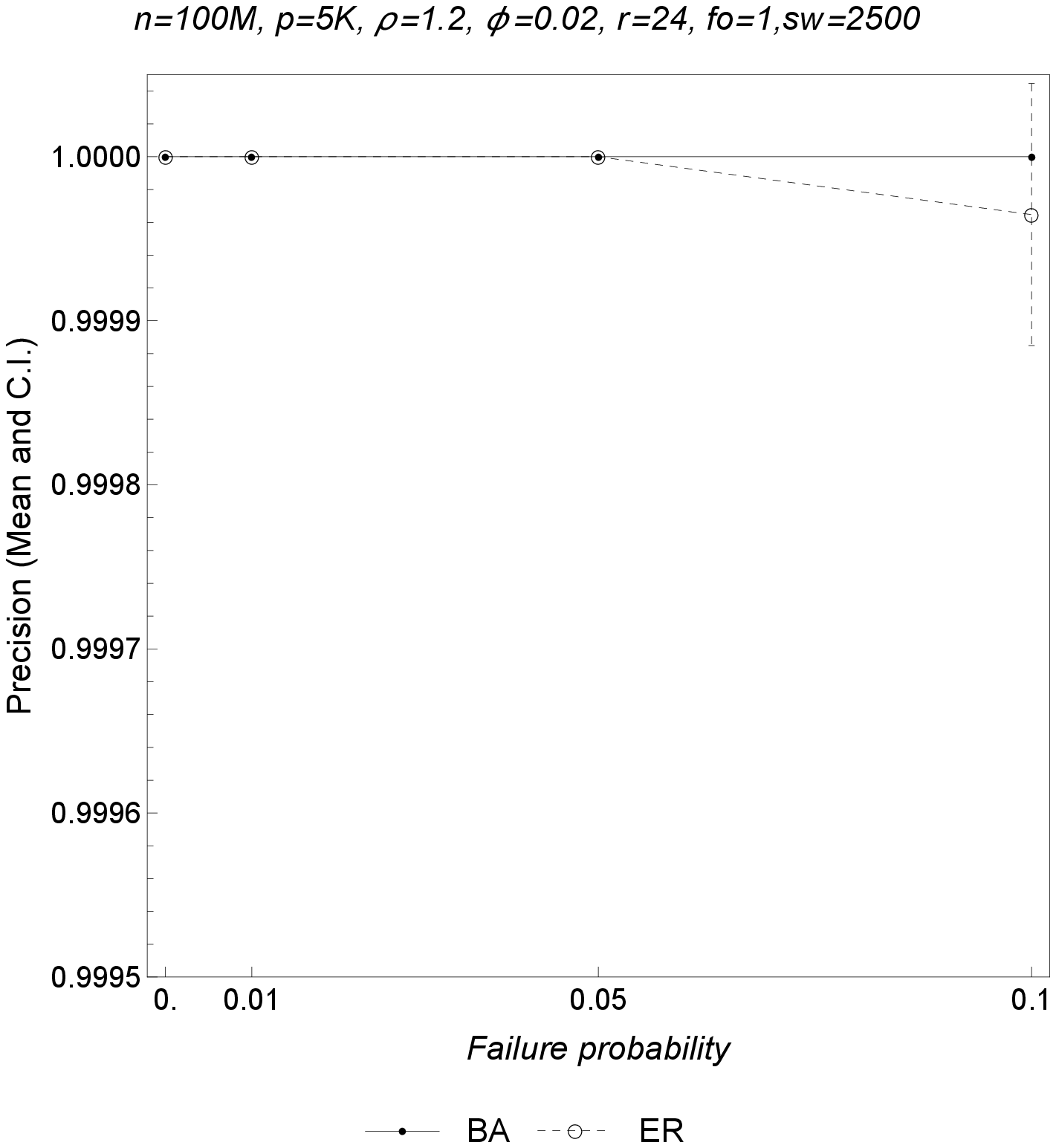}
			\label{fp-prec}
		} &
		
		\subfloat[Average Relative Error ]{
			\includegraphics[width=0.3\textwidth]{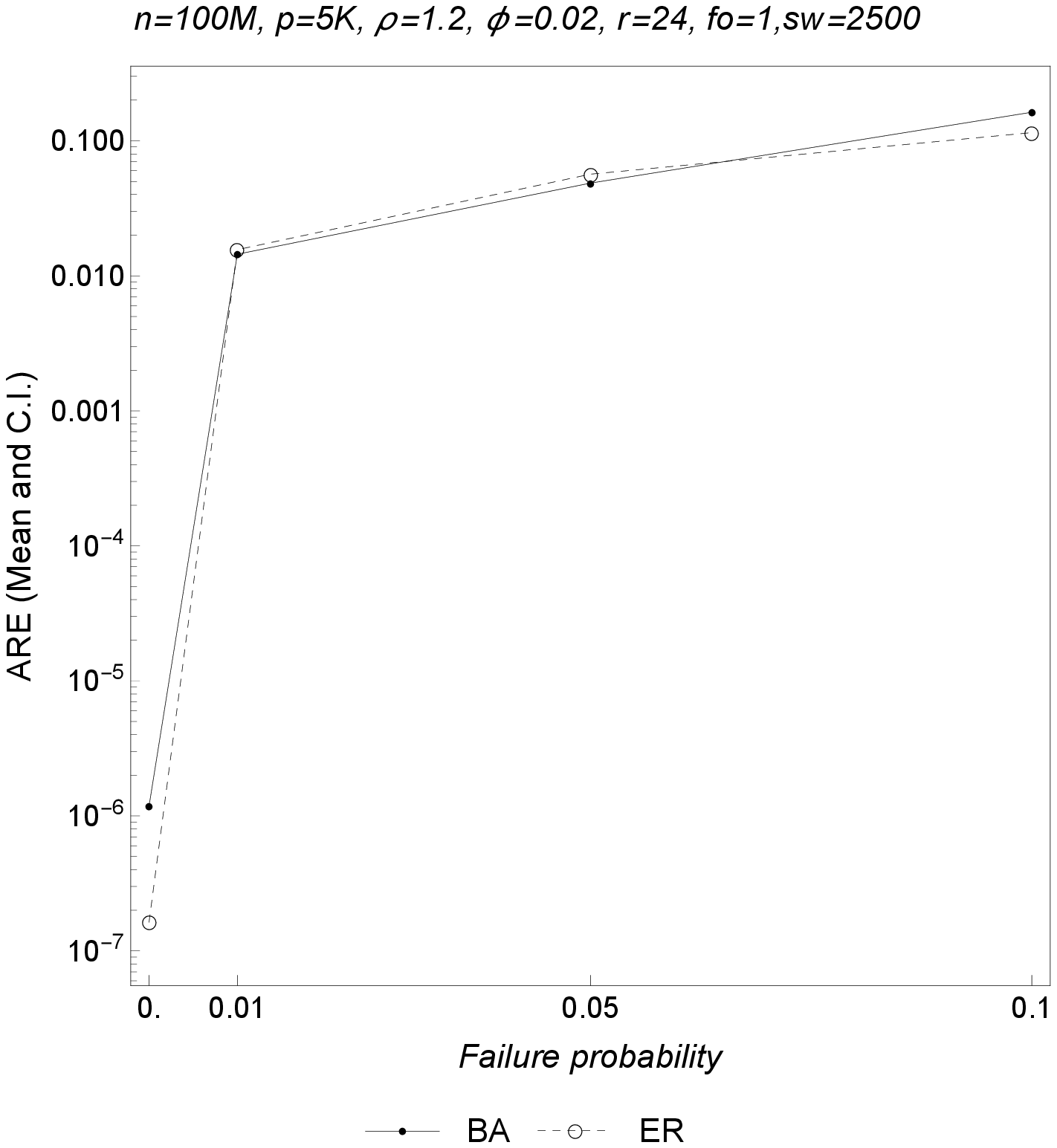}
			\label{fp-are}
		} 
		
	\end{tabular}
	
	\caption{Recall, Precision and Average Relative Error (mean and confidence interval) varying the failure probability in a fail-stop model of churning,  for both a Barabasi-Albert (BA) and an Erdos-Renyi (ER) type of network graph.} 
	\label{fp_plot}
\end{figure*}

When the Yao model of churning was adopted, we run our algorithm with the default values of Table \ref{experiments} and the parameters of churning already discussed, varying the number of peers available to  participate in the gossip and the number of rounds executed. Also in these cases, recall and precision are not  affected by the introduction of churning. In fact, we obtained the same plots for recall and precision varying the number of peers and the number of rounds as those in Figures~\ref{p-rec}, \ref{p-prec} and \ref{r-rec} and \ref{r-prec}, the reason why we do not report these plots again. On the other hand, the average relative error is affected by the churning, as expected: Figures \ref{p-churn} and \ref{r-churn} are related respectively to the ARE measured varying the number of peers and the number of rounds with Pareto distributions for lifetimes, meanwhile Figures \ref{p-expchurn} and \ref{r-expchurn} refer to the ARE measured when using Exponential distributions for lifetimes.

\begin{figure*}[h]
	\centering
	\begin{tabular}{cc}		
		\subfloat[Average Relative Error]{
			\includegraphics[width=0.45\textwidth]{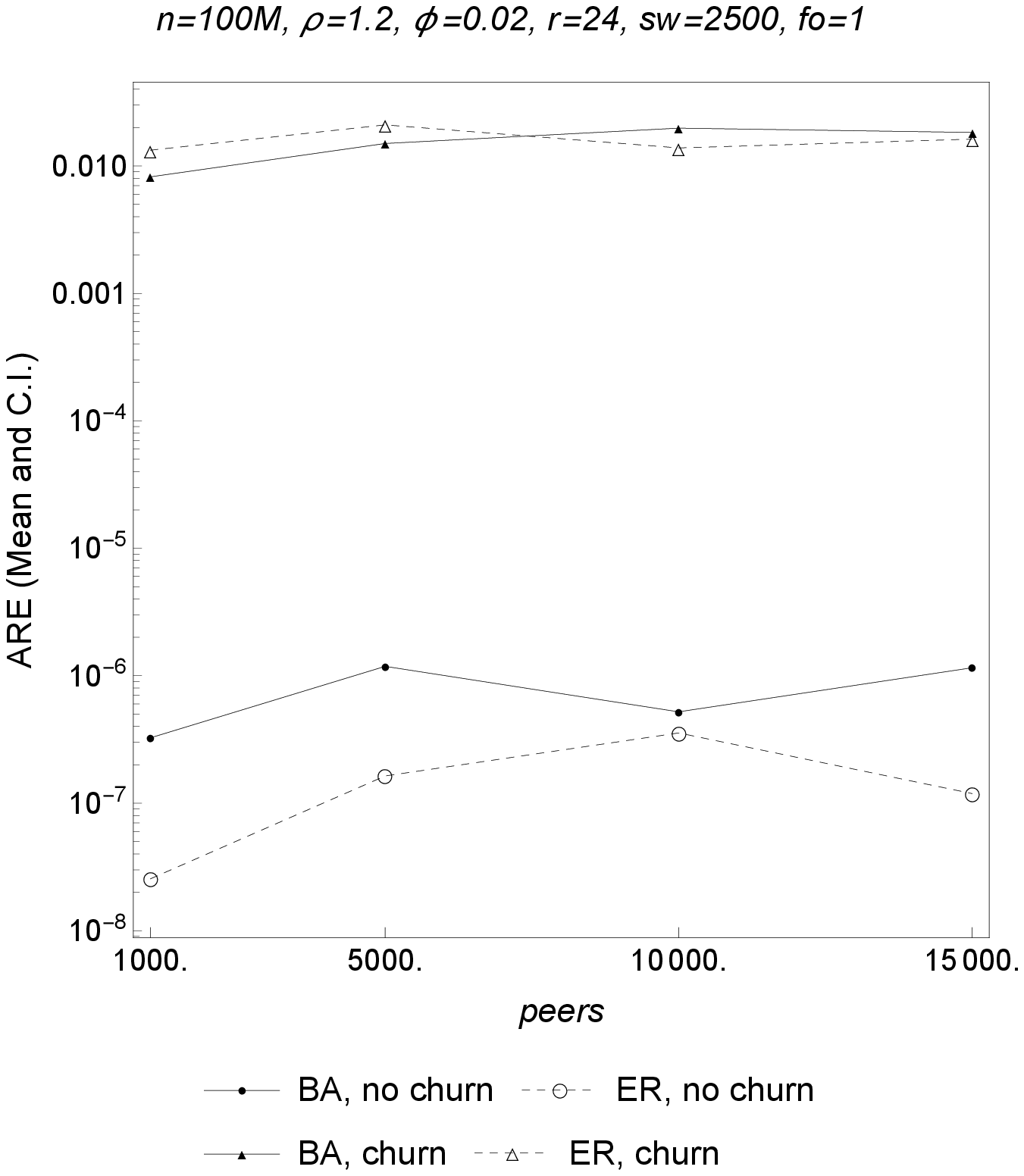}
			\label{p-churn}
		} &
		
		\subfloat[Average Relative Error]{
			\includegraphics[width=0.45\textwidth]{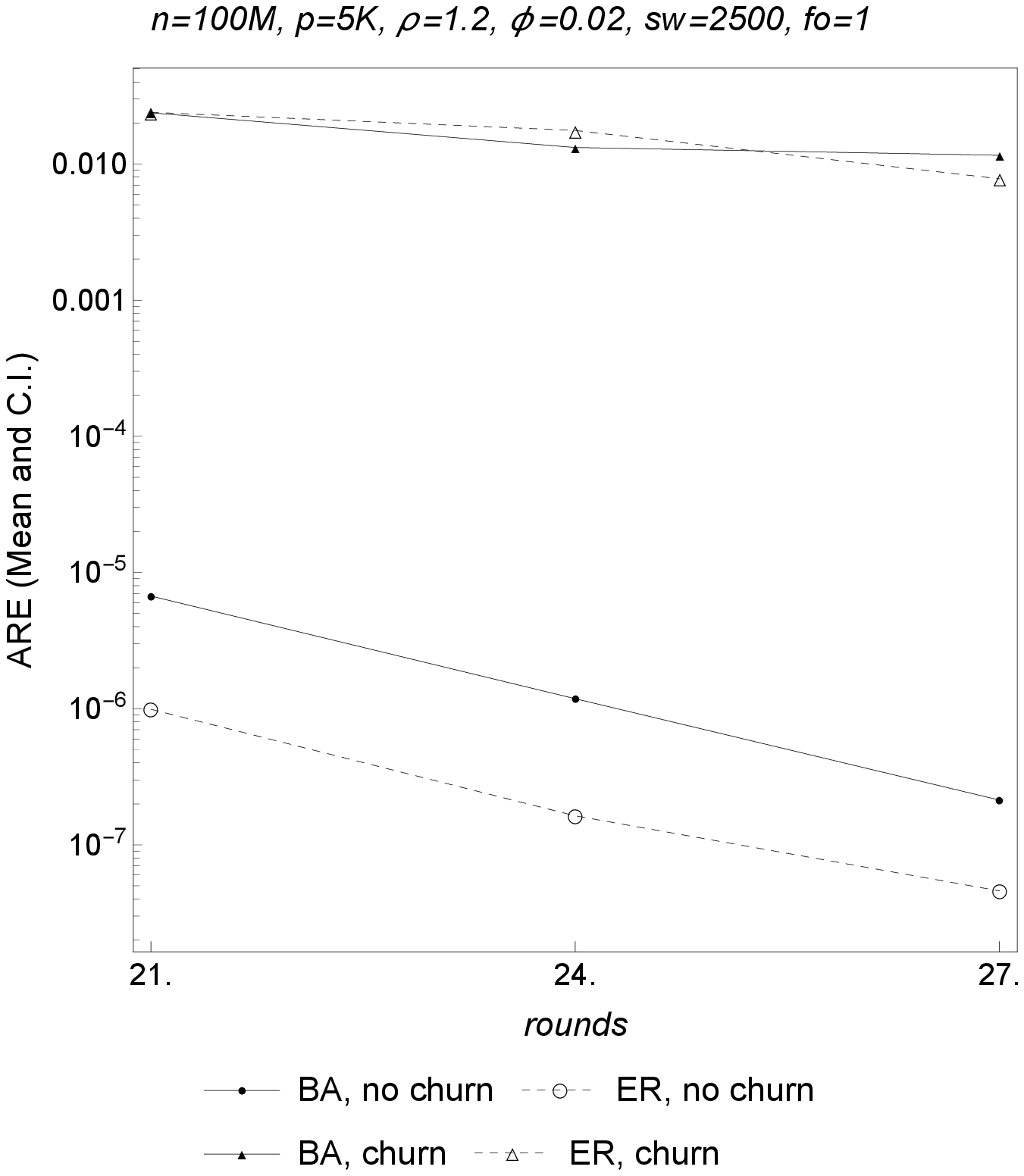}
			\label{r-churn}
		} \\
		
		\subfloat[Average Relative Error ]{
			\includegraphics[width=0.45\textwidth]{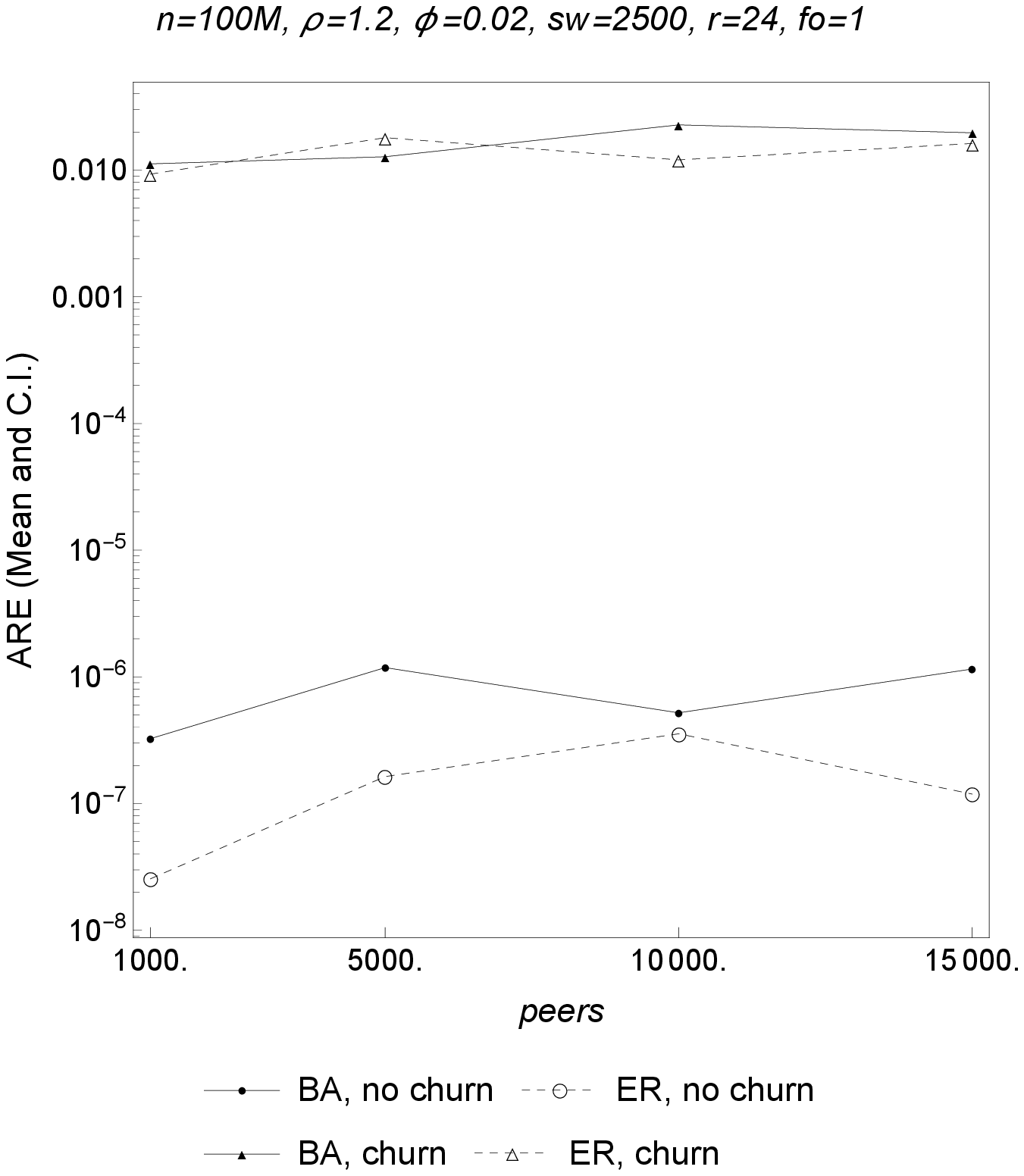}
			\label{p-expchurn}
		} &
		
		\subfloat[Average Relative Error ]{
			\includegraphics[width=0.45\textwidth]{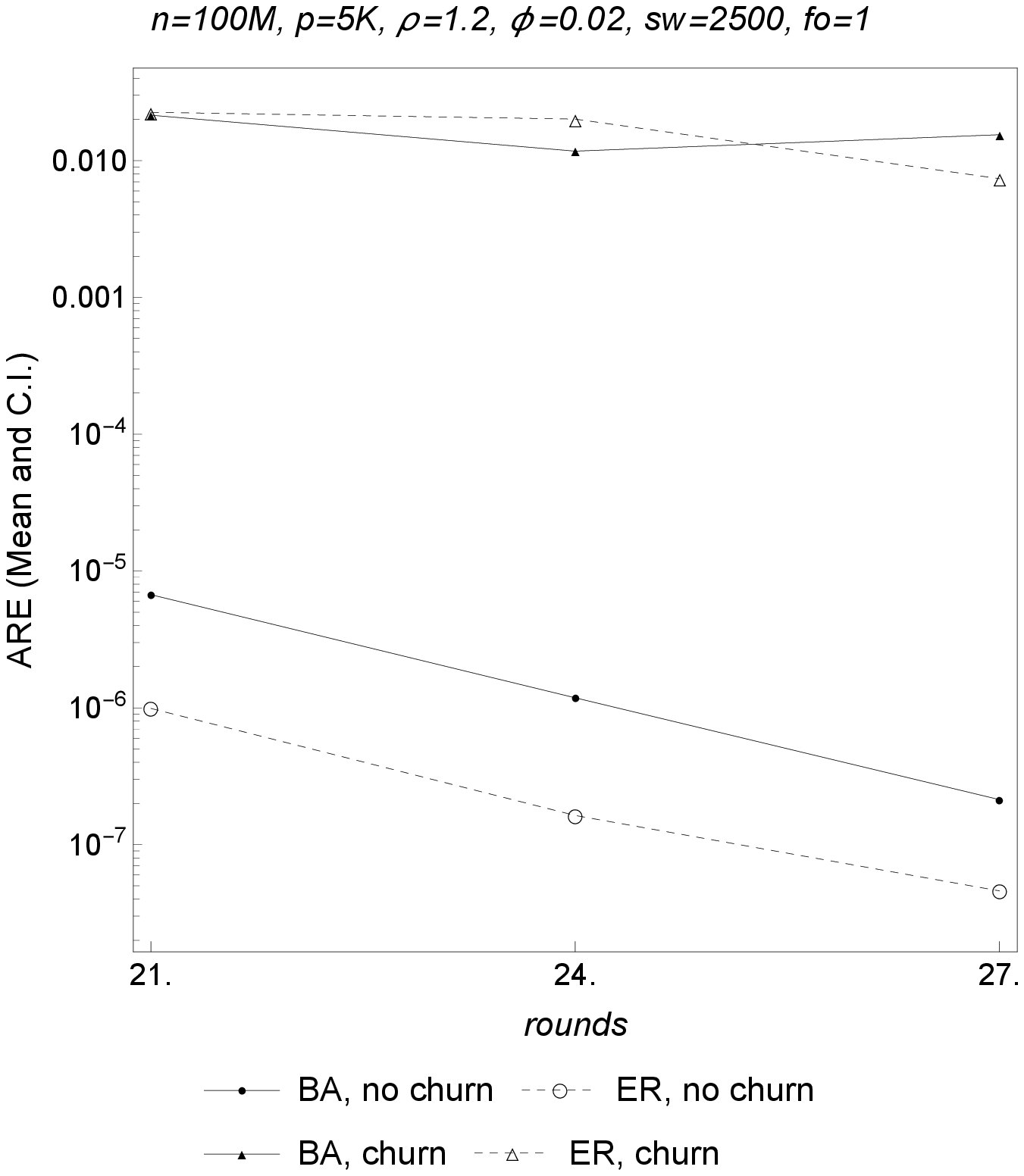}
			\label{r-expchurn}
		} 
		
	\end{tabular}
	
	\caption{Average Relative Error (mean and confidence interval) varying the number of peers and rounds in Yao model of churning and Pareto (a,b) or Exponential (c, d) lifetimes,  for both a Barabasi-Albert (BA) and an Erdos-Renyi (ER) type of network graph.} 
	\label{yao_plots}
\end{figure*}

\section{Conclusions}
\label{conclusions}

In this paper we presented the \textsc{P2PTFHH} (Peer--to--Peer Time--Faded Heavy Hitters) distributed algorithm. Building on the distributed averaging protocol, we designed a gossip--based version of our sequential \textsc{FDCMSS} (Forward Decay Count-Min Space-Saving) algorithm for distributed mining of time--faded heavy hitters. To the best of our knowledge, this is the first distributed protocol designed specifically for the problem of mining on unstructured P2P networks time--faded heavy hitters, a data mining task with wide applicability, for instance in the context of sensor data mining, for business decision support, analysis of web query logs, network measurement, monitoring and traffic analysis. In these applications, the use of the time--fading model discounts the effect of old data, by putting heavier weights on recent batches of streaming  data than older batches. We formally proved the algorithm's correctness and error bounds on frequency estimation. The extensive experimental results also proved that \textsc{P2PTFHH} is extremely accurate and fast, allowing near real time processing of large datasets.

% the following \clearpage command will prevent floats to appear in or after the references

\clearpage

\section*{References}
\bibliographystyle{elsarticle-harv}
\bibliography{bibliography}

\begin{thebibliography}{45}
\expandafter\ifx\csname natexlab\endcsname\relax\def\natexlab#1{#1}\fi
\providecommand{\url}[1]{\texttt{#1}}
\providecommand{\href}[2]{#2}
\providecommand{\path}[1]{#1}
\providecommand{\DOIprefix}{doi:}
\providecommand{\ArXivprefix}{arXiv:}
\providecommand{\URLprefix}{URL: }
\providecommand{\Pubmedprefix}{pmid:}
\providecommand{\doi}[1]{\href{http://dx.doi.org/#1}{\path{#1}}}
\providecommand{\Pubmed}[1]{\href{pmid:#1}{\path{#1}}}
\providecommand{\bibinfo}[2]{#2}
\ifx\xfnm\relax \def\xfnm[#1]{\unskip,\space#1}\fi
%Type = Inproceedings
\bibitem[{Cafaro et~al.(2017a)Cafaro, Epicoco, Aloisio and Pulimeno}]{HPCS2017}
\bibinfo{author}{Cafaro, M.}, \bibinfo{author}{Epicoco, I.},
  \bibinfo{author}{Aloisio, G.}, \bibinfo{author}{Pulimeno, M.},
  \bibinfo{year}{2017}a.
\newblock \bibinfo{title}{Cuda based parallel implementations of space-saving
  on a gpu}, in: \bibinfo{booktitle}{2017 International Conference on High
  Performance Computing \& Simulation (HPCS), Genoa, Italy, 2017},
  \bibinfo{publisher}{IEEE}. pp. \bibinfo{pages}{707--714}.
\newblock \DOIprefix\doi{10.1109/HPCS.2017.108}.
%Type = Article
\bibitem[{Cafaro et~al.(2017b)Cafaro, Epicoco, Pulimeno and
  Aloisio}]{FDCMSSvsFSSQ}
\bibinfo{author}{Cafaro, M.}, \bibinfo{author}{Epicoco, I.},
  \bibinfo{author}{Pulimeno, M.}, \bibinfo{author}{Aloisio, G.},
  \bibinfo{year}{2017}b.
\newblock \bibinfo{title}{On frequency estimation and detection of frequent
  items in time faded streams}.
\newblock \bibinfo{journal}{IEEE Access}
  \DOIprefix\doi{10.1109/ACCESS.2017.2757238}. \bibinfo{note}{in press.}
%Type = Inproceedings
\bibitem[{Cafaro and Pulimeno(2016)}]{Cafaro-Pulimeno}
\bibinfo{author}{Cafaro, M.}, \bibinfo{author}{Pulimeno, M.},
  \bibinfo{year}{2016}.
\newblock \bibinfo{title}{Merging frequent summaries}, in:
  \bibinfo{booktitle}{Proceedings of the 17th Italian Conference on Theoretical
  Computer Science (ICTCS 2016), Volume 1720}, \bibinfo{publisher}{CEUR
  Proceedings}. pp. \bibinfo{pages}{280--285}.
%Type = Article
\bibitem[{Cafaro et~al.(2018)Cafaro, Pulimeno and Epicoco}]{CAFARO2018115}
\bibinfo{author}{Cafaro, M.}, \bibinfo{author}{Pulimeno, M.},
  \bibinfo{author}{Epicoco, I.}, \bibinfo{year}{2018}.
\newblock \bibinfo{title}{Parallel mining of time-faded heavy hitters}.
\newblock \bibinfo{journal}{Expert Systems with Applications}
  \bibinfo{volume}{96}, \bibinfo{pages}{115 -- 128}.
\newblock \URLprefix
  \url{http://www.sciencedirect.com/science/article/pii/S0957417417307777},
  \DOIprefix\doi{https://doi.org/10.1016/j.eswa.2017.11.021}.
%Type = Article
\bibitem[{Cafaro et~al.(2016a)Cafaro, Pulimeno, Epicoco and
  Aloisio}]{Cafaro-Pulimeno-Epicoco-Aloisio}
\bibinfo{author}{Cafaro, M.}, \bibinfo{author}{Pulimeno, M.},
  \bibinfo{author}{Epicoco, I.}, \bibinfo{author}{Aloisio, G.},
  \bibinfo{year}{2016}a.
\newblock \bibinfo{title}{Mining frequent items in the time fading model}.
\newblock \bibinfo{journal}{Information Sciences} \bibinfo{volume}{370--371},
  \bibinfo{pages}{221--238}.
\newblock \DOIprefix\doi{10.1016/j.ins.2016.07.077}.
%Type = Article
\bibitem[{Cafaro et~al.(2017c)Cafaro, Pulimeno, Epicoco and
  Aloisio}]{CPE:CPE4160}
\bibinfo{author}{Cafaro, M.}, \bibinfo{author}{Pulimeno, M.},
  \bibinfo{author}{Epicoco, I.}, \bibinfo{author}{Aloisio, G.},
  \bibinfo{year}{2017}c.
\newblock \bibinfo{title}{Parallel space saving on multi- and many-core
  processors}.
\newblock \bibinfo{journal}{Concurrency and Computation: Practice and
  Experience} \bibinfo{volume}{30}, \bibinfo{pages}{e4160--n/a}.
\newblock \URLprefix \url{http://dx.doi.org/10.1002/cpe.4160},
  \DOIprefix\doi{10.1002/cpe.4160}. \bibinfo{note}{e4160 cpe.4160}.
%Type = Article
\bibitem[{Cafaro et~al.(2016b)Cafaro, Pulimeno and
  Tempesta}]{Cafaro-Pulimeno-Tempesta}
\bibinfo{author}{Cafaro, M.}, \bibinfo{author}{Pulimeno, M.},
  \bibinfo{author}{Tempesta, P.}, \bibinfo{year}{2016}b.
\newblock \bibinfo{title}{A parallel space saving algorithm for frequent items
  and the hurwitz zeta distribution}.
\newblock \bibinfo{journal}{Information Sciences} \bibinfo{volume}{329},
  \bibinfo{pages}{1 -- 19}.
\newblock \DOIprefix\doi{http://dx.doi.org/10.1016/j.ins.2015.09.003}.
%Type = Article
\bibitem[{Cafaro and Tempesta(2011)}]{cafaro-tempesta}
\bibinfo{author}{Cafaro, M.}, \bibinfo{author}{Tempesta, P.},
  \bibinfo{year}{2011}.
\newblock \bibinfo{title}{Finding frequent items in parallel}.
\newblock \bibinfo{journal}{Concurrency and Computation: Practice and
  Experience} \bibinfo{volume}{23}, \bibinfo{pages}{1774--1788}.
\newblock \DOIprefix\doi{10.1002/cpe.1761}.
%Type = Inproceedings
\bibitem[{Cao and Wang(2004)}]{Cao:2004}
\bibinfo{author}{Cao, P.}, \bibinfo{author}{Wang, Z.}, \bibinfo{year}{2004}.
\newblock \bibinfo{title}{Efficient top-k query calculation in distributed
  networks}, in: \bibinfo{booktitle}{Proceedings of the Twenty-third Annual ACM
  Symposium on Principles of Distributed Computing}, \bibinfo{publisher}{ACM},
  \bibinfo{address}{New York, NY, USA}. pp. \bibinfo{pages}{206--215}.
\newblock \URLprefix \url{http://doi.acm.org/10.1145/1011767.1011798},
  \DOIprefix\doi{10.1145/1011767.1011798}.
%Type = Inproceedings
\bibitem[{Charikar et~al.(2002)Charikar, Chen and Farach-Colton}]{Charikar}
\bibinfo{author}{Charikar, M.}, \bibinfo{author}{Chen, K.},
  \bibinfo{author}{Farach-Colton, M.}, \bibinfo{year}{2002}.
\newblock \bibinfo{title}{Finding frequent items in data streams}, in:
  \bibinfo{booktitle}{ICALP '02: Proceedings of the 29th International
  Colloquium on Automata, Languages and Programming},
  \bibinfo{publisher}{Springer--Verlag}. pp. \bibinfo{pages}{693--703}.
%Type = Article
\bibitem[{Chen and Mei(2014)}]{Chen-Mei}
\bibinfo{author}{Chen, L.}, \bibinfo{author}{Mei, Q.}, \bibinfo{year}{2014}.
\newblock \bibinfo{title}{Mining frequent items in data stream using time
  fading model}.
\newblock \bibinfo{journal}{Information Sciences} \bibinfo{volume}{257},
  \bibinfo{pages}{54 -- 69}.
\newblock \DOIprefix\doi{http://dx.doi.org/10.1016/j.ins.2013.09.007}.
%Type = Inproceedings
\bibitem[{Cormode et~al.(2008)Cormode, Korn and Tirthapura}]{exp-decay}
\bibinfo{author}{Cormode, G.}, \bibinfo{author}{Korn, F.},
  \bibinfo{author}{Tirthapura, S.}, \bibinfo{year}{2008}.
\newblock \bibinfo{title}{Exponentially decayed aggregates on data streams},
  in: \bibinfo{booktitle}{Proceedings of the 2008 IEEE 24th International
  Conference on Data Engineering}, \bibinfo{publisher}{IEEE Computer Society},
  \bibinfo{address}{Washington, DC, USA}. pp. \bibinfo{pages}{1379--1381}.
\newblock \URLprefix \url{http://dx.doi.org/10.1109/ICDE.2008.4497562},
  \DOIprefix\doi{10.1109/ICDE.2008.4497562}.
%Type = Article
\bibitem[{Cormode and Muthukrishnan(2005a)}]{Cormode05}
\bibinfo{author}{Cormode, G.}, \bibinfo{author}{Muthukrishnan, S.},
  \bibinfo{year}{2005}a.
\newblock \bibinfo{title}{An improved data stream summary: the count-min sketch
  and its applications}.
\newblock \bibinfo{journal}{J. Algorithms} \bibinfo{volume}{55},
  \bibinfo{pages}{58--75}.
\newblock \DOIprefix\doi{http://dx.doi.org/10.1016/j.jalgor.2003.12.001}.
%Type = Article
\bibitem[{Cormode and Muthukrishnan(2005b)}]{Cormode-grouptest}
\bibinfo{author}{Cormode, G.}, \bibinfo{author}{Muthukrishnan, S.},
  \bibinfo{year}{2005}b.
\newblock \bibinfo{title}{What's hot and what's not: Tracking most frequent
  items dynamically}.
\newblock \bibinfo{journal}{ACM Trans. Database Syst.} \bibinfo{volume}{30},
  \bibinfo{pages}{249--278}.
\newblock \DOIprefix\doi{10.1145/1061318.1061325}.
%Type = Inproceedings
\bibitem[{Cormode et~al.(2009)Cormode, Shkapenyuk, Srivastava and
  Xu}]{forward-decay}
\bibinfo{author}{Cormode, G.}, \bibinfo{author}{Shkapenyuk, V.},
  \bibinfo{author}{Srivastava, D.}, \bibinfo{author}{Xu, B.},
  \bibinfo{year}{2009}.
\newblock \bibinfo{title}{Forward decay: A practical time decay model for
  streaming systems}, in: \bibinfo{booktitle}{2009 IEEE 25th International
  Conference on Data Engineering}, pp. \bibinfo{pages}{138--149}.
\newblock \DOIprefix\doi{10.1109/ICDE.2009.65}.
%Type = Article
\bibitem[{Csardi and Nepusz(2006)}]{libigraph}
\bibinfo{author}{Csardi, G.}, \bibinfo{author}{Nepusz, T.},
  \bibinfo{year}{2006}.
\newblock \bibinfo{title}{The igraph software package for complex network
  research}.
\newblock \bibinfo{journal}{InterJournal} \bibinfo{volume}{Complex Systems},
  \bibinfo{pages}{1695}.
\newblock \URLprefix \url{http://igraph.org}.
%Type = Article
\bibitem[{Das et~al.(2009)Das, Antony, Agrawal and El~Abbadi}]{Das2009}
\bibinfo{author}{Das, S.}, \bibinfo{author}{Antony, S.},
  \bibinfo{author}{Agrawal, D.}, \bibinfo{author}{El~Abbadi, A.},
  \bibinfo{year}{2009}.
\newblock \bibinfo{title}{Thread cooperation in multicore architectures for
  frequency counting over multiple data streams}.
\newblock \bibinfo{journal}{Proc. VLDB Endow.} \bibinfo{volume}{2},
  \bibinfo{pages}{217--228}.
\newblock \DOIprefix\doi{10.14778/1687627.1687653}.
%Type = Inproceedings
\bibitem[{Datar et~al.(2002)Datar, Gionis, Indyk and Motwani}]{Datar}
\bibinfo{author}{Datar, M.}, \bibinfo{author}{Gionis, A.},
  \bibinfo{author}{Indyk, P.}, \bibinfo{author}{Motwani, R.},
  \bibinfo{year}{2002}.
\newblock \bibinfo{title}{Maintaining stream statistics over sliding windows:
  (extended abstract)}, in: \bibinfo{booktitle}{Proceedings of the Thirteenth
  Annual ACM-SIAM Symposium on Discrete Algorithms},
  \bibinfo{publisher}{Society for Industrial and Applied Mathematics},
  \bibinfo{address}{Philadelphia, PA, USA}. pp. \bibinfo{pages}{635--644}.
%Type = Inproceedings
\bibitem[{Demaine et~al.(2002)Demaine, L{\'o}pez-Ortiz and Munro}]{DemaineLM02}
\bibinfo{author}{Demaine, E.D.}, \bibinfo{author}{L{\'o}pez-Ortiz, A.},
  \bibinfo{author}{Munro, J.I.}, \bibinfo{year}{2002}.
\newblock \bibinfo{title}{Frequency estimation of internet packet streams with
  limited space}, in: \bibinfo{booktitle}{Proceedings of the 10th Annual
  European Symposium on Algorithms}, pp. \bibinfo{pages}{348--360}.
%Type = Inproceedings
\bibitem[{Demers et~al.(1987)Demers, Greene, Hauser, Irish, Larson, Shenker,
  Sturgis, Swinehart and Terry}]{Demers:1987}
\bibinfo{author}{Demers, A.}, \bibinfo{author}{Greene, D.},
  \bibinfo{author}{Hauser, C.}, \bibinfo{author}{Irish, W.},
  \bibinfo{author}{Larson, J.}, \bibinfo{author}{Shenker, S.},
  \bibinfo{author}{Sturgis, H.}, \bibinfo{author}{Swinehart, D.},
  \bibinfo{author}{Terry, D.}, \bibinfo{year}{1987}.
\newblock \bibinfo{title}{Epidemic algorithms for replicated database
  maintenance}, in: \bibinfo{booktitle}{Proceedings of the Sixth Annual ACM
  Symposium on Principles of Distributed Computing}, \bibinfo{publisher}{ACM},
  \bibinfo{address}{New York, NY, USA}. pp. \bibinfo{pages}{1--12}.
\newblock \URLprefix \url{http://doi.acm.org/10.1145/41840.41841},
  \DOIprefix\doi{10.1145/41840.41841}.
%Type = Article
\bibitem[{Epicoco et~al.(2018)Epicoco, Cafaro and
  Pulimeno}]{Epicoco:2018:FAM:3182040.3182103}
\bibinfo{author}{Epicoco, I.}, \bibinfo{author}{Cafaro, M.},
  \bibinfo{author}{Pulimeno, M.}, \bibinfo{year}{2018}.
\newblock \bibinfo{title}{Fast and accurate mining of correlated heavy
  hitters}.
\newblock \bibinfo{journal}{Data Min. Knowl. Discov.} \bibinfo{volume}{32},
  \bibinfo{pages}{162--186}.
\newblock \URLprefix \url{https://doi.org/10.1007/s10618-017-0526-x},
  \DOIprefix\doi{10.1007/s10618-017-0526-x}.
%Type = Article
\bibitem[{Erra and Frola(2012)}]{Erra2012}
\bibinfo{author}{Erra, U.}, \bibinfo{author}{Frola, B.}, \bibinfo{year}{2012}.
\newblock \bibinfo{title}{Frequent items mining acceleration exploiting fast
  parallel sorting on the \{GPU\}}.
\newblock \bibinfo{journal}{Procedia Computer Science} \bibinfo{volume}{9},
  \bibinfo{pages}{86 -- 95}.
\newblock \DOIprefix\doi{http://dx.doi.org/10.1016/j.procs.2012.04.010}.
  \bibinfo{note}{proceedings of the International Conference on Computational
  Science, \{ICCS\} 2012}.
%Type = Inproceedings
\bibitem[{Govindaraju et~al.(2005)Govindaraju, Raghuvanshi and
  Manocha}]{Govindaraju2005}
\bibinfo{author}{Govindaraju, N.K.}, \bibinfo{author}{Raghuvanshi, N.},
  \bibinfo{author}{Manocha, D.}, \bibinfo{year}{2005}.
\newblock \bibinfo{title}{Fast and approximate stream mining of quantiles and
  frequencies using graphics processors}, in: \bibinfo{booktitle}{Proceedings
  of the 2005 ACM SIGMOD International Conference on Management of Data},
  \bibinfo{publisher}{ACM}. pp. \bibinfo{pages}{611--622}.
\newblock \DOIprefix\doi{10.1145/1066157.1066227}.
%Type = Article
\bibitem[{Jelasity et~al.(2005)Jelasity, Montresor and Babaoglu}]{Jelasity2005}
\bibinfo{author}{Jelasity, M.}, \bibinfo{author}{Montresor, A.},
  \bibinfo{author}{Babaoglu, O.}, \bibinfo{year}{2005}.
\newblock \bibinfo{title}{Gossip-based aggregation in large dynamic networks}.
\newblock \bibinfo{journal}{ACM Trans. Comput. Syst.} \bibinfo{volume}{23},
  \bibinfo{pages}{219--252}.
\newblock \URLprefix \url{http://doi.acm.org/10.1145/1082469.1082470},
  \DOIprefix\doi{10.1145/1082469.1082470}.
%Type = Inproceedings
\bibitem[{Jin et~al.(2003)Jin, Qian, Sha, Yu and Zhou}]{Jin03}
\bibinfo{author}{Jin, C.}, \bibinfo{author}{Qian, W.}, \bibinfo{author}{Sha,
  C.}, \bibinfo{author}{Yu, J.X.}, \bibinfo{author}{Zhou, A.},
  \bibinfo{year}{2003}.
\newblock \bibinfo{title}{Dynamically maintaining frequent items over a data
  stream}, in: \bibinfo{booktitle}{Proceedings of the Twelfth International
  Conference on Information and Knowledge Management},
  \bibinfo{publisher}{ACM}, \bibinfo{address}{New York, NY, USA}. pp.
  \bibinfo{pages}{287--294}.
\newblock \URLprefix \url{http://doi.acm.org/10.1145/956863.956918},
  \DOIprefix\doi{10.1145/956863.956918}.
%Type = Article
\bibitem[{Karp et~al.(2003)Karp, Shenker and Papadimitriou}]{Karp}
\bibinfo{author}{Karp, R.M.}, \bibinfo{author}{Shenker, S.},
  \bibinfo{author}{Papadimitriou, C.H.}, \bibinfo{year}{2003}.
\newblock \bibinfo{title}{A simple algorithm for finding frequent elements in
  streams and bags}.
\newblock \bibinfo{journal}{ACM Trans. Database Syst.} \bibinfo{volume}{28},
  \bibinfo{pages}{51--55}.
\newblock \DOIprefix\doi{http://doi.acm.org/10.1145/762471.762473}.
%Type = Inproceedings
\bibitem[{Keralapura et~al.(2006)Keralapura, Cormode and
  Ramamirtham}]{Keralapura:2006}
\bibinfo{author}{Keralapura, R.}, \bibinfo{author}{Cormode, G.},
  \bibinfo{author}{Ramamirtham, J.}, \bibinfo{year}{2006}.
\newblock \bibinfo{title}{Communication-efficient distributed monitoring of
  thresholded counts}, in: \bibinfo{booktitle}{Proceedings of the 2006 ACM
  SIGMOD International Conference on Management of Data},
  \bibinfo{publisher}{ACM}, \bibinfo{address}{New York, NY, USA}. pp.
  \bibinfo{pages}{289--300}.
\newblock \URLprefix \url{http://doi.acm.org/10.1145/1142473.1142507},
  \DOIprefix\doi{10.1145/1142473.1142507}.
%Type = Article
\bibitem[{Lahiri et~al.(2016)Lahiri, Mukherjee and Tirthapura}]{Lahiri2016}
\bibinfo{author}{Lahiri, B.}, \bibinfo{author}{Mukherjee, A.P.},
  \bibinfo{author}{Tirthapura, S.}, \bibinfo{year}{2016}.
\newblock \bibinfo{title}{Identifying correlated heavy-hitters in a
  two-dimensional data stream}.
\newblock \bibinfo{journal}{Data Mining and Knowledge Discovery}
  \bibinfo{volume}{30}, \bibinfo{pages}{797--818}.
\newblock \URLprefix \url{http://dx.doi.org/10.1007/s10618-015-0438-6},
  \DOIprefix\doi{10.1007/s10618-015-0438-6}.
%Type = Article
\bibitem[{Lahiri and Tirthapura(2010)}]{LAHIRI20101241}
\bibinfo{author}{Lahiri, B.}, \bibinfo{author}{Tirthapura, S.},
  \bibinfo{year}{2010}.
\newblock \bibinfo{title}{Identifying frequent items in a network using
  gossip}.
\newblock \bibinfo{journal}{Journal of Parallel and Distributed Computing}
  \bibinfo{volume}{70}, \bibinfo{pages}{1241 -- 1253}.
\newblock \URLprefix
  \url{http://www.sciencedirect.com/science/article/pii/S0743731510001358},
  \DOIprefix\doi{https://doi.org/10.1016/j.jpdc.2010.07.006}.
%Type = Inproceedings
\bibitem[{Manjhi et~al.(2005)Manjhi, Shkapenyuk, Dhamdhere and
  Olston}]{recent-freq-items}
\bibinfo{author}{Manjhi, A.}, \bibinfo{author}{Shkapenyuk, V.},
  \bibinfo{author}{Dhamdhere, K.}, \bibinfo{author}{Olston, C.},
  \bibinfo{year}{2005}.
\newblock \bibinfo{title}{Finding (recently) frequent items in distributed data
  streams}, in: \bibinfo{booktitle}{Proceedings of the 21st International
  Conference on Data Engineering}, \bibinfo{publisher}{IEEE Computer Society},
  \bibinfo{address}{Washington, DC, USA}. pp. \bibinfo{pages}{767--778}.
\newblock \URLprefix \url{http://dx.doi.org/10.1109/ICDE.2005.68},
  \DOIprefix\doi{10.1109/ICDE.2005.68}.
%Type = Inproceedings
\bibitem[{Manku and Motwani(2002)}]{Manku02approximatefrequency}
\bibinfo{author}{Manku, G.S.}, \bibinfo{author}{Motwani, R.},
  \bibinfo{year}{2002}.
\newblock \bibinfo{title}{Approximate frequency counts over data streams}, in:
  \bibinfo{booktitle}{Proceedings of the 28th International Conference on Very
  Large Data Bases}, \bibinfo{publisher}{VLDB Endowment}. pp.
  \bibinfo{pages}{346--357}.
\newblock \URLprefix \url{http://dl.acm.org/citation.cfm?id=1287369.1287400}.
%Type = Article
\bibitem[{Metwally et~al.(2006)Metwally, Agrawal and Abbadi}]{Metwally2006}
\bibinfo{author}{Metwally, A.}, \bibinfo{author}{Agrawal, D.},
  \bibinfo{author}{Abbadi, A.E.}, \bibinfo{year}{2006}.
\newblock \bibinfo{title}{An integrated efficient solution for computing
  frequent and top-k elements in data streams}.
\newblock \bibinfo{journal}{ACM Trans. Database Syst.} \bibinfo{volume}{31},
  \bibinfo{pages}{1095--1133}.
\newblock \DOIprefix\doi{10.1145/1166074.1166084}.
%Type = Article
\bibitem[{Misra and Gries(1982)}]{Misra82}
\bibinfo{author}{Misra, J.}, \bibinfo{author}{Gries, D.}, \bibinfo{year}{1982}.
\newblock \bibinfo{title}{Finding repeated elements}.
\newblock \bibinfo{journal}{Science of Computer Programming}
  \bibinfo{volume}{2}, \bibinfo{pages}{143--152}.
%Type = Article
\bibitem[{Muthukrishnan(2005)}]{TCS-002}
\bibinfo{author}{Muthukrishnan, S.}, \bibinfo{year}{2005}.
\newblock \bibinfo{title}{Data streams: Algorithms and applications}.
\newblock \bibinfo{journal}{Foundations and Trends® in Theoretical Computer
  Science} \bibinfo{volume}{1}, \bibinfo{pages}{117--236}.
\newblock \DOIprefix\doi{10.1561/0400000002}.
%Type = Inproceedings
\bibitem[{Pulimeno et~al.(2018)Pulimeno, Epicoco, Cafaro, Melle and
  Aloisio}]{Pulimeno:WPDM2018}
\bibinfo{author}{Pulimeno, M.}, \bibinfo{author}{Epicoco, I.},
  \bibinfo{author}{Cafaro, M.}, \bibinfo{author}{Melle, C.},
  \bibinfo{author}{Aloisio, G.}, \bibinfo{year}{2018}.
\newblock \bibinfo{title}{Parallel mining of correlated heavy hitters}, in:
  \bibinfo{editor}{Gervasi, O.}, \bibinfo{editor}{Murgante, B.},
  \bibinfo{editor}{Misra, S.}, \bibinfo{editor}{Stankova, E.},
  \bibinfo{editor}{Torre, C.M.}, \bibinfo{editor}{Rocha, A.M.A.},
  \bibinfo{editor}{Taniar, D.}, \bibinfo{editor}{Apduhan, B.O.},
  \bibinfo{editor}{Tarantino, E.}, \bibinfo{editor}{Ryu, Y.} (Eds.),
  \bibinfo{booktitle}{Computational Science and Its Applications -- ICCSA
  2018}, \bibinfo{publisher}{Springer International Publishing},
  \bibinfo{address}{Cham}. pp. \bibinfo{pages}{627--641}.
%Type = Inproceedings
\bibitem[{Roy et~al.(2012)Roy, Teubner and Alonso}]{Roy2012}
\bibinfo{author}{Roy, P.}, \bibinfo{author}{Teubner, J.},
  \bibinfo{author}{Alonso, G.}, \bibinfo{year}{2012}.
\newblock \bibinfo{title}{Efficient frequent item counting in multi-core
  hardware}, in: \bibinfo{booktitle}{Proceedings of the 18th ACM SIGKDD
  International Conference on Knowledge Discovery and Data Mining},
  \bibinfo{publisher}{ACM}. pp. \bibinfo{pages}{1451--1459}.
\newblock \DOIprefix\doi{10.1145/2339530.2339757}.
%Type = Article
\bibitem[{Sacha and Montresor(2013)}]{Sacha}
\bibinfo{author}{Sacha, J.}, \bibinfo{author}{Montresor, A.},
  \bibinfo{year}{2013}.
\newblock \bibinfo{title}{Identifying frequent items in distributed data sets}.
\newblock \bibinfo{journal}{Computing} \bibinfo{volume}{95},
  \bibinfo{pages}{289--307}.
\newblock \DOIprefix\doi{10.1007/s00607-012-0220-1}.
%Type = Inproceedings
\bibitem[{Tangwongsan et~al.(2014)Tangwongsan, Tirthapura and
  Wu}]{Tangwongsan2014}
\bibinfo{author}{Tangwongsan, K.}, \bibinfo{author}{Tirthapura, S.},
  \bibinfo{author}{Wu, K.L.}, \bibinfo{year}{2014}.
\newblock \bibinfo{title}{Parallel streaming frequency-based aggregates}, in:
  \bibinfo{booktitle}{Proceedings of the 26th ACM Symposium on Parallelism in
  Algorithms and Architectures}, \bibinfo{publisher}{ACM}. pp.
  \bibinfo{pages}{236--245}.
\newblock \DOIprefix\doi{10.1145/2612669.2612695}.
%Type = Inproceedings
\bibitem[{Venkataraman et~al.(2005)Venkataraman, Xiaodong~Song, B.~Gibbons and
  Blum}]{Venkataraman}
\bibinfo{author}{Venkataraman, S.}, \bibinfo{author}{Xiaodong~Song, D.},
  \bibinfo{author}{B.~Gibbons, P.}, \bibinfo{author}{Blum, A.},
  \bibinfo{year}{2005}.
\newblock \bibinfo{title}{New streaming algorithms for fast detection of
  superspreaders}, in: \bibinfo{booktitle}{Proceedings of the Network and
  Distributed System Security Symposium, NDSS}.
%Type = Article
\bibitem[{Wu et~al.(2017)Wu, Lin, U, Gao and Lu}]{FSSQ}
\bibinfo{author}{Wu, S.}, \bibinfo{author}{Lin, H.}, \bibinfo{author}{U, L.H.},
  \bibinfo{author}{Gao, Y.}, \bibinfo{author}{Lu, D.}, \bibinfo{year}{2017}.
\newblock \bibinfo{title}{Novel structures for counting frequent items in time
  decayed streams}.
\newblock \bibinfo{journal}{World Wide Web} \bibinfo{volume}{20},
  \bibinfo{pages}{1111--1133}.
\newblock \URLprefix \url{http://dx.doi.org/10.1007/s11280-017-0433-5},
  \DOIprefix\doi{10.1007/s11280-017-0433-5}.
%Type = Inproceedings
\bibitem[{Yao et~al.(2006)Yao, Leonard, Wang and Loguinov}]{Yao2006}
\bibinfo{author}{Yao, Z.}, \bibinfo{author}{Leonard, D.},
  \bibinfo{author}{Wang, X.}, \bibinfo{author}{Loguinov, D.},
  \bibinfo{year}{2006}.
\newblock \bibinfo{title}{Modeling heterogeneous user churn and local
  resilience of unstructured p2p networks}, in: \bibinfo{booktitle}{Proceedings
  of the 2006 IEEE International Conference on Network Protocols}, pp.
  \bibinfo{pages}{32--41}.
\newblock \DOIprefix\doi{10.1109/ICNP.2006.320196}.
%Type = Inproceedings
\bibitem[{Zhang(2012)}]{Zhang2012}
\bibinfo{author}{Zhang, Y.}, \bibinfo{year}{2012}.
\newblock \bibinfo{title}{Parallelizing the weighted lossy counting algorithm
  in high-speed network monitoring}, in: \bibinfo{booktitle}{Second
  International Conference on Instrumentation, Measurement, Computer,
  Communication and Control (IMCCC),}, pp. \bibinfo{pages}{757--761}.
\newblock \DOIprefix\doi{10.1109/IMCCC.2012.183}.
%Type = Article
\bibitem[{Zhang et~al.(2014)Zhang, Sun, Zhang, Xu and Wu}]{Zhang2013}
\bibinfo{author}{Zhang, Y.}, \bibinfo{author}{Sun, Y.}, \bibinfo{author}{Zhang,
  J.}, \bibinfo{author}{Xu, J.}, \bibinfo{author}{Wu, Y.},
  \bibinfo{year}{2014}.
\newblock \bibinfo{title}{An efficient framework for parallel and continuous
  frequent item monitoring}.
\newblock \bibinfo{journal}{Concurrency and Computation: Practice and
  Experience} \bibinfo{volume}{26}, \bibinfo{pages}{2856--2879}.
\newblock \DOIprefix\doi{10.1002/cpe.3182}.
%Type = Inproceedings
\bibitem[{Zhao et~al.(2006)Zhao, Ogihara, Wang and Xu}]{Zhao:2006}
\bibinfo{author}{Zhao, Q.G.}, \bibinfo{author}{Ogihara, M.},
  \bibinfo{author}{Wang, H.}, \bibinfo{author}{Xu, J.J.}, \bibinfo{year}{2006}.
\newblock \bibinfo{title}{Finding global icebergs over distributed data sets},
  in: \bibinfo{booktitle}{Proceedings of the Twenty-fifth ACM
  SIGMOD-SIGACT-SIGART Symposium on Principles of Database Systems},
  \bibinfo{publisher}{ACM}, \bibinfo{address}{New York, NY, USA}. pp.
  \bibinfo{pages}{298--307}.
\newblock \URLprefix \url{http://doi.acm.org/10.1145/1142351.1142394},
  \DOIprefix\doi{10.1145/1142351.1142394}.
%Type = Article
\bibitem[{Çem and Öznur Özkasap(2013)}]{CEM20131544}
\bibinfo{author}{Çem, E.}, \bibinfo{author}{Öznur Özkasap},
  \bibinfo{year}{2013}.
\newblock \bibinfo{title}{Profid: Practical frequent items discovery in
  peer-to-peer networks}.
\newblock \bibinfo{journal}{Future Generation Computer Systems}
  \bibinfo{volume}{29}, \bibinfo{pages}{1544 -- 1560}.
\newblock \URLprefix
  \url{http://www.sciencedirect.com/science/article/pii/S0167739X12001859},
  \DOIprefix\doi{https://doi.org/10.1016/j.future.2012.10.002}.
  \bibinfo{note}{including Special sections: High Performance Computing in the
  Cloud \& Resource Discovery Mechanisms for P2P Systems}.

\end{thebibliography}

% that's all folks
\end{document}